\documentclass[journal,onecolumn]{IEEEtran}

\usepackage{graphicx} 
\usepackage{booktabs} 
\usepackage{amsmath}
\usepackage{mathrsfs}
\usepackage{amsfonts}
\usepackage{amssymb}
\usepackage{epstopdf}
\usepackage{dblfloatfix}
\usepackage[caption=false,font=footnotesize]{subfig}
\usepackage{color}
\usepackage{enumitem}
\usepackage{amsthm}
\usepackage{algorithm,algorithmic}
\usepackage{multirow}
\usepackage{array}
\usepackage{cite}

\newtheorem{proposition}{Proposition}

\makeatletter
\newcommand*{\defeq}{\mathrel{\rlap{%
                     \raisebox{0.3ex}{$\m@th\cdot$}}%
                     \raisebox{-0.3ex}{$\m@th\cdot$}}%
                     =}
\newcommand{\thickhline}{%
    \noalign {\ifnum 0=`}\fi \hrule height 1pt
    \futurelet \reserved@a \@xhline
		}
		
\newcommand{\subscript}[2]{$#1 _ #2$}

\makeatother

\begin{document}

\title{Decentralized DC MicroGrid Monitoring and Optimization via Primary Control Perturbations
}


\author{Marko Angjelichinoski,~\IEEEmembership{Student Member,~IEEE,}
       Anna Scaglione,~\IEEEmembership{Fellow,~IEEE,}
       Petar Popovski,~\IEEEmembership{Fellow,~IEEE,}
       \v Cedomir Stefanovi\' c,~\IEEEmembership{Senior Member,~IEEE}
\thanks{M. Angjelichinoski, P. Popovski and \v C. Stefanovi\'c are with the Department of Electronic Systems, Aalborg University, Denmark (e-mail: $\left\{\mbox{maa,petarp,cs}\right\}$@es.aau.dk). A. Scaglione is with the School of Electrical, Computer and Energy Engineering, Arizona State University, AZ, USA (e-mail: Anna.Scaglione@asu.edu).}%
\thanks{The work presented in this paper was supported in part by EU, under grant agreement no. 607774 ``ADVANTAGE''.}
}


\maketitle


\begin{abstract}
We treat the emerging power systems with direct current (DC) MicroGrids, characterized with high penetration of power electronic converters. We rely on the power electronics to propose a decentralized solution for autonomous learning of and adaptation to the operating conditions of the DC Mirogrids; the goal is to eliminate the need to rely on an external communication system for such purpose. The solution works within the primary droop control loops and uses only local bus voltage measurements. Each controller is able to estimate (i) the generation capacities of power sources, (ii) the load demands, and (iii) the conductances of the distribution lines.
To define a well-conditioned estimation problem, we employ decentralized strategy where the primary droop controllers temporarily switch between operating points in a coordinated manner, following amplitude-modulated training sequences.
We study the use of the estimator in a decentralized solution of the Optimal Economic Dispatch problem. 
The evaluations confirm the usefulness of the proposed solution for autonomous MicroGrid operation.
\end{abstract}


\IEEEpeerreviewmaketitle

\section{Introduction}
\label{sec:intro}
\IEEEPARstart{S}{ince} {their inception, MicroGrids (MGs) have evolved substantially, particularly in the domain of low voltages (LV), leading to variety of use cases and topologies \cite{1,2,3,ref:new_last_1,2018_new1,2018_new2,2018_new5,2018_new6,2018_new7,2018_new8,2018_new9}: from small clusters of distributed energy resources (DERs) serving houses or buildings, to large meshes of small MGs covering large areas, such as neighborhoods, industrial complexes and remote villages.
As a result, the future smart grid (SG) is envisioned as a mesh of interconnected autonomous MG systems.
It is also within the field of MGs where \emph{direct current} (DC) power networks have experienced a renaissance 
due to 
the seamless integration with DC renewable generation, DC energy storage systems and DC smart loads \cite{2,3,ref:new_last_1}.
Hence, LV DC MGs are considered as a solution for residential and industrial use cases.} 

{A distinctive characteristic of DC MGs is the use of programmable DC/DC and AC/DC \emph{power electronic converters} (PECs) to connect the DERs to the DC distribution system.
PECs are digital signal processors (DSPs) that allow for software implementation of \emph{advanced control systems} \cite{2,3}. 
Leveraging on the advanced features of PECs 
the control system design also shifted from simple strategies, suitable for small systems \cite{14,15,16}, to modular hierarchical architectures where several interacting control layers dynamically respond to state variations on different time scales and pursue various complementary objectives \cite{3,ref:new_last_1,4,5,6,2018_new3,2018_new4,2018_new11,8,9,10}. 
Specifically, the MG control plane is organized into \emph{dual-layer} architecture, comprising \emph{primary} and \emph{upper} control layer \cite{3,4}.
The primary control is decentralized and deals with high frequency dynamic compensation and state regulation  \cite{3}. 
The upper control layer deals with slow, global changes in the MG by providing updated primary control references and is implemented in distributed/centralized fashion 
\cite{4,5,6,2018_new3,2018_new4,2018_new11,8,9,10}.
An exemplary upper layer application is the \emph{Optimal Economic Dispatch} (OED), which aims to compute the optimal dispatch policies that minimize the total generation cost while keeping the load balanced \cite{6}.
}

{The standard design assumption is that the feedback of the upper control layer is closed via an \emph{external communication system}, usually via off-the-shelf wireless technologies \cite{3,6,8}.
However, this approach was challenged recently due to several issues \cite{2,3}.
First, the distributed power systems, particularly MGs, are dynamic and ad-hoc in nature, thus the installation of communication hardware may prove impractical and cost inefficient.
Second, the external communication system reduces the resilience of the overall MG system, as it becomes a factor in the system reliability/availability. 
Finally, there is a growing concern about the cyber-security of power systems that exploit external communications, as the related security threats and attacks might severely compromise their stability and operation, 
leading to blackouts, equipment damage, data theft and investment losses \cite{12,13,ref:new_last_3,ref:new_last_4}.
}

{A straightforward solution would be to remove the upper layer completely and run the DC MG only with primary control without any further coordination.
However, the approach is not suitable for advanced MG topologies, 
as it can not foster optimal and sustainable regulation. 
The 
\emph{DC bus signaling} has been introduced as an enhancement of the above idea \cite{14,15,16}.
It uses the variations of the steady state bus voltage as an implicit coordination signal that tells the DERs how to behave in specific conditions.
The idea is motivated by the fact that DC systems are inherently tolerant to steady state voltage variations, allowing for voltage ripples of up to $10\%$ \cite{2,3,2018_new1}. 
Each PEC monitors the local voltage and if it the crosses predefined threshold, the PEC takes predefined actions.
This approach has reliability, availability and security advantages over traditional networked design and requires only software modifications of the PECs.
However, it is configuration-dependent, performing well in environments with predictable loads, but not in large, dynamic and general-purpose MGs.
Moreover, the range of upper layer applications that can be supported is limited. 
Another alternative to wireless communications is to use conventional powerline communications (PLC)~\cite{17}.
This way, some of the security concerns can be alleviated as now an attacker would need physical access to the MG.
Nevertheless, PLCs are still essentially an external communication system coupled to the control of the MG, as they require installation of dedicated modems.
}

{Motivated by the shortcomings of the above approaches, we propose a \emph{decentralized} dual-layer control architecture for \emph{autonomous} DC MGs in which each primary controller \emph{locally} acquires the information required for the operation of the upper layer and determines the updated primary control references \emph{without} the support of external communication enabler.
To support the majority of applications, the upper control layer \emph{requires} information about: i) the generation capacities of the dispatchable DERs, ii) the demands of the loads, and iii) the conductance matrix of the distribution network \cite{6,2018_new4}.
This information can be \emph{inferred} from local voltage observations, since the bus voltages are \emph{functionally related} to the MG parameters 
through a non-linear model.
To extract these parameters, the PECs deliberately move the MG through a sequence of sub-optimal states via coordinated and amplitude-modulated perturbations of the primary control parameters, referred to as \emph{training sequences}.
This way, the PECs obtain sequences of local bus voltage measurements from which the required information can be uniquely estimated, provided that the training sequences satisfy \emph{sufficiency} criteria.
To this end, we formulate a constrained Maximum Likelihood (ML) estimation problem that estimates the MG parameters \emph{jointly} with the state of the DC MG.
To solve the non-convex optimization problem, we develop an iterative algorithm and compare its performance against the Cramer-Rao Lower Bound (CRLB).
We illustrate the practical potential of the method by applying it in \emph{decentralized OED} (DOED) and we show how to minimize the operational cost by optimizing the design of the training sequences.
The proposed solution does not rely on any additional communication hardware, as it exploits the signal processing capabilities of the PECs and its locally available voltage measurements, such that it can be implemented \emph{only} in software.
}

The rest of the paper is organized as follows.
Section~\ref{sec:example} gives an overview of the main contributions. 
Section~\ref{sec:sys_model} introduces the system model.
Section~\ref{sec:problem} presents the training protocol and formulates the decentralized system identification problem.
Section~\ref{sec:main1} is the pivotal section of the paper, presenting our take to the problem formulated in Section~\ref{sec:problem}.
Section~\ref{sec:DOED} introduces the periodic DOED protocol.
Section~\ref{sec:Results} presents the results and Section~\ref{sec:conc} concludes the paper.

\emph{Notation:} 
Column vectors and matrices are denoted by lowercase and uppercase bold letters, e.g., $\mathbf{a}\in\mathbb{R}^{N\times 1}$ and $\mathbf{A}\in\mathbb{R}^{N\times M}$.
$\mathbf{a}_{-n}\in\mathbb{R}^{(N-1)\times 1}$ is obtained from $\mathbf{a}$ by removing the element at position $n$.
Similarly, $\mathbf{A}_{-m}\in\mathbb{R}^{N\times (M-1)}$ is obtained from $\mathbf{A}$ by removing the $m$-th column $\mathbf{a}_m$. 
$(\cdot)^{\mathsf{T}}$, $(\cdot)^{\dagger}$, $\mathsf{vec}(\cdot)$, $\mathsf{dim}(\cdot)$, $\mathsf{rank}(\cdot)$, $\mathsf{trace}(\cdot)$ and $\|\cdot\|_l$ denote the transpose, the pseudo-inverse, the {vectorization}, the {dimension}, the {rank}, the trace and the $l$-norm of the argument.
$\otimes$ denotes the Kroneker product while $\odot$ and $\oslash$ denote the Hadamard (element-vise) product and division of vectors/matrices of adequate dimensions. 
The vectors $\mathbf{1}_N$, $\mathbf{0}_N$ and $\mathbf{e}_{n},~n\in\mathcal{N}$, denote the all-one, all-zero and the principal coordinate vector, $\mathbf{1}_{N\times M}$, $\mathbf{0}_{N\times M}$ denote the $N\times M$ all-one and all-zero matrices, and $\mathbf{I}_{N}$ is the $N\times N$ identity matrix.
$\mathsf{D}(\mathbf{a})$ denotes diagonal matrix with the entries of $\mathbf{a}$ on the main diagonal. 
We frequently use the identity $\mathsf{vec}(\mathsf{D}(\mathbf{a}))=\mathbf{O}_N\mathbf{a}$ where the $N^2\times N$ matrix $\mathbf{O}_N = \sum_{n=1}^N\mathbf{e}_{n}\otimes(\mathbf{e}_{n}\mathbf{e}_{n}^{\mathsf{T}})$.


\section{Overview of Contributions}
\label{sec:example}

{The proposed solution is illustrated in Fig.~\ref{PropSol}. 
We consider a generic DC MG model with multiple buses, described in Sections~\ref{sec:sys_model} and~\ref{sec:problem}. 
We assume that the MG does not have access to reliable external communication resources.
The physical \emph{state} of DC MGs is characterized by the steady state \emph{bus voltages}.
We introduce a \emph{parameter vector} $\boldsymbol{\theta}$ that collects all system variables whose values are determined by \emph{exogenous} influences; this includes the generation capacities of the DERs, the load demands and the distribution network topology, i.e., the conductance matrix, see Section~\ref{sec:param}.
Using the power balance equation, we represent the bus voltages thorough a \emph{non-linear} and \emph{implicit} model, parametrized by $\boldsymbol{\theta}$, see Section~\ref{sec:steady_state}.
Evidently, $\boldsymbol{\theta}$ \emph{varies} with time; to respond to its variations on different time scales, the DC MG is governed by a \emph{hierarchical control system}, comprising \emph{primary} and \emph{upper} control layer. 
The primary control is \emph{decentralized}: several controllers regulate the bus voltages, using only local feedbacks without exchanging any information with peer controllers.
They are very fast and capable of responding to high frequency variations in $\boldsymbol{\theta}$. 
Popular primary controller in DC MGs is the \emph{Voltage Source Converter (VSC)} with \emph{voltage droop control}, which is reminiscent to the widespread \emph{frequency droop control} in AC systems, but defined over the DC voltage; it is therefore standard practice to refer to it simply as \emph{droop controller}
\cite{2,3}.
The upper control layer, on the other hand, responds to less frequent changes in $\boldsymbol{\theta}$ that affect the global behavior of the system; examples include changes of the load/generation profile, faults, attacks, etc.
Its main role is to \emph{adapt} the system to the new conditions by computing \emph{updated optimal control references} for the primary controllers; \emph{all} upper layer control applications require full/partial knowledge of $\boldsymbol{\theta}$ to determine the control references that adequately reflect the new conditions \cite{4,5,6,2018_new3,2018_new4,2018_new11,8,9,10}.
}

{Unlike conventional centralized networked control solutions, where the upper control layer is supported by an external communication enabler, 
we propose a \emph{decentralized} control architecture that relies solely on the DSP capabilities of the PECs: namely, in our solution the upper control layer is implemented \emph{locally} within each PEC, and uses only the locally available state measurements, as depicted in Fig.~\ref{PropSol}.
The solution comprises two main functional blocks, i.e., the \emph{monitoring} and \emph{optimization}, executed \emph{sequentially}.}

\begin{figure}[t]
\centering
\includegraphics[scale=0.35]{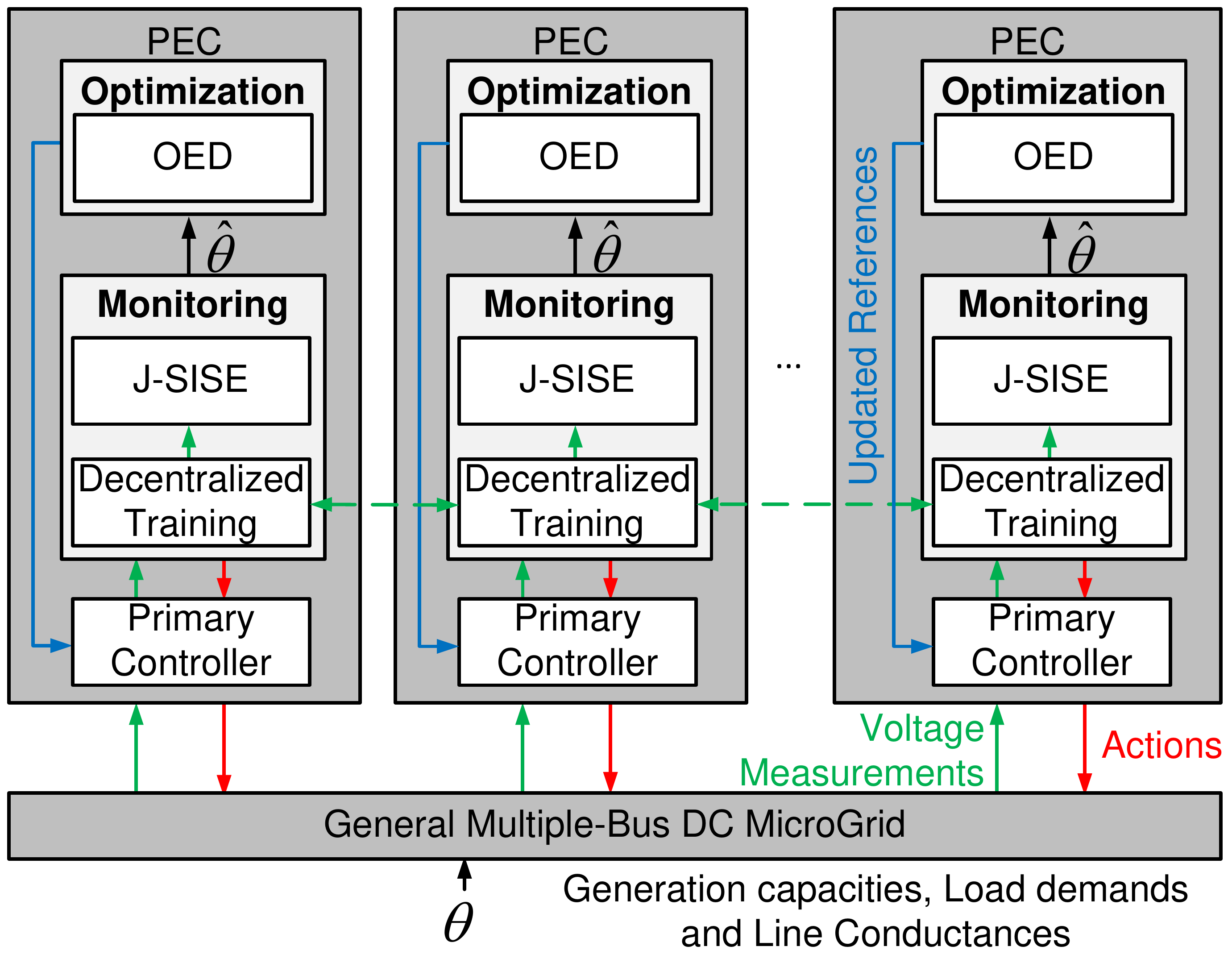}
\centering
\caption{{Overview of the proposed decentralized dual-layer control architecture.}}
\label{PropSol}
\end{figure}

{\textit{Monitoring}.
This functional block 
exploits the fact that the steady state bus voltages are functionally related with $\boldsymbol{\theta}$ through the \emph{power balance} equation; hence, each controller can compute a local estimate of $\boldsymbol{\theta}$. 
The key challenge 
is that it is impossible to infer $\boldsymbol{\theta}$ by using only local measurements of a single realization of the state, as the system is not observable and the estimation is ill-conditioned.
To address this, the monitoring block comprises two procedures: (1) coordinated \emph{decentralized training} \cite{25,27} via \emph{primary control perturbations}, see Section~\ref{sec:problem}, and (2) \emph{Joint System Identification and State Estimation} (J-SISE), see Section~\ref{sec:main1}.
During training, the controllers \emph{perturb} the values of the local droop control parameters, for a limited period of time, 
following predetermined \emph{training sequences}.
This generates a sequence of different realizations of the state.
The controllers collect the local measurements of the state sequence and \emph{modulate} them into the perturbation signals, see Section~\ref{sec:phases_subphases}.
In other words, the relation between the primary control perturbation signals and the induced state deviations is interpreted as the input-output relation of an \emph{implicit communication channel} \cite{18,19,20,21,last}, through which the controllers \emph{exchange} their local observations.
Hence, the training sequences are used both for generating multiple states \emph{and} communicating the local state observations. 
If the training sequences satisfy sufficiency criteria, see Section~\ref{sec:suff_exct}, each controller is able to compute unique estimate $\hat{\boldsymbol{\theta}}$ using the steady state voltage measurements acquired during training and the J-SISE algorithm, see Section~\ref{sec:JMLE}. 
The J-SISE is formulated as non-convex, constrained ML optimization problem in classical estimation framework which we solve via iterative algorithm based on partially linearized constraints and evaluate its performance using the CRLB, see Sections~\ref{sec:CRLB} and~\ref{sec:jmle_perf}.}

{\textit{Optimization}.
The local estimates $\hat{\boldsymbol{\theta}}$ are used as inputs to an \emph{energy management application} which computes \emph{updated} primary control references, see Fig.~\ref{PropSol}. 
Any application for which $\boldsymbol{\theta}$ is sufficient can be applied. 
We focus on DOED with linear generation cost model, since a simple, decentralized closed form solution is available in this case \cite{6,last}.
To this end, we design \emph{periodic protocol}, detailed in Section~\ref{sec:DOED}, where the controllers first perform training and obtain $\hat{\boldsymbol{\theta}}$ via J-SISE, then re-dispatch.
Finally, we show how to minimize the \emph{operational cost} of the protocol by calibrating the training parameters, see Section~\ref{sec:cost_tradeoff}.}

{We conclude by highlighting the benefits of the proposed solution.
First and foremost, it promotes the principle of \emph{self-sustainability} in SG as it reuses the DSP features of the available power electronics and obviates critical reliance on external communication system. 
Further, the optimization block is not limited only to OED,
as the knowledge of $\boldsymbol{\theta}$ allows each controller to solve locally 
a great deal of energy management optimizations (even if they do not have decentralized formulation) such as \emph{Optimal Power Flow} (OPF), \emph{Unit Commitment} (UC) and security-related applications, such as \emph{Fault Detection and Diagnosis} (FDD) \cite{8,2018_new4}.
This flexibility strengthens the \emph{autonomous} operation of the DC MG. 
Finally, the developed framework 
can be adapted for \emph{arbitrary} DC MG systems, as discussed in Section~\ref{sec:jmle_perf}. 
}

\section{System Model}
\label{sec:sys_model}

{The terminology and the notation system applied to the model is standardly used in power engineering literature \cite{3,2018_new4}. 
Section~\ref{sec:problem} introduces compact, matrix notation of the power balance equation which is easier to manipulate later on; this can be also seen as a standalone contribution, as this is the first work that introduces such compact notation for droop-controlled DC MG.}



\subsection{General Multiple-Bus DC MicroGrid}
\label{sec:DCMGs}

\subsubsection{Buses and Distribution Network}
A DC MG is a collection of DERs and loads, connected to low voltage DC distribution system, see Fig.~\ref{MBMG}.
The distribution system consist of $N\geq 1$ buses, indexed in the set $\mathcal{N}=\left\{1,...,N\right\}$. 
Each bus $n$ in steady state is characterized by a bus voltage $v_n$, and all DERs and loads connected to bus $n$ measure the same voltage $v_n$.
The distribution line connecting buses $n$ and $m,~n\neq m$ has a \emph{line conductance} denoted by $y_{n,m},~y_{n,m}\equiv y_{m,n}\geq 0$ \cite{3}.
The topology of the distribution system is specified via the symmetric $N\times N$ \emph{conductance matrix} $\mathbf{Y}$ with elements:
\begin{align}\label{eq:Psi_def}
[\mathbf{Y}]_{n,m} & = \left\{
  \begin{array}{lr}
	  \sum_{j\in\mathcal{N}}y_{n,j}, & n=m,\\
     -y_{n,m}, & n\neq m,
  \end{array}
\right.
~n,m\in\mathcal{N}
\end{align}

\subsubsection{Distributed Energy Resources}
We model each DER as separate bus, i.e., we assume that each bus hosts at most one DER; hence, the total number of DERs is $N$ and they are indexed in the set $\mathcal{N}$.
This modeling choice simplifies the notation  without losing generality; in fact, if DERs $n$ and $m$ are connected to the same physical point, i.e., the same bus, by definition $y_{n,m}=\infty$.
The $n-$th DER has current $i_n$ and power output $p_n=v_ni_n$.
We assume that the DERs in the MG are small-scale power sources such as renewables (RESs) or distributed generators (DGs) based on traditional fossil fuel.
Each DER $n$ has an instantaneous \emph{generation capacity} $g_n\geq 0$, and the output power $p_n$ should satisfy $0\leq p_n\leq g_n$.

\subsubsection{Loads}
The $n-$th bus hosts a collection of loads, represented through an aggregate model as a mixture of three components (also known as ZIP load model \cite{28}): 1) \emph{constant conductance} $y_{n}^{\text{ca}}={x^{-2}}{d_{n}^{\text{ca}}}$, 2) \emph{constant current} $i_{n}^{\text{cc}}={x^{-1}}{d_{n}^{\text{cc}}}$, and 3) \emph{constant power} component $d_{n}^{\text{cp}}$, see Fig.~\ref{MBMG}.
The quantities $d_{n}^{\text{ca}}$, $d_{n}^{\text{cc}}$ and $d_{n}^{\text{cp}}$ are the instantaneous \emph{power demands} of the components at a rated voltage $x$.
For a given $d_{n}^{\text{cp}}$, the constant power component in steady state is approximated with an equivalent positive current source in parallel with negative conductance and 
the electrical parameters are \cite{3}:
\begin{equation}\label{eq:cpl}
	i_{n}^{\text{cp}} \approx \frac{2d_{n}^{\text{cp}}}{v_n},\;y_{n}^{\text{cp}} \approx -\frac{v_n^2}{d_{n}^{\text{cp}}},~n\in\mathcal{N}.
\end{equation}

\begin{figure}[t]
\centering
\includegraphics[scale=0.35]{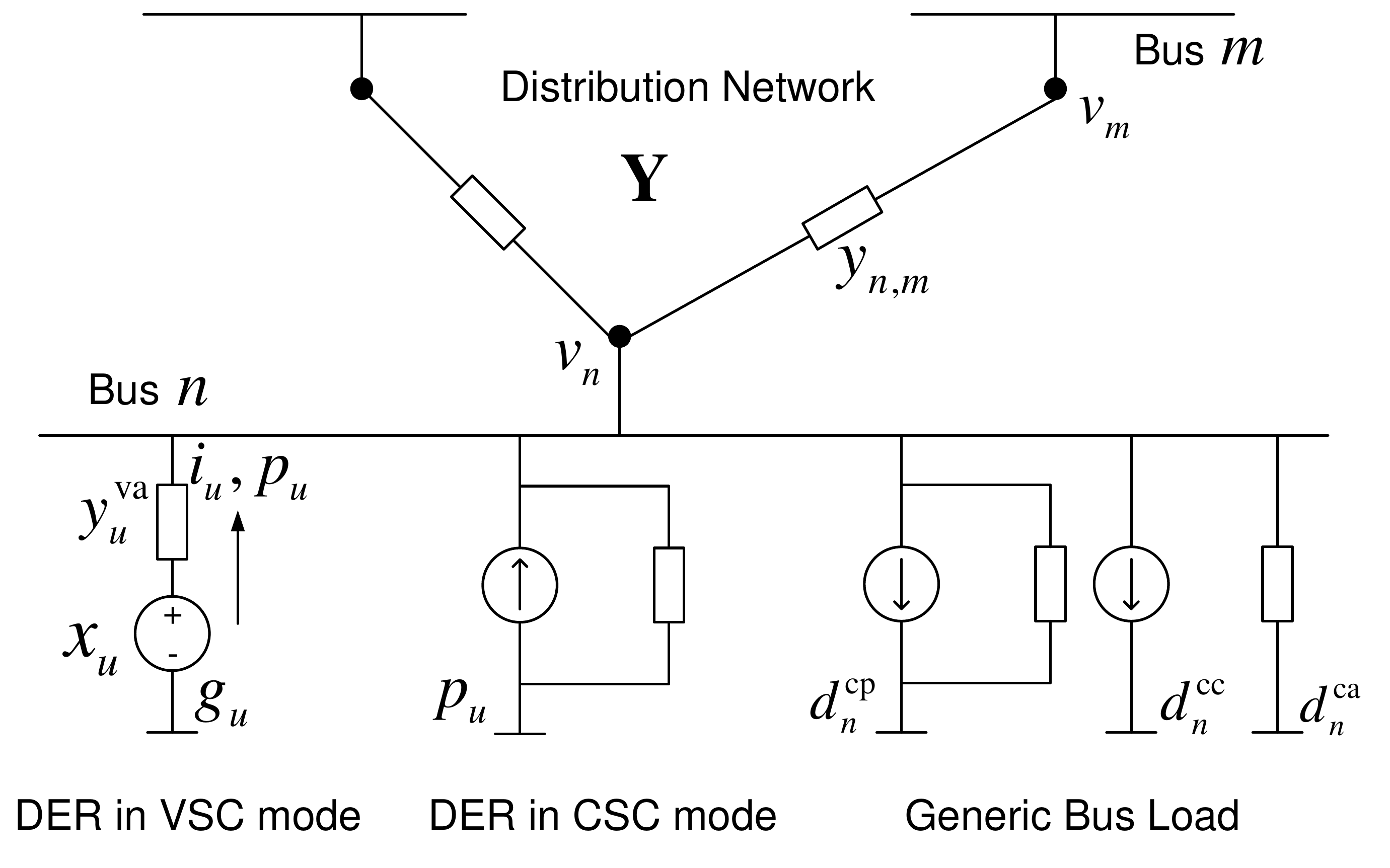}
\centering
\caption{System model of general multiple-bus DC MG in steady state.}
\label{MBMG}
\end{figure}

\subsubsection{Primary Control}
The DERs use PECs to interface the buses; the bus voltage $v_n$ and/or current $i_n$, i.e., power $p_n$ are locally controlled through decentralized primary controller, which is a software program executed by the PEC \cite{3}.
Two primary control schemes, i.e., \emph{modes} are commonly used, see Fig.~\ref{PrimaryControl}: 1) a closed loop \emph{Voltage Source Converter (VSC)}, and 2) an open loop \emph{Current Source Converter (CSC)}.
VSC regulates the bus voltage and current of the DER as the loads/generation in the system change in order to keep the bus voltage within predefined margins and foster fair power sharing.
It contains fast inner and slow outer control loops.
An inner control loop consists of a cascade of voltage and current loops with control bandwidth of the order of several tens of $\text{kHz}$, equal to the sampling frequency $\phi_S$ of the converter. Its role is to maintain the output bus voltage $v_n$ to specific reference value, dictated by the outer control loop.
The outer control loop is closed via filtered current feedback, and is slower than the inner control loop by an order of magnitude.
The current feedback generates the reference value for the inner voltage loop, via the following steady state control law:

\begin{equation}\label{eq:droop}
	v_n = x_{n} - (y_{n}^{\text{va}})^{-1}i_{n},~n\in\mathcal{N}.
\end{equation}
This is known as \emph{decentralized droop control} for DC MGs \cite{3,4} with 
two controllable parameters: 
the \emph{reference voltage} $x_{n}$ and the \emph{virtual conductance} $y_{n}^{\text{va}}$.
Their values are set (i) to keep the bus voltage, as closely as possible to the rated voltage $x$, within predefined margins $v_{\max}\leq v_n\leq v_{\min}$ for any $n\in\mathcal{N}$, and (ii) to enable fair power sharing among DERs based on their instantaneous generation capacities \cite{3}.
Fig.~\ref{VI} depicts a widespread droop control law that meets the above conditions, with droop control parameters set as follows:
\begin{align}\label{eq:droop_gen}
v_{\min} < x_n \leq v_{\max},\; y_n^{\text{va}} = \frac{g_n}{(x_n - \Delta v_n)\Delta v_n} \equiv s_n g_n,
\end{align}
where $s_n \equiv ((x_n - \Delta v_n)\Delta v_n)^{-1}$ is the \emph{droop slope} (in volts$^{-2}$).
The configuration $y_n^{\text{va}}=s_ng_n$ enables proportional power sharing among the DERs.
When the DER operates close to its capacity, the maximal voltage drop is $\Delta v_n$, $0 < \Delta v_n \leq x_n - v_{\min}$. 
In steady state, the droop-controlled VSC units are modeled as voltage sources in series with virtual conductance, see Fig.~\ref{MBMG}.

\begin{figure}[t]
\centering
\includegraphics[scale=0.35]{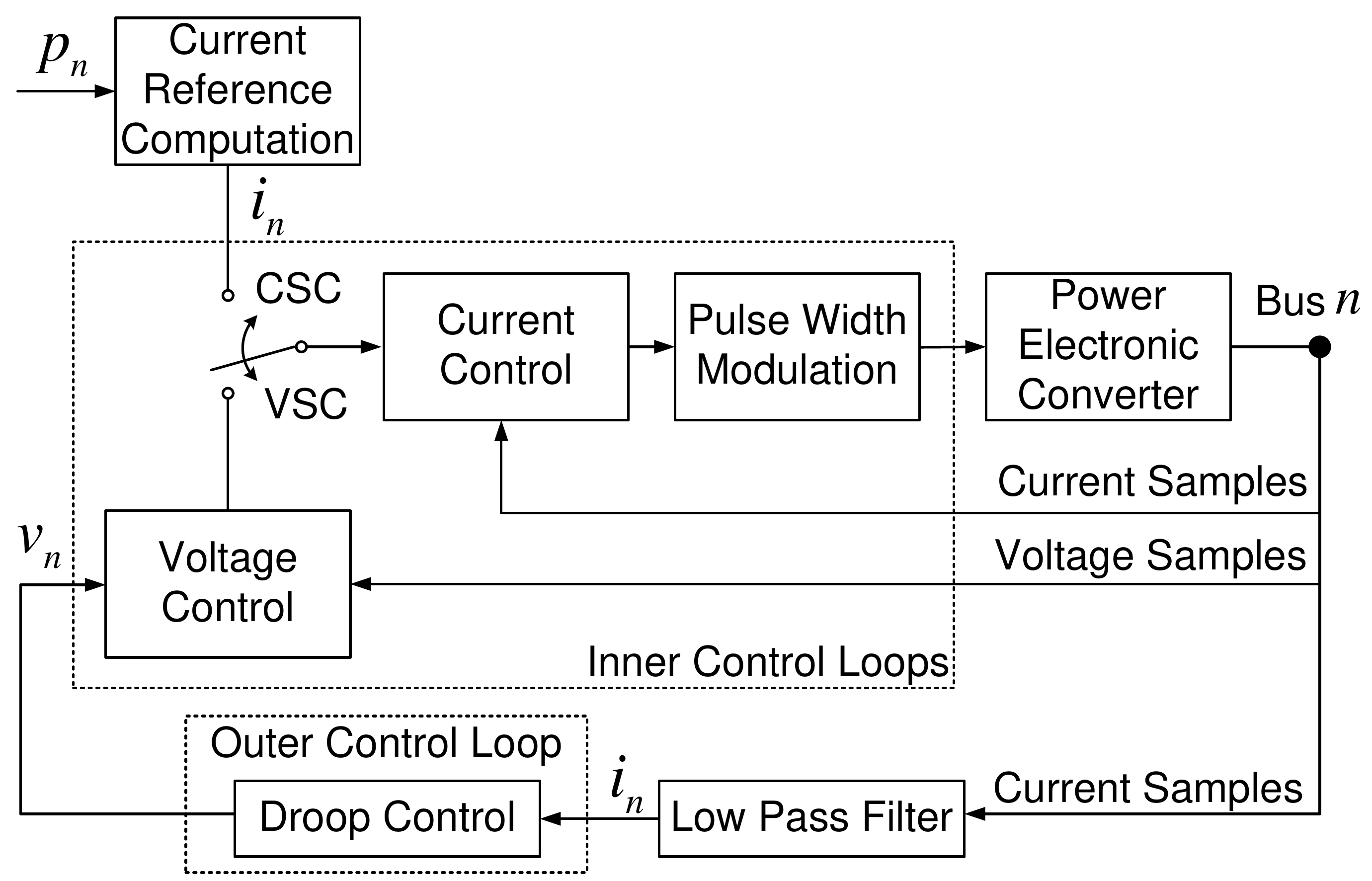}
\caption{Configuration of a primary controller: VSC and CSC modes.}
\label{PrimaryControl}
\end{figure}


\begin{figure}[t]
\centering
\includegraphics[scale=0.35]{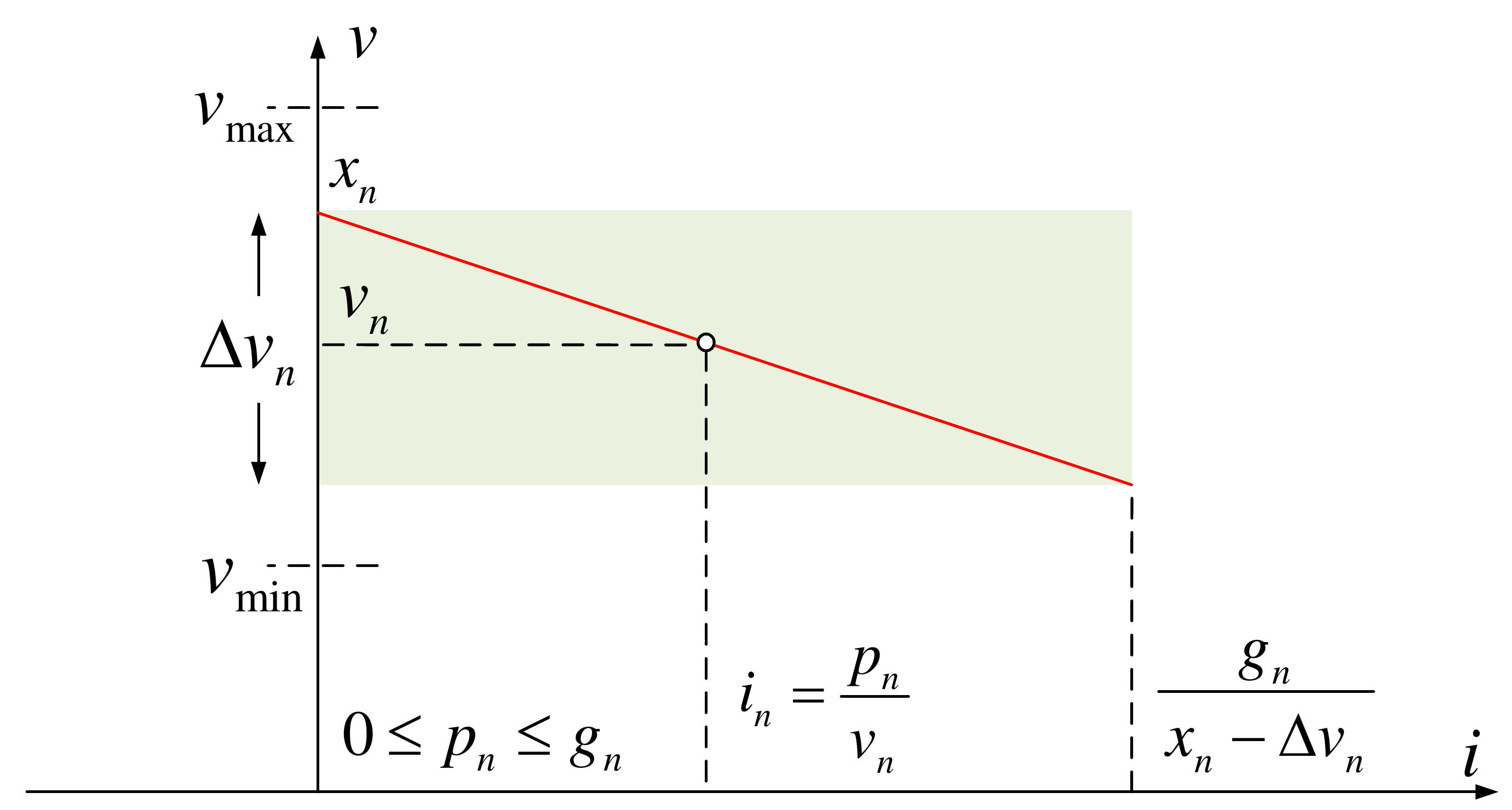}
\caption{$v-i$ diagram of droop-controlled DER in VSC mode. As the net load demand increases/decreases, controller $n\in\mathcal{N}$ responds by shifting the point $[v_n,i_n]$ down/up the (red) droop line.}
\label{VI}
\end{figure}

The other primary control mode CSC does not have outer control loop and inner voltage loops, see Fig.~\ref{PrimaryControl}. The reference for the inner current loop is generated via a separate algorithm that gets as an input fixed power reference \cite{3}.
Hence, a CSC acts as a constant power component, neither participating in voltage regulation nor power sharing.
It is modeled as a negative current source and parallel conductance, as in \eqref{eq:cpl} but with opposite sign. It is architecturally equivalent to a 
negative constant power load, see Fig.~\ref{MBMG}.

The subsets of DERs operating in VSC/CSC, denoted respectively by $\mathcal{N}^{\text{V}} /\mathcal{N}^{\text{C}}$, are determined 
\emph{dynamically} by the upper layer application, see Section~\ref{sec:DOED} for an example.
To support this dynamic operation, each converter is assumed to have \emph{dual mode}, and is capable to switch between VSC and CSC control mode seamlessly \cite{3,rev8}, see also Fig.~\ref{PrimaryControl}.

\subsection{Steady State Equations}
\label{sec:steady_state}

A DC MG is governed by Ohm's and Kirchhoff's laws, resulting in a system of $N$ steady state power balance equations for $N$ buses:
\begin{equation}\label{eq:power_balance_system}
\omega_n = 0,~n\in\mathcal{N},
\end{equation}  
with $\omega_n$ given by:
\begin{equation}\label{eq:MBMG_power_balance_compact}
\omega_n = v_n^2\left(\zeta_n y_n^{\text{va}} + \frac{1}{x^2}d_n^{\text{ca}} + \sum_{m\in\mathcal{N}}y_{n,m}\right) - v_n\sum_{m\in\mathcal{N}}v_my_{n,m} - v_n\left(\zeta_nx_ny_n^{\text{va}} - \frac{1}{x}d_n^{\text{cc}}\right) + d_n^{\text{cp}} - (1-\zeta_n)p_n. 
\end{equation}
The binary variable $\zeta_n$ in \eqref{eq:MBMG_power_balance_compact} is $1/0$ if DER $n$ is configured in VSC/CSC control mode, respectively. 
The system of equations is quadratic in the bus voltages, such that, in general, a closed form solution for $v_n,~n\in\mathcal{N}$ is not possible.
The non-linear nature of the power balance equations stems from the presence of constant power components \cite{29}, both constant power loads and CSCs.
Hence, in the case when $d_n^{\text{cp}}=0$ for all $n$ and $\mathcal{N}^{\text{C}}=\emptyset$, the system \eqref{eq:power_balance_system} becomes linear in the bus voltages.
Another special case with closed-form solution is the Single-Bus DC MG which we have studied separately \cite{21} due to its practical importance.

\section{Problem Formulation and Training Epoch}
\label{sec:problem}

The DC MG is \emph{not} connected to an external communication system and the PECs only have the local voltage/current measurements to work with.
To learn {(i) the generation capacities of remote DERs, (ii) the power demands of the loads and, (iii) the conductances of the distribution lines}, the controllers need to solve a \emph{decentralized system identification} problem, formulated below.

Before we begin, we list the main assumptions:
\begin{enumerate}[label=(\subscript{A}{\arabic*})]
\item The primary controllers are fully synchronized to a common time reference.
\item No prior knowledge on the generation capacities, load demands or the conductance matrix is used.
\item The rate of load/generation/topology variations is an order of magnitude smaller than the frequency of the primary controllers. 
\end{enumerate}



\subsection{Parameter Vector}
\label{sec:param}

Let $\mathbf{g} = [g_1,\hdots,g_N]^{\mathsf{T}}$ be a $N\times 1$ vector that collects the instantaneous generation capacities of all DERs in the MG.
Similarly, the instantaneous load demands are collected in separate $N\times 1$ vectors: $\mathbf{d}^{\text{ca}}=[d_1^{\text{ca}},\hdots,d_N^{\text{ca}}]^{\mathsf{T}}$, $\mathbf{d}^{\text{cc}}=[d_1^{\text{cc}},\hdots,d_N^{\text{cc}}]^{\mathsf{T}}$ and $\mathbf{d}^{\text{cp}}=[d_1^{\text{cp}},\hdots,d_N^{\text{cp}}]^{\mathsf{T}}$.
The $3N\times 1$ load demand vector is defined as $\mathbf{d}=[(\mathbf{d}^{\text{ca}})^{\mathsf{T}},\;(\mathbf{d}^{\text{cc}})^{\mathsf{T}},\;(\mathbf{d}^{\text{cp}})^{\mathsf{T}}]_{n\in\mathcal{N}}^{\mathsf{T}}$.
Further, we observe that $\mathbf{Y}$ is fully specified by its supra(infra)-diagonal elements, see \eqref{eq:Psi_def}.
We organize these elements in a vector $\boldsymbol{\psi}=[\hdots,y_{n,m},\hdots]^{\mathsf{T}}$, $n,m\in\mathcal{N}$, $m > n$, with dimension $\mathsf{dim}(\boldsymbol{\psi}) = \frac{1}{2}N(N-1)\times 1$.
Using $\boldsymbol{\psi}$, we can write $\mathbf{Y}$ as the weighted Laplacian $\mathbf{Y} = \mathbf{A}\mathsf{D}(\boldsymbol{\psi})\mathbf{A}^{\mathsf{T}}$, where $\mathbf{A}\in\left\{-1,0,1\right\}^{N\times\mathsf{dim}(\boldsymbol{\psi})}$ is the oriented incidence matrix \cite{rev9}.

The deterministic parameter vector $\boldsymbol{\theta}$ is defined as:
\begin{equation}\label{eq:param_k_general}
\boldsymbol{\theta} = [\mathbf{g}^{\mathsf{T}},\mathbf{d}^{\mathsf{T}},\boldsymbol{\psi}^{\mathsf{T}}]^{\mathsf{T}},
\end{equation}
with dimension $\mathsf{dim}(\boldsymbol{\theta}) = \frac{1}{2}N(N+7)\times 1$.
From the discussion in Section~\ref{sec:steady_state}, the steady state bus voltage $v_n$ depends on $\boldsymbol{\theta}$, see eq. \eqref{eq:droop_gen}, \eqref{eq:MBMG_power_balance_compact}.
This suggests that an arbitrary controller can infer the parameter vector $\boldsymbol{\theta}$ locally, using local measurements of the steady state bus voltage (see also \cite{ref:new_last_6} and references therein for similar approaches).
However, it is impossible to determine $\boldsymbol{\theta}$ uniquely in classical, non-Bayesian estimation framework, using only a single observation of the local steady state bus voltage.
To address this issue, the following subsection introduces a technique based on \emph{decentralized training} via \emph{primary control perturbations}. 

\begin{figure}[t]
\centering
\includegraphics[scale=0.35]{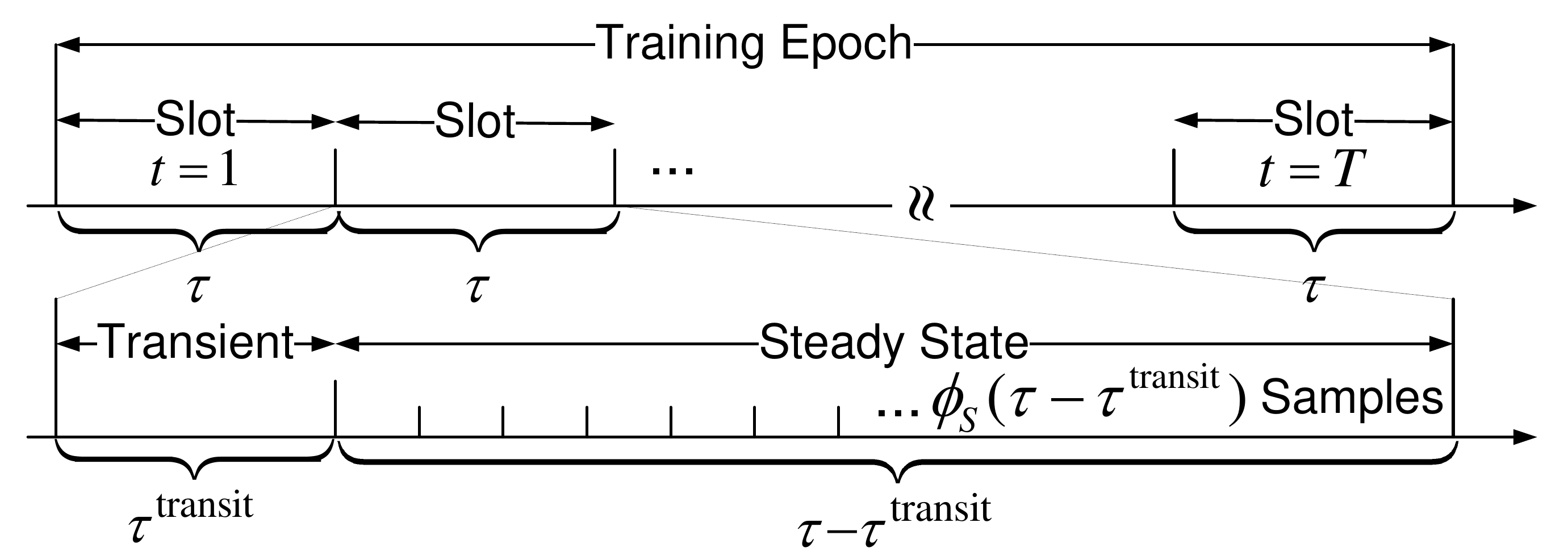}
\caption{Training epoch: time-slotted organization.}
\label{TimeAxis}
\end{figure}

\subsection{Training Protocol and Training Sequences}
\label{sec:train}

We introduce a dedicated \emph{training epoch} of predefined 
duration, in which (i) all controllers switch to VSC mode using a droop control law of the form \eqref{eq:droop_gen}, (ii) perturb their local droop control parameters, causing deviations of the bus voltages, and (iii) measure the local bus voltage response, collecting sequences of steady state bus voltage measurements.
The training epoch design uses the assumptions $(A_1)$ and $(A_2)$.
Specifically, the time axis during the training epoch is divided into $T$ \emph{time slots}, see Fig.~\ref{TimeAxis}, and \emph{all} controllers are synchronized to this structure.
We index each slot with $t\in\mathcal{T}=\left\{1,\hdots,T\right\}$.
The slot duration $\tau$ complies with the control bandwidth of the primary control loops, allowing the bus to reach a steady state after a {transient time} $\tau^{\text{transit}}\ll \tau$, yielding $\phi_S(\tau-\tau^{\text{transit}})$ voltage samples per slot for each controller, see Fig.~\ref{TimeAxis}.
The system constant $\tau^{\text{transit}}$, usually several milliseconds \cite{3},
is determined by the sampling frequency $\phi_S$ and the line capacitors.
Following $(A_2)$, $\boldsymbol{\theta}$ can be assumed to remain constant during the training epoch.

We use $\tilde{\cdot}$ to denote the unperturbed, i.e., \emph{nominal} droop control parameters during the training epoch; we use the law \eqref{eq:droop_gen} with equal reference voltages and droop slopes:
\begin{equation}\label{eq:nominal_droop}
\tilde{x}_n \equiv \tilde{x},~\Delta\tilde{v}_n \equiv \Delta\tilde{v},~\tilde{s}_n \equiv \tilde{s},~n\in\mathcal{N}. 
\end{equation}
In slot $t$, all controllers \emph{simultaneously} perturb the reference voltages and droop slopes, according to \emph{perturbation signals} $x_n(t)\neq \tilde{x}$, $s_n(t)\neq \tilde{s}$, $n\in\mathcal{N}$; they are organized in $T\times N$ \emph{training matrices} $\mathbf{X}$, $\mathbf{S}$, defined as $[\mathbf{X}]_{t,n}=x_n(t)$ and $[\mathbf{S}]_{t,n}=s_n(t)$, $n\in\mathcal{N}$, $t\in\mathcal{T}$.
The columns $\mathbf{x}_{n}$/$\mathbf{s}_{n}$ of $\mathbf{X}$/$\mathbf{S}$, correspond to the \emph{training sequence} injected by controller $n$.

\subsection{Steady State Bus Voltages and Measurement Vectors}
\label{sec:meas}
 
The steady state bus voltage $\tilde{v}_n$ corresponds to the nominal, unperturbed, droop parameters $\tilde{x}_n$, $\tilde{s}_n$.
The steady state bus voltage response in the $t-$th slot is $v_n(t)\neq \tilde{v}_n,~n\in\mathcal{N}$. The $T\times N$ steady state bus voltage matrix $\mathbf{V}$ is defined as $[\mathbf{V}]_{t,n}=v_n(t)$, $n\in\mathcal{N}$, $t\in\mathcal{T}$. The following proposition characterizes 
$\mathbf{V}$ in terms of $\mathbf{X}$, $\mathbf{S}$ and $\boldsymbol{\theta}$:

\begin{proposition}\label{prop2}
The steady state of DC MG during the training epoch is characterized by the implicit power balance equation:
\begin{align}\label{eq:g_cond_slot}
\mathbf{\Omega} & = \mathbf{0}_{T\times N},
\end{align}
where $\mathbf{\Omega}:[v_{\min},\;v_{\max}]^{T\times N}\times\mathbb{X}\times\mathbb{S}\times\mathbb{R}^{\mathsf{dim}(\boldsymbol{\theta})}\mapsto\mathbf{0}_{T\times N}$ is defined as $[\mathbf{\Omega}]_{t,n}=\omega_n(t),~n\in\mathcal{N},~t\in\mathcal{T}$, and given by:
\begin{align}\label{eq:PowerBalanceMatrix}
	\mathbf{\Omega} = \bigg( \mathbf{S}\mathsf{D}(\mathbf{g}) + \frac{1}{x^2}\mathbf{1}_{T}(\mathbf{d}^{\text{ca}})^{\mathsf{T}} \bigg)\odot\mathbf{V}\odot\mathbf{V} + (\mathbf{V}\mathbf{Y})\odot\mathbf{V} - \bigg((\mathbf{S}\odot\mathbf{X})\mathsf{D}(\mathbf{g}) - \frac{1}{x}\mathbf{1}_{T}(\mathbf{d}^{\text{cc}})^{\mathsf{T}}\bigg)\odot \mathbf{V} + \mathbf{1}_{T}(\mathbf{d}^{\text{cp}})^{\mathsf{T}}.
\end{align}
The subsets $\mathbb{X}\subset\mathbb{R}^{T\times N}$ and $\mathbb{S}\subset\mathbb{R}^{T\times N}$ comprise all training matrices $\mathbf{X}$ and $\mathbf{S}$ that keep $\mathbf{V}$ within $[v_{\min},\;v_{\max}]^{T\times N}$.
\end{proposition}
\begin{proof}
See Appendix~\ref{app:propI}.
\end{proof}

The power balance equation \eqref{eq:g_cond_slot} reflects the requirement to keep the system balanced and stable, i.e., in a valid (albeit suboptimal) operating point, in each slot during training.
It also gives an \emph{implicit} relation between $\mathbf{V}$ and $\boldsymbol{\theta}$, since \eqref{eq:g_cond_slot} cannot be solved in closed form for $\mathbf{V}$. 

The $n-$th controller measures the $n$-th column $\mathbf{v}_{n}$ of $\mathbf{V}$ during the training epoch.
The noisy measurement obtained by controller $n$ in slot $t$ is an average of multiple voltage samples collected during the steady state period of the slot, and can be written as $w_n(t) = v_n(t) + z_n(t)$ with $z_n(t)$ denoting the additive noise.
The $T\times N$ bus-voltage measurements matrix $\mathbf{W}$, with $[\mathbf{W}]_{t,n} = w_{n}(t)$, $n\in\mathcal{N}$, $t\in\mathcal{T}$, is given as: 
\begin{equation}\label{eq:noisy_measurement}
\mathbf{W} = \mathbf{V} + \mathbf{Z},
\end{equation}
where $\mathbf{Z}$ represents the noise and $\mathsf{vec}(\mathbf{Z})$ is a zero-mean, white Gaussian random vector with standard deviation $\sigma$ \cite{rev11}, such that the probability density function (pdf) of $\mathsf{vec}(\mathbf{W})$ is:
\begin{equation}\label{eq:pdf1}
 \rho(\mathsf{vec}(\mathbf{W});\boldsymbol{\theta}) = \mathsf{N}(\mathsf{vec}({\mathbf{V}}),\sigma^2\mathbf{I}_{NT}).
\end{equation}

The decentralized system identification problem for DC MGs is about devising an \emph{efficient} and \emph{unbiased} estimator of the local parameter vector $\boldsymbol{\theta}_{-n}$, denoted with $\hat{\boldsymbol{\theta}}_{-n}$, using only local bus voltage measurements $\mathbf{w}_{n}$, for any $n\in\mathcal{N}$.

\subsection{Relaxing Assumptions $(A_1)-(A_3)$}
\label{sec:relax}

We briefly discuss the implications that arise when assumptions $(A_1)-(A_3)$ are no longer valid; addressing these implications is out of the paper's scope.
We start with $(A_1)$, as the strongest assumption.
Maintaining precise synchronization among the controllers on the level of slot and training epoch can be easily achieved if the PECs are equipped with GPS modules.
Alternatively, one can use common decentralized network synchronization approaches, typically used in sensor networks \cite{rev4}. Since the method operates in a time scale in the order of milliseconds, it should be significantly easier to maintain (at least coarse) synchronization for long periods of time.
Finally, if synchronization is not possible, and the controllers inject perturbation signals without any prior coordination, then the formulation of the problem should be modified accordingly to account for asynchronous training.
For instance, the parameter vector should be extended to include binary variables that capture the activity patterns of the controllers and the start times of individual training sequences, as well as their end times in case of variable training sequence durations. 

Assumption $(A_2)$ simply casts our problem in classical estimation framework.
In practice, prior knowledge is always available to some extent; in fact, $\boldsymbol{\theta}$ can be assumed to evolve over time following a stochastic process, paving the way for formulating the identification problem in sophisticated Bayesian filtering/prediction framework \cite{32}.
Nevertheless, the analysis of the non-Bayesian case naturally comes first.  

We use assumption $(A_3)$ to postulate that $\boldsymbol{\theta}$ remains fixed during training, which is not true in general.
In practice, $\boldsymbol{\theta}$ might change at any time due to load/generation variation or a system fault.
To incorporate this notion we should reformulate the problem accordingly.
One way is to first relax assumption $(A_2)$ and model the dynamic evolution of $\boldsymbol{\theta}$ via stochastic process, where relaxing assumption $(A_3)$ arises naturally.
We can avoid relaxing $(A_2)$ and still use the classical framework as presented in the paper, but with modified definition of the parameter vector.
For instance, let us assume that $\boldsymbol{\theta}$ has changed no more than $J\geq 0$ times during training; then, the parameter vector should comprise $J+1$ different values for $\boldsymbol{\theta}$ as defined in \eqref{eq:param_k_general}, in addition to the time instances when the changes have occurred.
Such formulations in the literature are known as \emph{model change detection}, see \cite{rev3}. 

\section{Decentralized Generation, Demand and Topology Estimation}
\label{sec:main1}

\subsection{Preliminaries and Notation}

In the case when the controllers do not not have any knowledge of the steady state bus voltages at remote buses, the system is \emph{not observable}; hence $\boldsymbol{\theta}_{-n}$ cannot be uniquely identified in classical, non-Bayesian sense (see Appendix~\ref{app:nonident}). 

Motivated by the ideas in \cite{18}, we propose a decentralized solution that splits the slots into two consecutive \emph{training phases}: (i) \emph{measurement} phase, denoted as ${M}$-phase, and (ii) \emph{communication} phase, denoted as ${C}$-phase.
The slots in the ${C}$-phase are used to \emph{disseminate} the local steady state voltage measurements acquired in the ${M}$-phase to remote controllers via \emph{amplitude modulation} of the reference voltage perturbation signals.
Each controller then uses a sequential-type of demodulator to process the local bus voltage measurements acquired in the $C$-phase and acquire full knowledge of the portion of $\mathbf{W}$ that corresponds to the ${M}$-phase. 
If the training matrices in the ${M}$-phase satisfy predefined conditions, elaborated in subsection~\ref{sec:suff_exct}, then knowing only the ${M}$-phase portion of $\mathbf{W}$ is sufficient to uniquely estimate the parameter vector locally.

\begin{figure}[t]
\centering
\includegraphics[scale=0.35]{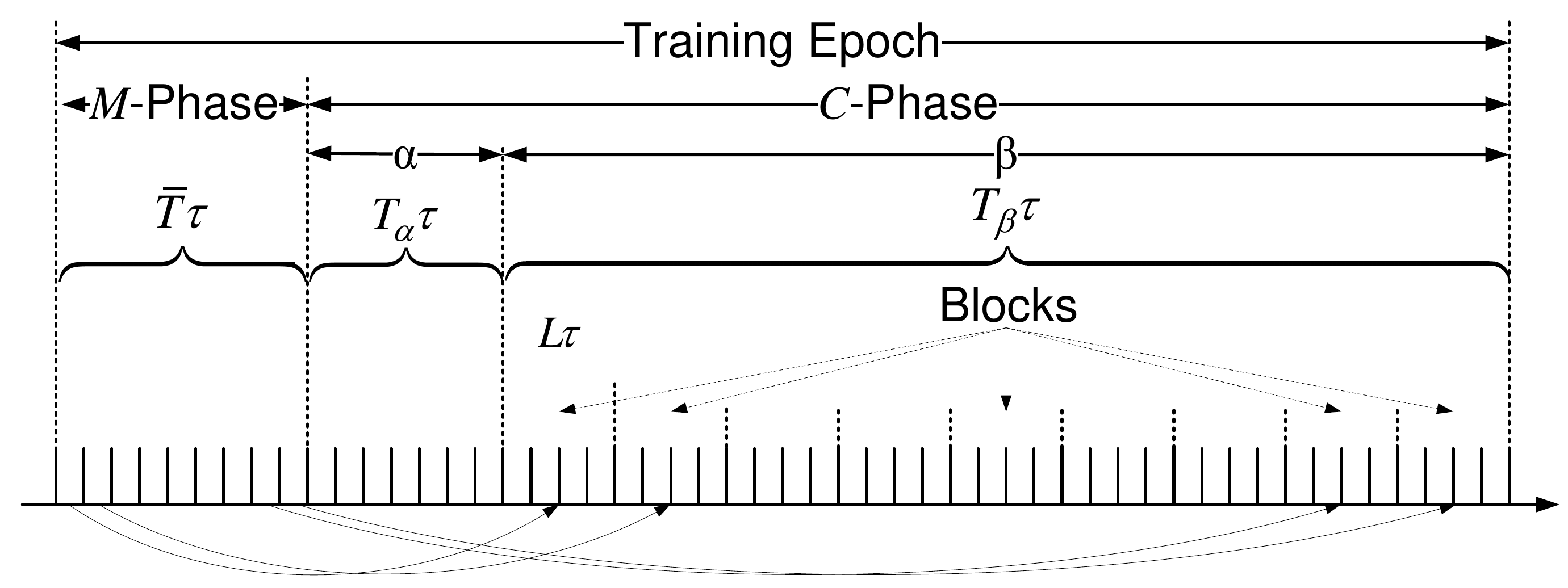}
\caption{Proposed training epoch organization: phases, sub-phases and blocks.}
\label{TimeAxis2}
\end{figure}

The temporal organization of the proposed training protocol is depicted in Fig.~\ref{TimeAxis2}, see also Fig.~\ref{PhasesSubphases}.
The $C$-phase is further split into sub-phases $\alpha$ (\emph{channel estimation} sub-phase) and $\beta$ (\emph{modulation and demodulation} sub-phase).
The $M$-phase contains the first $\overline{T}$ slots, indexed in $\overline{\mathcal{T}}=\left\{1,\hdots,\overline{T}\right\}$, the sub-phase $\alpha$ takes the subsequent $T^{\alpha}$ slots indexed in $\mathcal{T}^{\alpha}=\left\{\overline{T}+1,\hdots,\overline{T}+T^{\alpha}\right\}$ and the sub-phase $\beta$ comprises the remaining $T^{\beta}=T-\overline{T}-T^{\alpha}$ slots indexed in $\mathcal{T}^{\beta} = \left\{\overline{T}+T^{\alpha}+1,\hdots,T\right\}$.
The sub-phase $\beta$ is further split into $\overline{T}$ \emph{blocks}, one for each slot in the $M$-phase, see Fig.~\ref{TimeAxis2}; hence, the blocks are indexed in $\overline{\mathcal{T}}$.
Each block is formed by $L$ consecutive time slots, such that $L\overline{T} = T^{\beta}$.
We write $\mathcal{T}^{\beta}=\cup_{b\in\overline{\mathcal{T}}}\mathcal{T}^{\beta;b}$ where $\mathcal{T}^{\beta;b}=\left\{\overline{T} + T^{\alpha} + (b-1)L + 1,\hdots,\overline{T} + T^{\alpha} + bL\right\},~b\in\overline{\mathcal{T}}$ is the set indexing the slots in block $b$.
As elaborated in subsection~\ref{sec:phases_subphases}, in block $b$, the controllers disseminate the measurements obtained in slot $b$ in the $M$-phase, see Fig.~\ref{PhasesSubphases}.
We introduce notation corresponding to (sub-)phase-wise and block-wise partition of the matrices $\mathbf{W}$, $\mathbf{X}$, $\mathbf{S}$, $\mathbf{V}$ and $\mathbf{\Omega}$.
Take the measurement matrix $\mathbf{W}$ as an example (analogous notation applies to $\mathbf{X}$, $\mathbf{S}$, $\mathbf{V}$ and $\mathbf{\Omega}$); it can be partitioned as, see Fig.~\ref{PhasesSubphases}:
\begin{equation}\nonumber
\mathbf{W} = 
\begin{bmatrix}
\overline{\mathbf{W}} \\
\mathbf{W}^{\alpha} \\
\mathbf{W}^{\beta}
\end{bmatrix},\;
\mathbf{W}^{\beta} =
\begin{bmatrix}
\mathbf{W}^{\beta;1} \\
\vdots\\
\mathbf{W}^{\beta;\overline{T}}
\end{bmatrix}.
\end{equation}
The $\overline{T}\times N$ matrix $\overline{\mathbf{W}}$, with $[\overline{\mathbf{W}}]_{t,n}=w_n(t),~n\in\mathcal{N},~t\in\overline{\mathcal{T}}$,  contains the steady state bus voltage measurements from the $M$-phase; $\mathbf{W}^{\alpha}$, $\mathbf{W}^{\beta}$ as well as each of the matrices $\mathbf{W}^{\beta;b},~b\in\overline{\mathcal{T}}$ are defined analogously.
$\overline{\mathbf{w}}_n$ denotes the $n-$th column of $\overline{\mathbf{W}}$; analogous notation applies to the other matrices.

\subsection{Sufficient Excitation}
\label{sec:suff_exct}

The purpose of the $C$-phase is to enable each controller to learn $\overline{\mathbf{W}}$, which is sufficient to generate locally a unique estimate of $\boldsymbol{\theta}_{-n}$ for any $n\in\mathcal{N}$, if and only if the Jacobians of $\mathsf{vec}(\overline{\mathbf{\Omega}})$ w.r.t. $\boldsymbol{\theta}_{-n}$ and $\mathsf{vec}(\overline{\mathbf{V}})$, denoted with $\mathbf{\Upsilon}_{-n}$ and $\mathbf{\Gamma}$, respectively, satisfy the rank conditions:
\begin{align}\label{eq:suff_cond1}
\mathsf{rank}(\mathbf{\Upsilon}_{-n}) & = \mathsf{dim}(\boldsymbol{\theta}_{-n}),\\\label{eq:suff_cond2}
\mathsf{rank}(\mathbf{\Gamma}) & = N\overline{T},
\end{align}
for any $n\in\mathcal{N}$.
The sufficient excitation conditions provide practical guidelines for designing the training matrices $\overline{\mathbf{X}}$ and $\overline{\mathbf{S}}$; this is further discussed in subsection~\ref{sec:disc}.

We note that the vectorization of $\overline{\mathbf{\Omega}}$ is \emph{linear} in $\boldsymbol{\theta}$:
\begin{equation}\label{eq:g_cond_slot_lin_param}
\mathsf{vec}(\overline{\mathbf{\Omega}}) = \mathbf{\Upsilon}\boldsymbol{\theta} =\mathbf{0}_{\overline{T}N}.
\end{equation}
In fact, it can be shown that it is always linear in $\mathbf{d}$ and $\boldsymbol{\psi}$; however, the linearity in $\mathbf{g}$ is a direct corollary of the virtual resistance configuration \eqref{eq:droop_gen} for proportional power sharing based on the instantaneous generation capacities.
This result is useful for finding good initial estimates of $\boldsymbol{\theta}_{-n}$ which will be used to initialize the iterative algorithm.

\begin{figure}[t]
\centering
\includegraphics[scale=0.35]{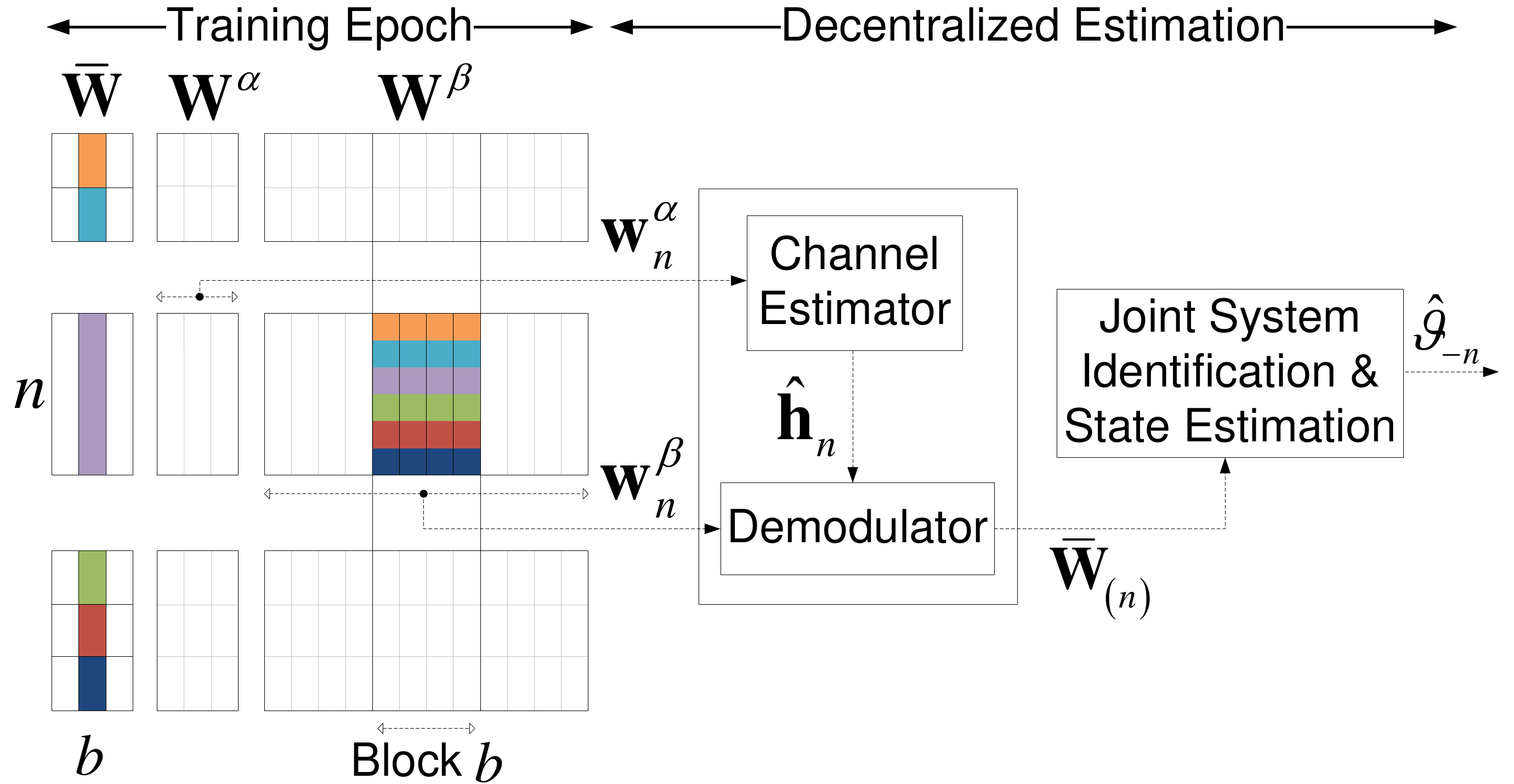}
\caption{Proposed decentralized solution.} 
\label{PhasesSubphases}
\end{figure}


\subsection{Training Phases and Sub-phases}
\label{sec:phases_subphases}

In the $M$-phase, the $n-$th controller obtains $\overline{\mathbf{w}}_n$, 
the $n$-th column of $\overline{\mathbf{W}}$.
Learning the remaining columns $\overline{\mathbf{w}}_n,~n\neq m$ and obtaining local copy of $\overline{\mathbf{W}}$, denoted with $\overline{\mathbf{W}}_{(n)}$, is done in the $C$-phase where controller $n$ disseminates $\overline{\mathbf{w}}_n$ to remote controllers by modulating the amplitudes of the reference voltage deviations and, in the same time, demodulates $\overline{\mathbf{w}}_m,~m\neq n$ from the locally available measurements $\mathbf{w}_n^{\alpha/\beta}$ via sequential demodulator, see Fig.~\ref{PhasesSubphases}.

In the $C$-phase, we adopt the following perturbation signals:
\begin{align}
x_n(t) = \tilde{x} + \sqrt{\pi_n(t)}\Delta x_n(t),~s_n(t) = \tilde{s},~n\in\mathcal{N},~t\in\mathcal{T}^{\alpha/\beta},
\end{align}
where $\Delta x_n(t)\in[-1,+1]$ is the reference voltage perturbation and {$\sqrt{\pi_n(t)}>0$ is the \emph{perturbation amplitude}}; hence, the droop slopes in the $C$-phase are kept fixed to the nominal value and the communication channel is established via the reference voltage perturbation signals.
The $C$-phase training matrices $\mathbf{X}^{\alpha}$ and $\mathbf{X}^{\beta}$ can then be written as follows:
\begin{equation}\label{eq:ref_vol_perturb}
	\mathbf{X}^{\alpha/\beta} = \tilde{{x}}\mathbf{1}_{T^{\alpha/\beta}\times N} + \mathbf{\Pi}^{\alpha/\beta}\odot\Delta\mathbf{X}^{\alpha/\beta},
\end{equation}
where $\Delta\mathbf{X}$ and $\mathbf{\Pi}$ are the reference voltage perturbation and perturbation amplitude matrices, defined as $[\Delta\mathbf{X}]_{t,n} = \Delta x_n(t)$ and $[\mathbf{\Pi}]_{t,n} = \sqrt{\pi_n(t)}$, $n\in\mathcal{N},~{t}\in\mathcal{T}^{\alpha/\beta}$, respectively.
To facilitate the design of the demodulator, we make the following small signal assumption: the reference voltage deviation amplitudes in the $C$-phase are relatively small w.r.t. the nominal reference voltage, i.e., $\pi_n(t) \ll {\tilde{x}}_n$, $n\in\mathcal{N}$, $t\in\mathcal{T}^{\alpha/\beta}$.
Using Taylor's series expansion, the signal collected by controller $n$ in the $C$-phase can be written as:
\begin{align}\label{eq:lin_approx_time}
 \mathbf{w}_n^{\alpha/\beta} \approx \tilde{v}_n\mathbf{1}_{T^{\alpha/\beta}} + (\mathbf{\Pi}^{\alpha/\beta}\odot\Delta\mathbf{X}^{\alpha/\beta})\mathbf{h}_n + \mathbf{z}_n^{\alpha/\beta}.
\end{align}
The model above defines the input-output relation of a real, linear, synchronous communication channel with \emph{channel vector} given by the gradient $\mathbf{h}_n$ (evaluated at the nominal droop values) which contains the real coefficients of the equivalent linear channels that controller $n$ sees to the other controllers; in localized and strongly connected MGs, the entries in $\mathbf{h}_n$ do not differ significantly (see also \cite{20}), i.e., the channel \eqref{eq:lin_approx_time} experiences strong all-to-all property.

We use the linear model to design sequential transceiver that operates as follows. 
First, in sub-phase $\alpha$, controller $k$ estimates $\mathbf{h}_k$; for this purpose, we fix the perturbation amplitudes to be all \emph{known} and \emph{equal} constants:
\begin{equation}\label{eq:pi_alpha}
	\sqrt{\pi_n(t)}=\sqrt{\pi^{\alpha}},~n\in\mathcal{N},~t\in\mathcal{T}^{\alpha}.
\end{equation}
Then, in sub-phase $\beta$ the controllers \emph{disseminate} the information acquired in the $M$-phase via the following linear amplitude modulation (without any additional error protection): 
\begin{equation}\label{eq:amp_beta}
	\sqrt{\pi_n(t)} = \sqrt{\pi^{\beta}}(\overline{w}_n(b) - \chi_n),~n\in\mathcal{N},~t\in\mathcal{T}^{\beta;b},~b\in\overline{\mathcal{T}},
\end{equation}
where $\pi^{\beta}$ and $\chi_{n}$ are known positive constants.
Clearly, $\pi_n(t)$ remains fixed in block $b\in\overline{\mathcal{T}}$, carrying the information about $\overline{w}_n(b)$ by embedding it into the amplitude of the perturbation signal $\Delta\mathbf{x}_n^{\beta;b}$.
The controllers operate in \emph{full duplex} transmission mode,  \emph{simultaneously} broadcasting and receiving one voltage measurement per block to/from all other controllers.\footnote{The scheme suits well channels with strong all-to-all property, i.e., channels where the gains in $\mathbf{h}_k$ do not differ significantly; this is the case for small and localized MGs. As the system grows in size and scope, the all-to-all property ceases to be valid and one should consider applying more sophisticated digital modulation/demodulation and scheduling schemes, including error protection coding; see \cite{19,20} for alternatives.} 

To guarantee the uniqueness of the local copies $\overline{\mathbf{W}}_{(n)}$, we restrict the columns of the reference voltage perturbation matrices $\Delta\mathbf{X}^{\alpha}$ and $\Delta\mathbf{X}^{\beta;b}$ to be zero mean and orthogonal:
\begin{align}\label{eq:train_alpha}
&(\Delta\mathbf{X}^{\alpha})^{\mathsf{T}}\mathbf{1}_{T^{\alpha}} =\mathbf{0}_N,~(\Delta\mathbf{X}^{\alpha})^{\mathsf{T}}\Delta\mathbf{X}^{\alpha} =\delta^{\alpha}\mathbf{I}_N,\\\label{eq:train_beta}
&(\Delta\mathbf{X}^{\beta;b})^{\mathsf{T}}\mathbf{1}_{L} =\mathbf{0}_N,~(\Delta\mathbf{X}^{\beta;b})^{\mathsf{T}}\Delta\mathbf{X}^{\beta;b} =\delta^{\beta}\mathbf{I}_N,
\end{align}
where $\delta^{\alpha}=\|\Delta\mathbf{x}_n^{\alpha}\|_2^2\leq T^{\alpha}$, $\delta^{\beta}=\|\Delta\mathbf{x}_n^{\beta;b}\|_2^2\leq L$, for every $n\in\mathcal{N}$, $b\in\overline{\mathcal{T}}$.
We note that the above assumptions are a bit restrictive.
Given the perturbation signals \eqref{eq:pi_alpha} and \eqref{eq:amp_beta} in sub-phases $\alpha$ and $\beta$, the sufficient conditions for uniqueness of $\overline{\mathbf{W}}_{(n)}$ for any $n\in\mathcal{N}$ are $\mathsf{rank}(\Delta\mathbf{X}^{\alpha})=\mathsf{rank}(\Delta\mathbf{X}^{\beta;b})=N$ for any $b\in\overline{\mathcal{T}}$; however, we use \eqref{eq:train_alpha}, \eqref{eq:train_beta} for convenience, namely, to obtain compact expression for ${\overline{\mathbf{W}}}_{(n)}$ without loosing generality.
Replacing \eqref{eq:pi_alpha} and \eqref{eq:amp_beta} in \eqref{eq:lin_approx_time} and using assumptions \eqref{eq:train_alpha}, \eqref{eq:train_beta}, we derive $\overline{\mathbf{W}}_{(n)}$: 

\begin{proposition}\label{prop:demodulator}
The local estimators of $\mathsf{vec}(\overline{\mathbf{W}})$ are given by:
\begin{align}\label{eq:Wk_full_est}
\mathsf{vec}({\overline{\mathbf{W}}}_{(n)}) = \frac{\sqrt{\pi^{\alpha}}\delta^{\alpha}}{\sqrt{\pi^{\beta}}\delta^{\beta}}\mathsf{D}^{-1}(\boldsymbol{\mathcal{X}}^{\alpha}\mathbf{w}_n^{\alpha})\sum_{b\in\overline{\mathcal{T}}}(\boldsymbol{\mathcal{X}}^{\beta;b}\mathbf{w}_n^{\beta;b}) + \boldsymbol{\mathcal{I}}\boldsymbol{\chi},
\end{align}
for any $n\in\mathcal{N}$; for notational brevity, we used $\boldsymbol{\mathcal{X}}^{\alpha} = (\Delta\mathbf{X}^{\alpha})^{\mathsf{T}}\otimes\mathbf{1}_{\overline{T}}$, $\boldsymbol{\mathcal{X}}^{\beta;b} = (\Delta\mathbf{X}^{\beta;b})^{\mathsf{T}}\otimes\mathbf{e}_{b}$, $\boldsymbol{\mathcal{I}}=\mathbf{I}_N\otimes\mathbf{1}_{\overline{T}}$ and $\boldsymbol{\chi} = [\chi_1,\hdots,\chi_N]^{\mathsf{T}}$.
\end{proposition}
\begin{proof}
See Appendix~\ref{app:propII}.
\end{proof}
By the end of the training epoch, the $n-$th controller has a local copy of the $M$-phase measurement matrix $\mathsf{vec}(\overline{\mathbf{W}}_{(n)})$; if the sufficient excitation conditions \eqref{eq:suff_cond1}, \eqref{eq:suff_cond2} hold, then $\mathsf{vec}(\overline{\mathbf{W}}_{(n)})$ is sufficient to estimate $\boldsymbol{\theta}_{-n}$.
Formulating an ML estimation problem using $\mathsf{vec}(\overline{\mathbf{W}}_{(n)})$ requires knowledge of the pdf $\rho(\mathsf{vec}(\overline{\mathbf{W}}_{(n)});\boldsymbol{\theta})$; however, obtaining the closed from expression is tedious since \eqref{eq:Wk_full_est} involves ratios of non-zero Gaussian random variables.
Therefore, we derive Gaussian approximation for the pdf of $\mathsf{vec}(\overline{\mathbf{W}}_{(n)})$ based on first-order perturbation-theoretic approach (see supplementary material).
Using Neumann series expansion, we get:
\begin{equation}\label{eq:Gauss_approx_full}
\rho(\mathsf{vec}(\overline{\mathbf{W}}_{(n)});\boldsymbol{\theta}) \approx \mathsf{N}(\mathsf{vec}(\overline{\mathbf{V}}),\mathbf{\Sigma}).
\end{equation}
The covariance matrix $\mathbf{\Sigma}$ can be computed via the first-order approximation (see Appendix~\ref{app:Gauss_app}) and is given by:
\begin{align}\label{eq:full_cov_est}
{\mathbf{\Sigma}} = \sigma^2\left(\mathbf{I}_{\overline{T}N} + \frac{\pi^{\alpha}(\delta^{\alpha})^2}{\pi^{\beta}\delta^{\beta}}\mathsf{D}^{-2}(\boldsymbol{\mathcal{X}}^{\alpha}\mathbf{v}_n^{\alpha}) + \frac{\pi^{\alpha}(\delta^{\alpha})^3}{\pi^{\beta}(\delta^{\beta})^2}\mathsf{D}^{-2}(\boldsymbol{\mathcal{X}}^{\alpha}\mathbf{v}_n^{\alpha})\left(\sum_{b\in\overline{\mathcal{T}}}\mathsf{D}(\boldsymbol{\mathcal{X}}^{\beta;b}\mathbf{v}_n^{\beta;b})(\mathbf{I}_N\otimes\mathbf{1}_{\overline{T}\times\overline{T}})\mathsf{D}(\boldsymbol{\mathcal{X}}^{\beta;b}\mathbf{v}_n^{\beta;b})\right)\mathsf{D}^{-2}(\boldsymbol{\mathcal{X}}^{\alpha}\mathbf{v}_n^{\alpha})\right).
\end{align}
The approximation is valid for sub-phase $\alpha$ signals satisfying $\mathbf{0}_{T^{\alpha}}<\mathbf{w}_n^{\alpha}< 2\mathbf{v}_n^{\alpha}$.
In practice, this is expected to be satisfied as the probability that {$\mathbf{w}_n^{\alpha}$} is negative or larger than $2\mathbf{v}_n^{\alpha}$ is negligible.
In light of this, one can easily verify that the Gaussian approximation converges to the true distribution of $\mathsf{vec}(\overline{\mathbf{W}}_{(n)})$ in the limit $\mathbf{w}_n^{\alpha}\rightarrow\mathbf{v}_n^{\alpha}$. 
Expression \eqref{eq:full_cov_est} also captures the effect of the $C$-phase and the transmission schemes we adopted there on the uncertainty in the local copies $\overline{\mathbf{W}}_{(n)}$; specifically, the initial uncertainty in $\overline{\mathbf{W}}$, represented with the first term in \eqref{eq:full_cov_est}, increases due to (1) measurement noise in sub-phase $\beta$ (second term) and, (2) the uncertainty in the channel estimates induced in sub-phase $\alpha$ (third term).

\subsection{{Joint System Identification and State Estimation}}
\label{sec:JMLE}

By the end of the training epoch, the $n-$th controller has $\overline{\mathbf{W}}_{(n)}$ and the $C$-phase measurement vectors $\mathbf{w}_n^{\alpha}$ and $\mathbf{w}_n^{\beta}$.
The reference voltage training matrix $\mathbf{X}^{\alpha}$ is deterministic, so $\mathbf{w}_n^{\alpha}$ can still be useful when formulating the estimation problem.
On the other hand, the training matrix $\mathbf{X}^{\beta}$ in sub-phase $\beta$ is modulated with $M$-phase measurements; since controller $n$ knows only the noisy copy $\overline{\mathbf{W}}_{(n)}$, it is impossible to reconstruct $\mathbf{X}^{\beta}$ perfectly which makes $\mathbf{w}_n^{\beta}$ of no further use.
The optimal ML that uses \emph{all} available information should be defined over an augmented vector, comprising $\mathsf{vec}(\overline{\mathbf{W}}_{(n)})$ and $\mathbf{w}_n^{\alpha}$. 
Including $\mathbf{w}_n^{\alpha}$ increases the dimensionality of the problem, but the numerical investigations indicate that it does not yield any practically significant performance gain. We therefore omit $\mathbf{w}_n^{\alpha}$ from the ML for clarity of exposition.

The relation between the steady state bus voltages and the parameter vector is defined implicitly in Proposition~\ref{prop2}; therefore, we define a \emph{joint system identification and state estimation} (J-SISE) problem via \emph{constrained} ML estimation \cite{32,31}. 
We introduce the joint parameter/state vector:
\begin{equation}
\boldsymbol{\vartheta} = \begin{bmatrix}\boldsymbol{\theta}\\\mathsf{vec}(\overline{\mathbf{V}})\end{bmatrix}.
\end{equation}
We define $\hat{\boldsymbol{\vartheta}}_{-n},~n\in\mathcal{N}$ as the globally optimal solution to:
\begin{align}\label{eq:MLE_general_SB}
\hat{\boldsymbol{\vartheta}}_{-n} & = \min_{\boldsymbol{\vartheta}_{-n}}\left\{-\ln{\rho(\mathsf{vec}(\overline{\mathbf{W}}_{(n)});\boldsymbol{\theta})}\right\},\\\nonumber
 \text{s.t. } & \mathsf{vec}(\overline{\mathbf{\Omega}}) = \mathbf{0}_{\overline{T}N},
\end{align}
formulated w.r.t. the true distribution of $\mathsf{vec}(\overline{\mathbf{W}}_{(n)})$.
The problem \eqref{eq:MLE_general_SB} is neither convex nor concave due to the quadratic nature of the constrains that contain bilinear terms in the decision variables.
Since $\mathsf{vec}({\mathbf{\Omega}})$ is sufficiently differentiable in $\boldsymbol{\vartheta}_{-n}$, the constrained optimization problem \eqref{eq:MLE_general_SB} can be restated as an unconstrained one using the Lagrange method of multipliers \cite{30}.
Using the Gaussian approximation \eqref{eq:Gauss_approx_full} and applying the Karush-Kuhn-Tucker (KKT) conditions yields a {non-linear} system of equations.
Using the result of the following proposition, we propose Algorithm~\ref{alg1} based on partially linearized constraints to solve the system iteratively \cite{31}.
Specifically, denote $\boldsymbol{\vartheta}^{(j)}$ in the 
 $j$-th iteration and let:
\begin{equation}\label{eq:constraint_lin}
\mathsf{vec}(\overline{\mathbf{\Omega}})\approx {\mathbf{\Upsilon}}^{(j)}\boldsymbol{\theta} + {\mathbf{\Gamma}}^{(j)}(\mathsf{vec}(\overline{\mathbf{V}})-\mathsf{vec}(\overline{\mathbf{V}}^{(j)})),
\end{equation}
be the linear approximation of $\mathsf{vec}(\overline{\mathbf{\Omega}})$ around $\boldsymbol{\vartheta}^{(j)}$. The Jacobians ${\mathbf{\Upsilon}}^{(j)}$, ${\mathbf{\Gamma}}^{(j)}$ are evaluated in $\boldsymbol{\vartheta}^{(j)}$.
We obtain the following result:

\begin{proposition}\label{prop4}
If the sufficient excitation conditions \eqref{eq:suff_cond1}, \eqref{eq:suff_cond2} are satisfied in $\boldsymbol{\vartheta}^{(j)}$, the global solution to \eqref{eq:MLE_general_SB} after substituting the power balance constraint with \eqref{eq:constraint_lin} is given by:
\begin{align}
\label{eq:sol_KKT_1}
&\boldsymbol{\theta}_{-n} = - (({\mathbf{\Upsilon}}_{-n}^{(j)})^{\mathsf{T}}({\mathbf{\Gamma}}^{(j)}{\mathbf{\Sigma}}({\mathbf{\Gamma}}^{(j)})^{\mathsf{T}})^{-1}{\mathbf{\Upsilon}}_{-n}^{(j)})^{-1}({\mathbf{\Upsilon}}_{-n}^{(j)})^{\mathsf{T}}({\mathbf{\Gamma}}^{(j)}{\mathbf{\Sigma}}({\mathbf{\Gamma}}^{(j)})^{\mathsf{T}})^{-1}({\boldsymbol{\upsilon}}_{n}^{(j)}g_n + ({\mathbf{\Gamma}}^{(j)})^{\mathsf{T}}(\mathsf{vec}(\overline{\mathbf{W}}_{(n)}) - \mathsf{vec}(\overline{\mathbf{V}}^{(j)}))),\\
\label{eq:sol_KKT_2}
&\mathsf{vec}(\overline{\mathbf{V}}) = \mathsf{vec}(\overline{\mathbf{W}}_{(n)}) - {\mathbf{\Sigma}}({\mathbf{\Gamma}}^{(j)})^{\mathsf{T}}({\mathbf{\Gamma}}^{(j)}{\mathbf{\Sigma}}({\mathbf{\Gamma}}^{(j)})^{\mathsf{T}})^{-1}({\mathbf{\Upsilon}}^{(j)}\boldsymbol{\theta} + ({\mathbf{\Gamma}}^{(j)})^{\mathsf{T}}(\mathsf{vec}(\overline{\mathbf{W}}_{(n)}) - \mathsf{vec}(\overline{\mathbf{V}}^{(j)}))).
\end{align}
\end{proposition}
\begin{proof}
See Appendix~\ref{app:propIII}.
\end{proof}

\begin{algorithm}[t]
\caption{J-SISE with partially linearized constraints}
\begin{algorithmic}[1]\label{alg1}
\renewcommand{\algorithmicrequire}{\textbf{Input:}}
\renewcommand{\algorithmicensure}{\textbf{Output:}}
\REQUIRE ${\overline{\mathbf{W}}}_{(n)}$, $\overline{\mathbf{X}}$, $\overline{\mathbf{S}}$, $\epsilon$, evaluate $\mathbf{\Sigma}$ via \eqref{eq:full_cov_est} using $\mathbf{w}_n^{\alpha}$ and $\mathbf{w}_n^{\beta}$
\ENSURE  $\hat{\boldsymbol{\vartheta}}_{-n}$
\\ \textit{Initialization}: $j = 0$, ${\varepsilon} = \infty$, compute $\boldsymbol{\theta}_{-n}^{(0)}$, $\mathsf{vec}(\overline{\mathbf{V}}^{(0)})$ via \eqref{eq:init_theta}, \eqref{eq:init_v}
\WHILE {${\varepsilon}\geq\epsilon$}
  \STATE evaluate $\mathbf{\Upsilon}^{(j)}$, $\mathbf{\Gamma}^{(j)}$ using $\boldsymbol{\vartheta}_{-n}^{(j)}$
	\STATE compute $\boldsymbol{\theta}_{-n}^{(j+1)}$, $\mathsf{vec}(\overline{\mathbf{V}}^{(j+1)})$ via \eqref{eq:sol_KKT_1} and \eqref{eq:sol_KKT_2}
	\STATE compute $\varepsilon = \|\boldsymbol{\vartheta}_{-n}^{(j+1)} - \boldsymbol{\vartheta}_{-n}^{(j)}\|$, $j=j+1$
\ENDWHILE
\RETURN $\boldsymbol{\vartheta}_{-n}^{(j+1)}$ 
\end{algorithmic} 
\end{algorithm}


The algorithm starts with an initial guess $\boldsymbol{\vartheta}_{-n}^{(0)}$.
Then, we apply Proposition~\ref{prop4} iteratively; the solutions \eqref{eq:sol_KKT_1}, \eqref{eq:sol_KKT_2} in each iteration serve as an input for the next iteration until convergence.
In order to apply Algorithm~\ref{alg1}, controller $n$ should know the covariance matrix $\mathbf{\Sigma}$ up to a scaling factor, i.e., knowledge of the noise variance $\sigma^2$ is not necessary.
To ensure fast convergence, we propose the following initialization: once $\overline{\mathbf{W}}_{(n)}$ is locally available, a reasonable initial estimate of the state $\mathsf{vec}(\overline{\mathbf{V}}^{(0)})$ can be obtained via eq. \eqref{eq:noisy_measurement}:
\begin{equation}\label{eq:init_v}
\mathsf{vec}(\overline{\mathbf{V}}^{(0)}) = \mathsf{vec}(\overline{\mathbf{W}}_{(n)}).
\end{equation}
Then, we evaluate ${\mathbf{\Upsilon}}$ in $\mathsf{vec}(\overline{\mathbf{V}}^{(0)})$ and solve \eqref{eq:g_cond_slot_lin_param} for $\boldsymbol{\theta}_{-n}$:
\begin{equation}\label{eq:init_theta}
\boldsymbol{\theta}_{-n}^{(0)} = - ({\mathbf{\Upsilon}}_{-n}^{(0)})^{\dagger}{\boldsymbol{\upsilon}}_{n}^{(0)}g_n,
\end{equation} 
where ${\boldsymbol{\upsilon}}_{n}^{(0)}$ is the $n$-th column of ${\mathbf{\Upsilon}}^{(0)}$.
It can be easily verified that $\boldsymbol{\vartheta}_{-n}^{(0)}$ satisfies the KKT conditions and is a stationary point of the objective in \eqref{eq:MLE_general_SB}. Section~\ref{sec:Results} shows that \eqref{eq:init_theta} is {unbiased} but not efficient estimator of $\boldsymbol{\theta}_{-n}$.
In this regard, Algorithm~\ref{alg1} serves to refine the initial estimate $\boldsymbol{\vartheta}_{-n}^{(0)}$ and further reduce its covariance.

\subsection{Performance}
\label{sec:CRLB}

The Mean Squared Error matrix of the unbiased estimator of ${\boldsymbol{\vartheta}}_{-n}$ is defined as:
\begin{equation}
\text{MSE}(\hat{\boldsymbol{\vartheta}}_{-n}) = \mathbb{E}\left\{(\hat{\boldsymbol{\vartheta}}_{-n} - \boldsymbol{\vartheta}_{-n})(\hat{\boldsymbol{\vartheta}}_{-n} - \boldsymbol{\vartheta}_{-n})^{\mathsf{T}}\right\}.
\end{equation}
$\text{MSE}(\hat{\boldsymbol{\theta}}_{-n})$ and $\text{MSE}(\mathsf{vec}(\hat{\overline{\mathbf{V}}}))$ are defined analogously.
In stead of deriving the MSE matrix directly, we use the CRLB inequality to bound it and derive an approximate lower bound using the Gaussian approximation \eqref{eq:Gauss_approx_full}. 
Referring to the optimization problem \eqref{eq:MLE_general_SB}, a straightforward way to bound $\text{MSE}(\hat{\boldsymbol{\vartheta}}_{-n})$ is to use the \emph{constrained} CRLB \cite{rev1}.
Let $\mathbf{O}$ denote the $\mathsf{dim}(\boldsymbol{\vartheta}_{-n})\times\mathsf{dim}(\boldsymbol{\theta}_{-n})$ matrix whose columns form the orthonormal basis for the null space of the Jacobian $[{\mathbf{\Upsilon}}_{-n},\;{\mathbf{\mathbf{\Gamma}}}]$. 
Then, $\text{MSE}(\hat{\boldsymbol{\vartheta}}_{-n})$ can be bounded as follows \cite{rev1}:
\begin{equation}\label{eq:CRLB_standard}
\textbf{\normalfont MSE}(\hat{\boldsymbol{\vartheta}}_{-n})\succeq \mathbf{O}\left(\mathbf{O}^{\mathsf{T}}\begin{bmatrix} \mathbf{0} & \mathbf{0}\\\mathbf{0} & \mathbf{\Sigma}^{-1} \end{bmatrix}\mathbf{O}\right)^{-1}\mathbf{O}^{\mathsf{T}},
\end{equation}
where $\mathbf{0}$ denote all-zero matrices of adequate dimensions.
$\mathbf{O}$ is is computed numerically, as it is unavailable in closed form.
To bound the MSE matrices of $\hat{\boldsymbol{\theta}}_{-k}$ and $\mathsf{vec}(\hat{\overline{\mathbf{V}}})$ separately, we need to perform numerical block inversion of the right-hand side of \eqref{eq:CRLB_standard}; the following proposition gives alternative and simpler closed form expressions for these bounds:

\begin{proposition}
\label{prop5}
The MSE matrices $\textbf{\normalfont MSE}(\hat{\boldsymbol{\theta}}_{-n})$ and $\textbf{\normalfont MSE}(\mathsf{vec}(\hat{\overline{\mathbf{V}}}))$ can be bounded from below as follows:
\begin{align}\label{eq:CRLB_theta}
& \textbf{\normalfont MSE}(\hat{\boldsymbol{\theta}}_{-n}) \succeq ({\mathbf{\Upsilon}}_{-n}^{\mathsf{T}}({\mathbf{\Gamma}}^{-1})^{\mathsf{T}}\mathbf{\Sigma}^{-1}{\mathbf{\Gamma}}^{-1}{\mathbf{\Upsilon}}_{-n})^{-1} =\boldsymbol{\mathcal{J}}^{-1},\\\label{eq:CRLB_v}
& \textbf{\normalfont MSE}(\mathsf{vec}(\hat{\overline{\mathbf{V}}})) \succeq {\mathbf{\Gamma}}^{-1}{\mathbf{\Upsilon}}_{-n}\boldsymbol{\mathcal{J}}^{-1}{\mathbf{\Upsilon}}_{-n}^{\mathsf{T}}({\mathbf{\Gamma}}^{-1})^{\mathsf{T}},
\end{align}
where $\boldsymbol{\mathcal{J}}$ denotes the Fisher Information Matrix of $\boldsymbol{\theta}_{-n}$.
\end{proposition}
\begin{proof}
See Appendix~\ref{app:propIV}.
\end{proof}

The expressions \eqref{eq:CRLB_theta} and \eqref{eq:CRLB_v} can be verified to be asymptotically tight; it can be shown that if Algorithm~\ref{alg1} converges to the global optimum, the MSE matrix is of the same analytical form as \eqref{eq:CRLB_theta} and \eqref{eq:CRLB_v}, but evaluated at $\hat{\boldsymbol{\vartheta}}_{-n}$.
Conversely, expressions \eqref{eq:CRLB_theta} and \eqref{eq:CRLB_v} prove the asymptotic efficiency of Algorithm~\ref{alg1}.

\subsection{Discussion}
\label{sec:disc}

We take a closer look on few crucial aspects that set the applicability boundaries of the proposed method.
We consider the sufficient excitation conditions, outlined in subsection~\ref{sec:suff_exct}; they provide guidelines for designing the training sequences and they determine the overall duration of the training epoch.
A straightforward way to guarantee \eqref{eq:suff_cond1}, \eqref{eq:suff_cond2} is to ensure that $\overline{T} \geq N^{-1}\mathsf{dim}(\boldsymbol{\theta}_{-n})$ and $\mathsf{rank}(\overline{\mathbf{X}}) = N$ and/or $\mathsf{rank}(\overline{\mathbf{S}}) = N$.
The minimal duration of the $C$-phase is determined by the conditions for uniqueness of ${\overline{\mathbf{W}}}_{(n)}$, such that the total duration of the training epoch (in slots) $T = \overline{T}(1+L) + T^{\alpha}$ in a system with $N$ DERs is lower bounded as:
\begin{equation}
T\geq \frac{1}{2}N^2 + 5N + \frac{5}{2} - \frac{1}{N} = T_{\min}.
\end{equation}
The lower bound on $T$ can be attained by random training sequences. 
Alternatively, when using deterministic codes such as orthogonal Walsh-Hadamard sequences, meeting the rank conditions and, possibly additional conditions such as \eqref{eq:train_alpha}, \eqref{eq:train_beta} might require more time slots than $T_{\min}$.

The frequency of the training epoch should match the requirements of the upper layer application.
If the application runs periodically, then the training epoch should be invoked in each period, preferably at the beginning, while in event-triggered applications, the training epoch should be invoked whenever the application is triggered.
While the frequencies should be equal, the total duration (in seconds) $T\tau$ is expected to constitute only a fraction $0<\gamma<1$ of the average time $\tau^{\text{u.app}}$ between two consecutive application runs.
Then, we have the following upper bound on the slot duration:
\begin{equation}
\tau \leq \frac{\gamma\tau^{\text{u.app}}}{T}\leq \frac{\tau^{\text{u.app}}}{T_{\min}}=\tau_{\max},
\end{equation}
where $\tau_{\max}$ is obtained by fixed $\gamma = 1$ and $T=T_{\min}$.

Further, since the proposed method is developed in classical estimation framework, each controller requires perfect knowledge of the training matrices $\overline{\mathbf{X}}$, $\overline{\mathbf{S}}$, $\Delta{\mathbf{X}}^{\alpha}$ and $\Delta{\mathbf{X}}^{\beta}$. 
This means that the training matrices should be designed a priori, delivered to the controllers and kept fixed afterward (via hard-coding for instance).
Relaxing this condition requires adequate modifications of the problem formulation, which is out of the paper's scope.
For instance, if no prior knowledge is available, we have no choice but to model the training matrices are deterministic unknowns and modify the definition of the parameter vector to include them. 

The method can only identify buses that host at least one DER whose primary controller engages in decentralized training.
In other words, buses that host only loads are unidentifiable.
However, we can still apply the method in MGs with (potentially many) load buses; in this case, the method identifies the \emph{Kron-reduced} conductance matrix, which is obtained by isolating the DER buses in the original network and applying block inversion on the original conductance matrix.
Analyzing the structure of the Kron-reduced conductance matrix, the controllers might be able to deduce some information on the original conductance matrix, see \cite{33}.


\section{Decentralized OED via Training}
\label{sec:DOED}

We illustrate the practical potential of the proposed system identification method by applying it in decentralized OED (DOED) as the most common upper layer application in power systems. In OED, each DER $n\in\mathcal{N}$ is assigned a monotonic and convex \emph{cost function} $c_n(p_n)$ that determines the cost of the output power $p_n$ of DER $n$.
The aim of the OED is to find the optimal local output powers, referred to as \emph{optimal dispatch policies} $p_n^*,~n\in\mathcal{N}$ that minimize the total cost $\sum_{n\in\mathcal{N}}c_n(p_n)$ such that the total load demand $d^{\star} = \mathbf{1}_{3N}^{\mathsf{T}}\mathbf{d}$ is balanced and the box constraints on the output powers are satisfied:
\begin{align}\label{eq:OED}
 \mathbf{p}^* & = \min_{\mathbf{p}}c(\mathbf{p}),\\\nonumber
\text{s.t. } & \mathbf{1}_N^{\mathsf{T}}\mathbf{p} = d^{\star},~\mathbf{0}_N \leq \mathbf{p} \leq \mathbf{g},~v_{\min}\mathbf{1}_N\leq\mathbf{v}\leq v_{\max}\mathbf{1}_N,
\end{align}
where $c(\mathbf{p})=\sum_{n\in\mathcal{N}}c_n(p_n)$, $\mathbf{p} = [p_1,\hdots,p_N]^{\mathsf{T}}$ and $\mathbf{v} = [v_1,\hdots,v_N]^{\mathsf{T}}$.
Distributed MGs with small-scale DERs typically use linear cost functions~\cite{6}.
Hence, we adopt $c_n(p_n) = a_np_n$ where $a_n$ is the constant \emph{marginal cost} of the $n-$th DER per unit of injected/stored power.
Without loss of generality, the costs are ordered as $a_n \leq a_{n+1},~n\in\mathcal{N}$, which divides the DERs in several ordered cost groups based on the marginal costs. The optimal solution to \eqref{eq:OED} is the following decentralized program:
\begin{align}\label{dOED}
p_{n}^* = \left\{
  \begin{array}{lr}
    g_n 		 &  d^{\star} > \sum_{m:a_m\leq a_n}g_m,\\ 
    0        &  d^{\star} < \sum_{m:a_m < a_n}g_m,\\ 
    g_n\frac{d^{\star} - \sum_{m:a_m < a_n}g_m}{\sum_{m:a_m = a_n}g_m} &  \text{\normalfont otherwise}
  \end{array}
\right.
\end{align}
for any $n\in\mathcal{N}$ (see also \cite{6,last}).
Specifically, the total load demand is first filled with the capacities of the DERs from the cheapest cost groups, until the third condition in \eqref{dOED} is met. 
Then, the DERs from the cost group that meets this condition share the remaining net load demand proportionally to their local capacities while the DERs from the remaining, most expensive cost groups do not inject power. 
The DERs that satisfy the first condition in \eqref{dOED} are operated at a constant power (at capacity) and their local controllers are configured in CSC mode (forming the subset $\mathcal{N}^{\text{C}}$), whereas the DERs that satisfy the third condition have flexible power outputs and their local controllers are configured in VSC mode, tuned for proportional power sharing (forming $\mathcal{N}^{\text{V}}$).

\begin{figure}[t]
\centering
\includegraphics[scale=0.35]{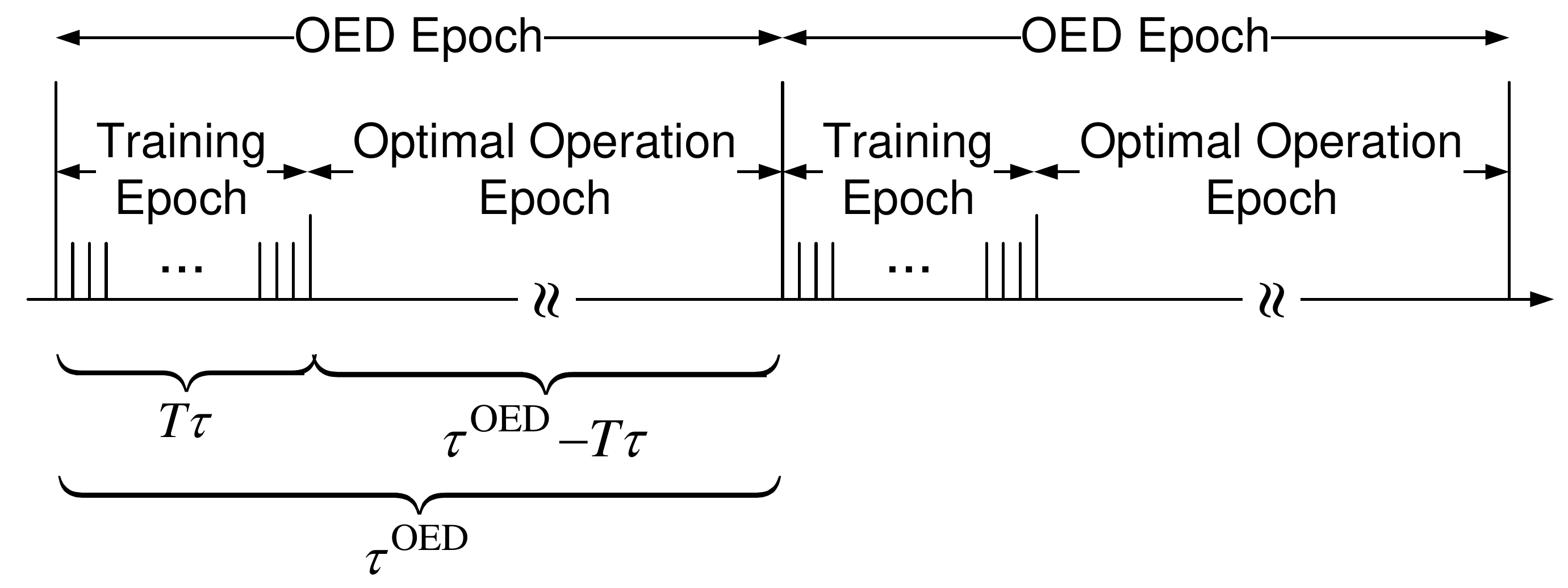}
\caption{Proposed decentralized OED organization with periodic training.}
\label{OEDtime}
\end{figure}

Knowing $\boldsymbol{\theta}$, specifically $\mathbf{g}$ and $d^{\star}$, is sufficient for implementing the decentralized program \eqref{dOED}.
We design a OED protocol in which the controllers utilize decentralized training and Algorithm~\ref{alg1} to acquire the information necessary to execute \eqref{dOED}.
Fig.~\ref{OEDtime} illustrates the temporal organization of the protocol.
The OED typically runs periodically, every $5-30$ minutes depending on the average rate of change of $\mathbf{g}$ and/or $\mathbf{d}$ \cite{3,6}. 
Therefore, we (i) divide the time axis into periodic \emph{OED epochs}, each of 
of duration $\tau^{\text{OED}}$, and (ii) assume that $\boldsymbol{\theta}$ changes independently at the beginning of and OED epoch and remains fixed throughout the epoch \cite{6}.
In each epoch, the DERs locally run the program \eqref{dOED} using up-to-date information about the generation capacities and load demands. 
To obtain this information, a fraction of the total duration $\tau^{\text{OED}}$ of the OED epoch is allocated for decentralized training, see Fig.~\ref{OEDtime}.
The OED epoch is split into a training epoch of duration $T\tau$ and an optimal operation epoch of total duration $\tau^{\text{OED}} - T\tau$.
In the training epoch, the DER controllers perform decentralized training and estimation as described in Sections~\ref{sec:problem} and \ref{sec:main1}. At the end of the training epoch, controller $n$ obtains $\hat{\boldsymbol{\theta}}_{-n}$, used at the beginning of the optimal operation epoch to determine the local dispatch policy $\hat{p}_n^*$, i.e., to determine which condition in \eqref{dOED} is satisfied locally.
Hence, each DER individually decides its primary control configuration via \eqref{dOED} using $\hat{\boldsymbol{\theta}}_{-n}$ and configures the local controller accordingly, forming the subsets $\hat{\mathcal{N}}^{\text{V}}\subset\mathcal{N}$ and $\hat{\mathcal{N}}^{\text{C}}\subset\mathcal{N}$.
We use $\hat{\cdot}$ to denote that \eqref{dOED} is solved using $\hat{\boldsymbol{\theta}}_{-n}$.

Implicit in the derivation of the decentralized program \eqref{dOED} is the assumption that the MG is balanced $d^{\star} \leq \sum_{m\in{\mathcal{N}}^{\text{C}}\cup{\mathcal{N}}^{\text{V}}}g_m$.
However, the stochastic renewable generation might sometimes violate the balance condition.
Moreover, due to estimation errors in $\hat{\boldsymbol{\theta}}_{-n}$, the resulting dispatch policies $\hat{{p}}_n^*$ will in general differ from ${p}_n^*$, attainable only when $\boldsymbol{\theta}$ is known perfectly; hence, $\hat{\mathcal{N}}^{\text{C/V}}\neq\mathcal{N}^{\text{C/V}}$ in general.
This leads to slightly suboptimal MG operation, but it might also 
violate the balance condition even when $\mathcal{N}^{\text{C/V}}$ satisfy it.
This results in loss of voltage regulation as the bus voltage quickly (i) drops towards the lower margin $v_{\min}$ when the net load demand is positive 
$d^{\star}>\sum_{m\in\hat{\mathcal{N}}^{\text{C}}\cup\hat{\mathcal{N}}^{\text{V}}}g_m$ or (ii) rises towards the upper margin $v_{\max}$ when the net-load demand is negative $d^{\star}<\sum_{m\in\hat{\mathcal{N}}^{\text{C}}}g_m$.
Clearly, additional generation/storage capacity is necessary to balance the remaining demand.
We employ a solution based on classical \emph{DC bus signaling}, where a \emph{backup source/storage} is activated if the bus voltage crosses certain thresholds \cite{14,15}.
The marginal costs of the backups are denoted with $c_{\text{source}}^{\text{extra}}/c_{\text{storage}}^{\text{extra}}$ per unit generated/stored power; these values are always larger than the largest marginal cost among the DERs in $\mathcal{N}$, i.e., $c_{\text{source}}^{\text{extra}}/c_{\text{storage}}^{\text{extra}}> c_N$.
In normal operating conditions, the MG is balanced, the backups are not active, and the bus voltage is regulated by the DERs in $\hat{\mathcal{N}}^{\text{V}}$, using the droop control law \eqref{eq:droop_gen} with parameters:
\begin{equation}
x_n = (1+\xi)x,\;\Delta v_n = 2\xi x,\;n\in\hat{\mathcal{N}}^{\text{V}},
\end{equation}
dimensioned to maintain the bus voltages in a tight region around the rated voltage $x$, i.e., in the interval $[(1-\xi)x,(1+\xi)x]$ with $\xi$ being a small positive number.
If the bus voltage drops below $(1-\xi)x$, it signals power deficit and the backup source is activated and configured in droop-controlled VSC mode, using \eqref{eq:droop_gen} with parameters set as:
\begin{equation}
x_{\text{source}}^{\text{extra}} = (1 + \xi)x,\;\Delta v_{\text{source}}^{\text{extra}} = (1 - \xi)x - v_{\min},
\end{equation}
maintaining the bus voltages in $[v_{\min},(1-\xi)x]$.
Conversely, if the voltage rises above $(1+\xi)x$, it signals power surplus and the storage is activated and also configured in droop-controlled VSC mode, using \eqref{eq:droop_gen} with parameters set as:
\begin{equation}
x_{\text{storage}}^{\text{extra}} = v_{\max},\;\Delta v_{\text{storage}}^{\text{extra}} = v_{\max} - (1 + \xi)x,
\end{equation}
maintaining the bus voltages in $[(1+\xi)x,v_{\max}]$.
Fig.~\ref{VIcomplete} summarizes the complete operational dynamics of the proposed system on a single $v-i$ diagram.
Note that installing backup generation/storage is standard practice when dimensioning standalone systems \cite{3,14,15}.
In grid-connected systems, the grid can be used as backup, effectively acting as ideal voltage source with infinite generation/storage capacity \cite{3}.

\begin{figure}[t]
\centering
\includegraphics[scale=0.35]{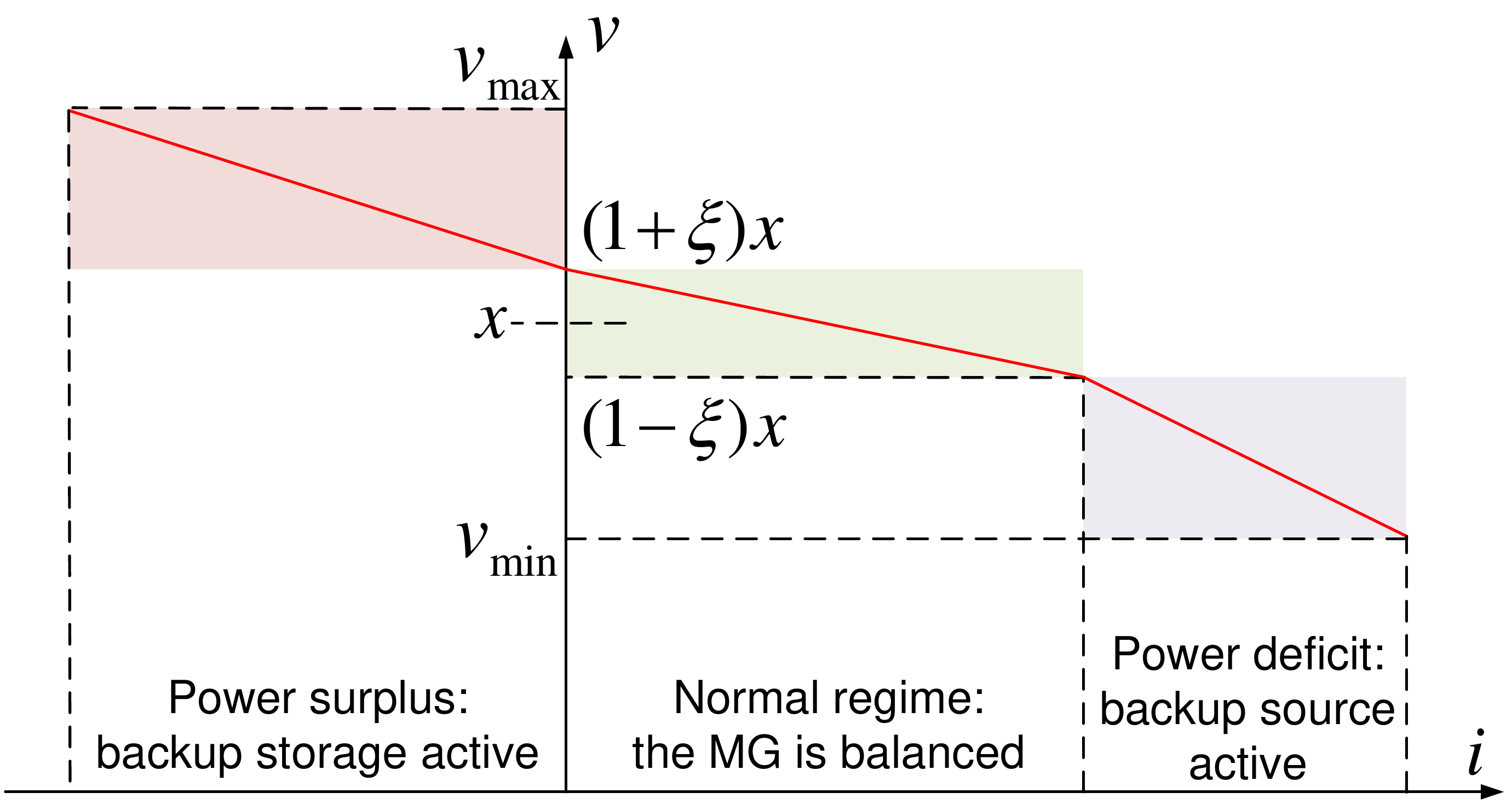}
\caption{Aggregate $v-i$ diagram of the proposed DC MG.}
\label{VIcomplete}
\end{figure}

\section{Evaluation}
\label{sec:Results}

\begin{table}
\renewcommand{\arraystretch}{1.15}
\caption{Fixed Simulation Parameters}
\label{table:table2}
\centering
\begin{tabular}{|c|c|}
\hline
Parameter & Value \\
\hline\hline
Simulation platform & MATLAB \\\hline
Reference voltage $x$ (volts) & $400$ \\\hline
Lower and upper voltage margins $v_{\min}$, $v_{\max}$ (volts) & $385$, $415$ \\\hline
Distribution network topology & cut-ring \\\hline
Max. gen. capacity per DER $g$ (kW) & $1$ \\\hline
Max. const. conductance demand per bus $d^{\text{ca}}$ (kW) & $0.2$ \\\hline
Max. const. current demand per bus $d^{\text{cc}}$ (kW) & $0.2$ \\\hline
Max. const. power demand per bus $d^{\text{cp}}$ (kW) & $0$ \\\hline
Average conductance per line $y$ (S) & $1$ \\\hline
Sampling frequency $\phi_S$ (kHz) & $50$ \\\hline
Sampling noise standard dev. $\sigma_S$ (volts/sample) & $0.1$ \\\hline
Transient time duration $\tau^{\text{transit}}$ (ms) & $2.5$ \\\hline
Nominal droop control params. $\tilde{x}$, $\Delta\tilde{v}$ (volts) & $400$, $15$ \\\hline
Max. voltage drop in $M$-phase $\Delta v$ (volts) & $15$ \\\hline
Total number of slots in the training epoch & $T=600$ \\\hline
Total number of slots in sub-phase $\alpha$ & $T_{\alpha} = 2N$ \\\hline
Total number of slots per block in sub-phase $\beta$ & $L=2N$ \\\hline
Other $C$-phase params. $\kappa^{\alpha}$, $\kappa^{\beta}$, $\chi_n,n\in\mathcal{N}$ & $1$, $1$, $\overline{v}_n$ \\\hline
OED epoch duration $\tau^{\text{OED}}$ (s) & $300$ \\\hline 
Backup gen./storage cost $c_{\text{source}}^{\text{extra}}/c_{\text{storage}}^{\text{extra}}$ (units/W) & $12$ \\\hline
DC bus signaling threshold $\xi$ & $6.25\cdot 10^{-4}$ \\\hline
\end{tabular}
\end{table}

\subsection{General Simulation Description and Design Parameters}

Table~\ref{table:table2} summarizes the numerical values of the simulation parameters that remain fixed in all simulation studies; the values of the remaining parameters are provided in the captions of the respective plots.
We consider a line, i.e., cut-ring distribution network topology, where all buses are connected to two other buses except for buses $n=1$ and $n=N$ that are connected to a single bus each.
As it is a regular practice for any power system, the MG is dimensioned to operate over a range of load demands.
For simplicity, we use $d_n^{\text{c}\cdot}\leq d^{\text{c}\cdot}$ for any $n\in\mathcal{N}$ (``$\cdot$'' stands for either ``a'', ``c'' or ``p''); similarly, $g_n\leq g$ for any $n\in\mathcal{N}$, see Table~\ref{table:table2}.

The measurement noise variance $\sigma^2$ after averaging $\phi_S(\tau - \tau^{\text{transit}})$ samples per slot, see Fig.~\ref{TimeAxis}, can be computed as:
\begin{equation}\label{eq:noise_calculation}
\sigma^2 = \frac{\sigma_S^2}{\phi_S(\tau - \tau^{\text{transit}})},
\end{equation}
where $\sigma_S^2$ is the noise variance of the PECs' ADCs \cite{rev11}.

The number of slots $\overline{T}$ in the $M$-phase for fixed $T^{\alpha}=2N$ and $L=2N$, see Table~\ref{table:table2}, is determined from the total number of slots $T=(1+L)\overline{T}+T^{\alpha}$ which is also fixed:
\begin{equation}
\overline{T} = \left\lfloor{\frac{T-2N}{1+2N}}\right\rfloor.
\end{equation}
The perturbation signals are set as (see also Fig.~\ref{VItraining}):
\begin{align}
x_n(t) & = x + \sqrt{\pi}\Delta x_n(t),\\\label{eq:droop_slopes_Mphase}
s_n(t) & = (\Delta v (x_n(t) - x + \Delta v))^{-1},~t\in\overline{\mathcal{T}},~n\in\mathcal{N}.
\end{align}
The binary sequences $\Delta x_n(t)\in\left\{-1,+1\right\},~t\in\overline{\mathcal{T}}$ are formed by tossing a fair coin for any $n\in\mathcal{N}$.
This is done a priori, i.e., $N$ binary Bernoulli sequences of length $\overline{T}$ are generated, confirmed to satisfy \eqref{eq:suff_cond1}, \eqref{eq:suff_cond2} and stored.
The droop slope perturbation laws \eqref{eq:droop_slopes_Mphase} ensure that the bus voltages will not drop below $x-\Delta v\geq v_{\min}$ or rise above $x + \Delta v\leq v_{\min}$ as long as $\sqrt{\pi}<\Delta v$, see Fig.~\ref{VItraining}.
The reference voltage training sequences in sub-phase $\alpha$ and block $b$ in sub-phase $\beta$ have fixed length of $2N$ slots and are set as:
\begin{align}
\Delta\mathbf{x}_n^{\alpha/\beta;b}=\mathbf{e}_n\otimes\begin{bmatrix}1\\-1\end{bmatrix},~b\in\overline{\mathcal{T}},~n\in\mathcal{N}.
\end{align}
Hence, $\delta^{\alpha} = \delta^{\beta} = 2$.
We also fix $\sqrt{\pi^{\alpha}} = \kappa^{\alpha}\sqrt{{\pi}}$ and $\sqrt{\pi^{\beta}} = \kappa^{\beta}\sqrt{\pi}$, where $0<\kappa^{\alpha},\kappa^{\beta}\leq 1$ are set to keep the reference voltage deviation amplitudes in the $C$-phase relatively small, ensuring that the model \eqref{eq:lin_approx_time} is valid for any $\sqrt{\pi}\in(0,\Delta v)$.

\begin{figure}[t]
\centering
\includegraphics[scale=0.35]{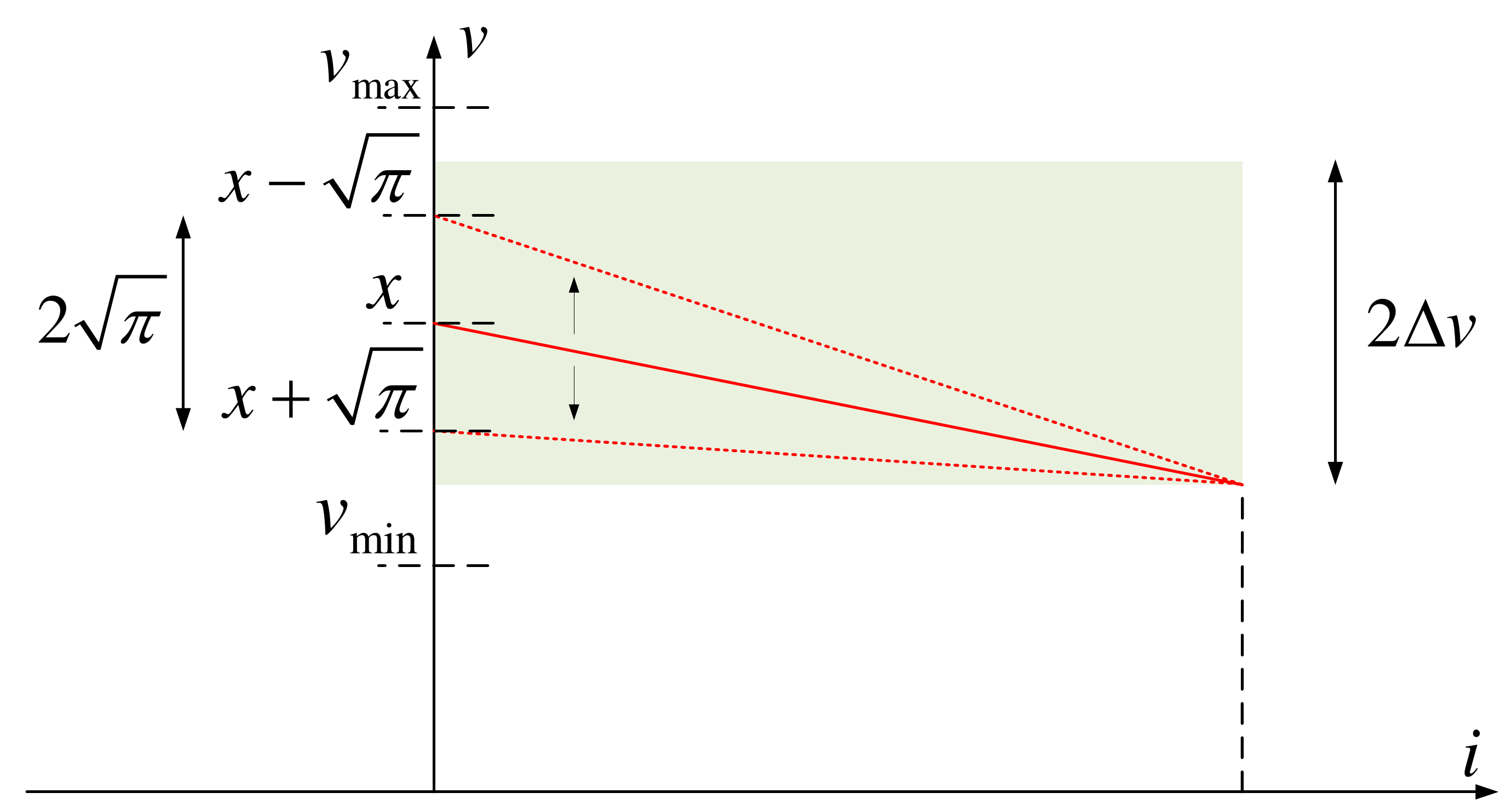}
\caption{$v-i$ diagram of the $M$-phase droop control perturbation law.}
\label{VItraining}
\end{figure}

The performance of J-SISE w.r.t. the MSE and the performance of DOED w.r.t. the cost, are determined by the configuration of the training epoch, which in turn is determined by variety of factors such as slot duration, number of slots, nominal droop control parameters, training matrices and deviation amplitudes.  
With all specifications listed above and in Table~\ref{table:table2}, most of these factors are kept fixed in our evaluations and the design parameters of the training epoch are the slot duration $\tau$ and the reference voltage deviation amplitude $\sqrt{\pi}$. Next, we evaluate the performance of the J-SISE in terms of the design parameters and show how to find their optimal values w.r.t. DOED.

\subsection{J-SISE Performance}
\label{sec:jmle_perf}

First, we investigate the performance, the scalability and the convergence properties of Algorithm~\ref{alg1} w.r.t. $\boldsymbol{\theta}_{-n}$ from the perspective of controller $n=1$ and compare it against CRLB.
We fix the generation capacities of all DERs to have equal values, i.e., $g_n=g,~n\in\mathcal{N}$ and we do the same with the load components $d_n^{\text{ca}} = d_n^{\text{ca}},~d_n^{\text{cc}} = d_n^{\text{cc}},~d_n^{\text{cp}} = d_n^{\text{cp}}$ and the line conductances $y_{n,m} = y$ for all $n,m\in\mathcal{N}$.
We use the \emph{Relative Root Mean Squared Error (RRMSE)} metric, derived from the MSE matrix as follows:
\begin{equation}\label{eq:rrmse}
\text{RRMSE}(\hat{\cdot})=\frac{\sqrt{\mathsf{trace}(\text{MSE}(\hat{\cdot}))}}{\|\cdot\|_2}.
\end{equation}
To evaluate the MSE matrix, we use statistical average of individual MSE matrices, obtained for $1000$ different realizations of the noise matrix $\mathbf{Z}$.
``$\cdot$'' in the above definition stands for either the full vector $\boldsymbol{\theta}_{-n}$ or its constituent vectors, i.e., $\mathbf{g}_{-n}$, $\mathbf{d}$ or $\boldsymbol{\psi}$; in either case, the RRMSE is interpreted as the standard deviation of the estimation error per component of the vector that is used as argument.
Note that, when applied to a constituent vector of $\boldsymbol{\theta}_{-n}$, we plug the diagonal block of the MSE matrix corresponding to that particular constituent vector.
To compute the corresponding lower bound on the RRMSE, we use the CRLB matrix in \eqref{eq:rrmse} instead of the MSE matrix.

We focus particularly on the RRMSE as function of $\sqrt{\pi}$, since RRMSE decreases linearly with ${\tau}$ in the log-domain, see eq. \eqref{eq:noise_calculation}).
Fig.~\ref{results1} depicts the performance of J-SISE for each of the constituent vectors of $\boldsymbol{\theta}_{-n}$, i.e., $\mathbf{g}_{-n}$, $\mathbf{d}$ and $\boldsymbol{\psi}$ against the corresponding lower bounds, for $N=6$ DERs. We have evaluated the lower bounds using both, the constrained CRLB \eqref{eq:CRLB_standard} and expression \eqref{eq:CRLB_theta} from Proposition~\ref{prop5}, and they both yield numerically identical results.
Empty markers correspond to the initial estimate that initializes Algorithm~\ref{alg1}, obtained via \eqref{eq:init_theta}, while filled markers correspond to $\hat{\boldsymbol{\theta}}_{-n}$ after Algorithm~\ref{alg1} converges.
As expected, J-SISE is efficient and attains the CRLB as $\sqrt{\pi}$ increases, except for values very close to $\Delta v$; here, the RRMSE hits a turning point, after which it increases sharply as a result of the fact that when $\sqrt{\pi}\rightarrow\Delta v$, the droop slope $s_n(t)$ grows arbitrarily large and the virtual resistance $y_n^{\text{va}}\rightarrow 0$.
Hence, the controller starts to behave as an ideal voltage source with infinite capacity, pushing the bus voltages to a fixed value $x-\Delta v$ and making the MG insusceptible to reference voltage perturbations.

We further observe that the generation capacities, Fig.~\ref{results1}\subref{results1a}, and the line conductances, Fig.~\ref{results1}\subref{results1c}, can be identified with very high precision (less than $1\%$ of the true value).
In contrast, the RRMSE of the load demands of individual components, Fig.~\ref{results1}\subref{results1b}, is several orders of magnitude higher.
We conclude that, identifying the individual components of the loads with satisfactory performance might require excessive (even prohibitive) training epoch durations to suppress the noise.
However, in many upper layer applications, detailed knowledge on the individual load component demands is not necessary and knowing only the total bus demand $d_n^{\star}=d_n^{\text{ca}}+d_n^{\text{cc}}+d_n^{\text{cp}}$ is sufficient \cite{6,last}; in such case, an estimate of the total load demand vector $\mathbf{d}^{\star} = [d_1^{\star},\hdots,d_N^{\star}]^{\mathsf{T}}$, comprising the total demands at each bus, can be obtained from $\hat{\mathbf{d}}$ via $\hat{\mathbf{d}}^{\star} = [\mathbf{I}_N,\mathbf{I}_N,\mathbf{I}_N]\hat{\mathbf{d}}$.
Fig.~\ref{results1}\subref{results1b} shows that $\hat{\mathbf{d}}^{\star}$ can be identified with a precision comparable to the one achieved for the generation capacities and line conductances.

\begin{figure*}[t]
\centering
\subfloat[$\mathbf{g}_{-n}$]{\includegraphics[scale=0.61]{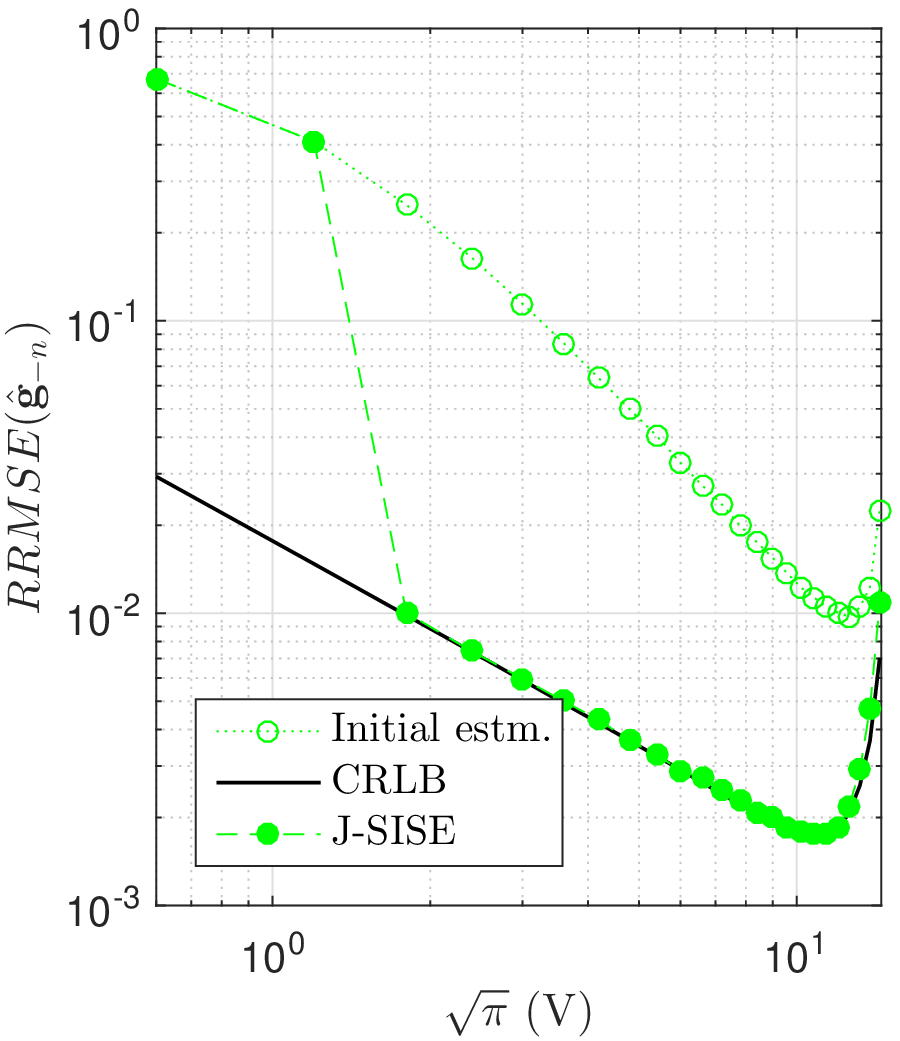}\label{results1a}}
\hfil
\subfloat[$\mathbf{d}$ and $\mathbf{d}^{\star}$]{\includegraphics[scale=0.61]{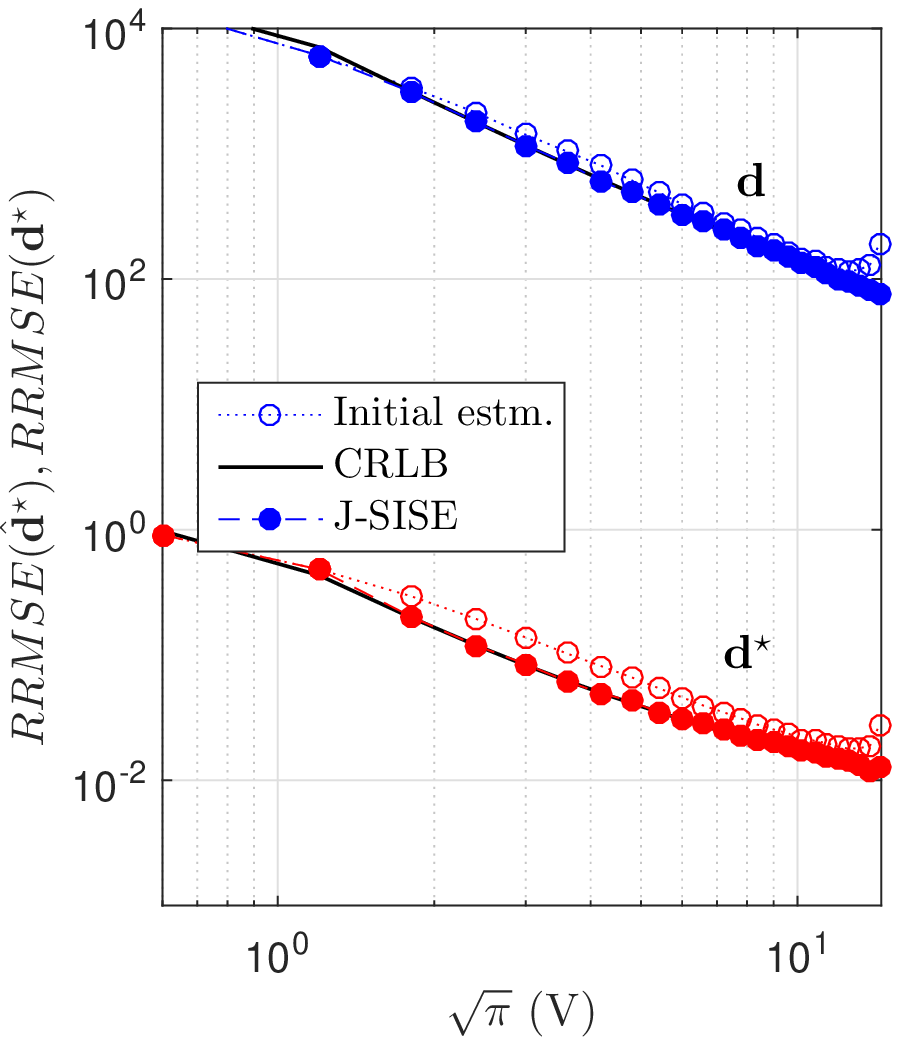}\label{results1b}}
\hfil
\subfloat[$\mathbf{\psi}$]{\includegraphics[scale=0.61]{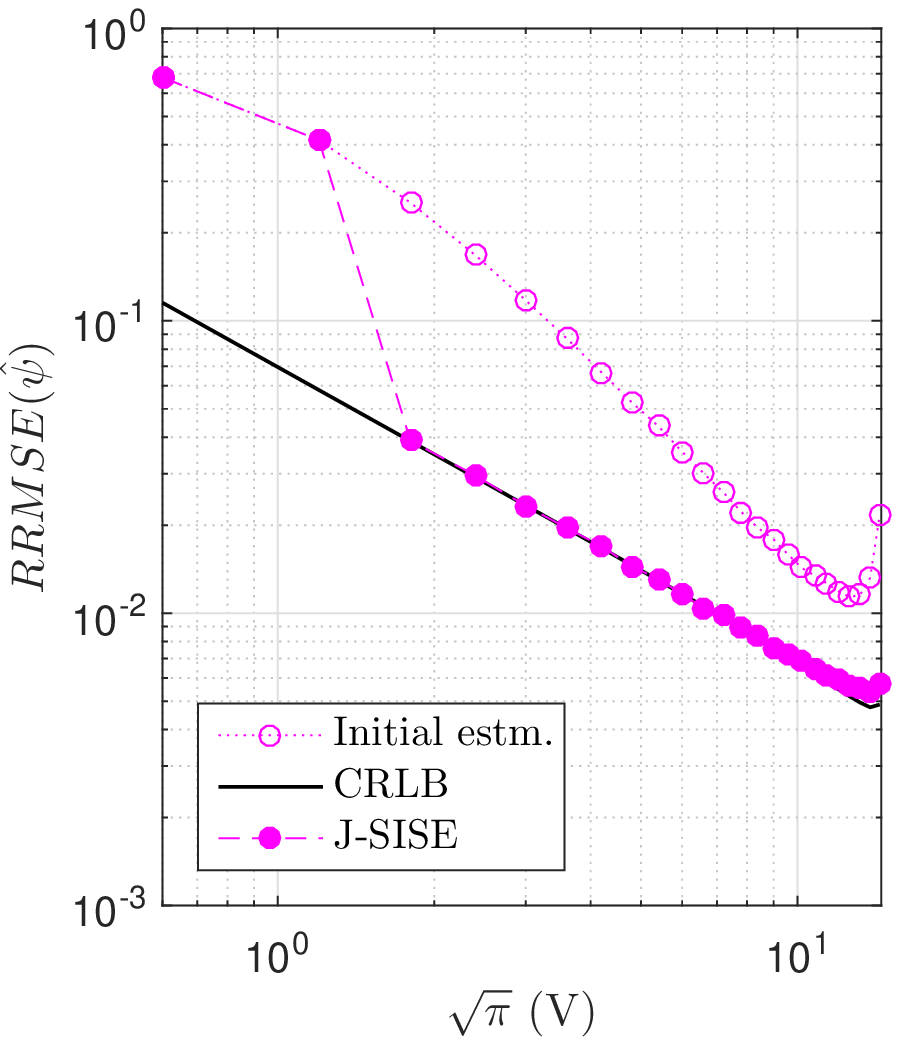}\label{results1c}}
\caption{Standard deviation of the estimation error per component of the vectors $\mathbf{g}_{-n}$, $\mathbf{d}$, $\mathbf{d}^{\star}$ and $\boldsymbol{\psi}$ in a cut-ring MG with $N=6$ DERs and slot duration $\tau=50$ ms.}
\label{results1}
\end{figure*}

\begin{figure}[t]
\centering
\includegraphics[scale=0.61]{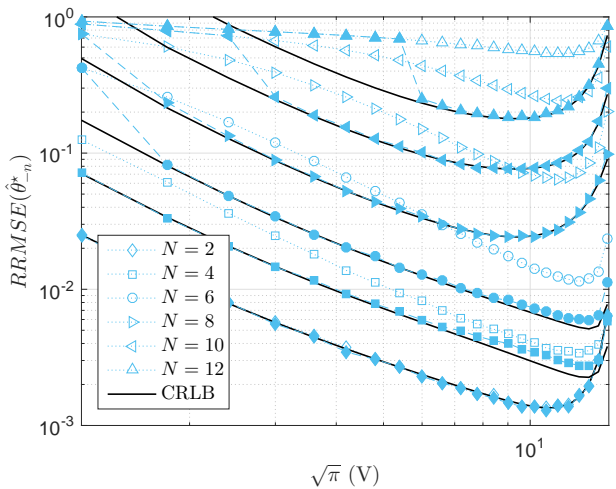}
\caption{Scalability performance per component of the parameter vector ${\boldsymbol{\theta}}_{-n}^{\star} = [\mathbf{g}^{\mathsf{T}},(\mathbf{d}^{\star})^{\mathsf{T}},\boldsymbol{\psi}^{\mathsf{T}}]$ in cut-ring MG with increasing number of DERs and slot duration $\tau=50$ ms.}
\label{results2}
\end{figure}

The improvement of $\hat{\boldsymbol{\theta}}_{-n}$ w.r.t. $\boldsymbol{\theta}_{-n}^{(0)}$ given with \eqref{eq:init_theta}, is also evident, clearly showing that the initial estimate is not efficient.
The numerical results (not shown here due to space limitations) show that the average of $\boldsymbol{\theta}_{-n}^{(0)} - \hat{\boldsymbol{\theta}}_{-n}$ converges to zero vector asymptotically.
We conclude that the initial estimate $\boldsymbol{\theta}_{-n}^{(0)}$ is indeed unbiased estimator of $\boldsymbol{\theta}_{-n}$ and can be still used in practice even though it is not efficient, particularly, when $\sqrt{\pi}$ is of the same order as/smaller than $\sigma$ or for small $N$.
In the first case, Algorithm~\ref{alg1} does not converge, see Fig.~\ref{results1}, and $\boldsymbol{\theta}_{-n}^{(0)}$ remains as the only reasonable choice.
The second case can be more clearly observed in Fig.~\ref{results2} that investigates the performance of the framework for increasing number of buses; we see that for small number of buses (e.g. $N=2$), the RRMSE of the initial estimate approaches the CRLB; in such case, the gain from applying Algorithm~\ref{alg1} is marginal, and $\boldsymbol{\theta}_{-n}^{(0)}$ is sufficient for all practical purposes.

From Fig.~\ref{results2}, we also observe that, the performance of J-SISE tends to deteriorate as the number of buses increases, which is expected due to the increase of $\mathsf{dim}(\boldsymbol{\theta}_{-n})$.
A straightforward way to improve the performance of Algorithm~\ref{alg1} and make the estimation error arbitrarily small for large $N$, is to increase $\tau$.
However, note that \eqref{eq:MLE_general_SB} treats the vector $\boldsymbol{\psi}$ as full vector, when in fact it may be sparse, containing many zero entries. 
This might prove to be problematic as the size of the MG scales, i.e., as the number of buses increases since larger distribution systems are significantly sparser \cite{rev9}, so estimating $\boldsymbol{\psi}$ as if it is full vector might lead to performance degradation \cite{rev10}.
So, an appropriate way to improve the performance when $N$ is large (which is out of the scope of this work) is to modify \eqref{eq:MLE_general_SB} by adding sparsity constraint on $\boldsymbol{\psi}$ and apply a common relaxation method \cite{rev10}.

Finally, we comment on the convergence speed of Algorithm~\ref{alg1}; in all tested cases, that is for $N\leq 12$, Algorithm~\ref{alg1} converges already after $10$ iterations.
This remarkable result can be mainly attributed to the fact that the initial estimates $\boldsymbol{\chi}^{(0)}$, $\boldsymbol{\theta}_{-n}^{(0)}$, given with eq. \eqref{eq:init_v}, \eqref{eq:init_theta}, respectively, form a stationary point of the optimization problem \eqref{eq:MLE_general_SB} (see subsection~\ref{sec:JMLE}).
The additional fact that they are also (asymptotically) unbiased, implies that $\boldsymbol{\theta}_{-n}^{(0)}$ must lie in a neighborhood around $\hat{\boldsymbol{\theta}}_{-n}$, possibly being an inflection point from which it can easily converge to the global optimum only after several iterations. 

\subsection{Optimizing the Cost Trade-off in DOED}
\label{sec:cost_tradeoff}

The results presented in the previous subsection do not consider (i) the effect that the estimation error has on the upper layer applications, and (ii) the effect that the power dissipation during training has on the overall performance of the MG.
In other words, improving the performance of J-SISE, which is desirable from the perspective of the upper layer application, comes at the ``price'' of increased power dissipation during training, either by using large perturbation amplitudes or long slot durations, which in turn compromises the performance of the upper layer control application.
This leads to a fundamental trade-off between the performance of J-SISE, which is determined by the configuration of the training epoch, and the performance of the application.
Our goal is to (i) show how to characterize this trade-off via \emph{utility function} that jointly captures the performances of J-SISE and the upper layer application, and (ii) provide guidelines on how to design optimal training epochs, namely, how to choose $\tau$ and $\sqrt{\pi}$ such that the utility function is optimized.

As a case study, we take the DOED protocol, described in subsection~\ref{sec:DOED}, noting that the approach described below can be applied to any upper layer application.
The performance of specific DOED policy vectors $\mathbf{p}$ is assessed via the cost $c(\mathbf{p})=\mathbf{a}^{\mathsf{T}}\mathbf{p}$. 
The cost of the optimal policy $\mathbf{p}^{*}$ is $c^* = c(\mathbf{p}^*) + c^{\text{extra}} = \mathbf{a}^{\mathsf{T}}\mathbf{p}^* + c^{\text{extra}}$ with $c^{\text{extra}}$ denoting any extra cost entailed by activating backups in case the MG is unbalanced.
$c^*$ is in fact the minimal cost, attainable only when $\mathbf{g}$ and ${d}^{\star}$ are \emph{perfectly known} to each controller.
However, when running the DOED protocol using the estimated parameter vector, see subsection~\ref{sec:DOED}, the cost of the resulting dispatch policy vector $\hat{\mathbf{p}}^*$ should also account for (i) the fact that $\hat{\mathbf{p}}^*\neq\mathbf{p}^*$, i.e., the DOED policy $\hat{\mathbf{p}}^*$ is, in general, suboptimal, (ii) the fact that $\hat{\mathbf{p}}^*$ is valid only in the optimal operation epoch within the OED epoch, and (iii) the power dissipation incurred in the training epoch.
We denote this cost with $\hat{c}^*$ and we write:
\begin{align}\label{eq:OED_cost_total}
\hat{c}^* = \frac{\tau}{\tau^{\text{OED}}}\mathbf{1}_T^{\mathsf{T}}\mathbf{P}\mathbf{a} + \frac{\tau^{\text{OED}} - T\tau}{\tau^{\text{OED}}}\left(\mathbf{a}^{\mathsf{T}}\hat{\mathbf{p}}^* + \hat{c}^{\text{extra}}\right),
\end{align}
where the $T\times N$ matrix $\mathbf{P}$ is defined as $[\mathbf{P}]_{t,n}=p_n(t),~n\in\mathcal{N},~t\in\mathcal{T}$ and $p_n(t)$ is the output power of DER $n$ in slot $t$.
The first term corresponds to the cost of training, whereas the second gives the actual cost of $\hat{\mathbf{p}}^*$.
We define the \emph{Relative Cost Increase (RCI)} $\hat{\mu}$, relative to the optimal cost $c^*$:
\begin{equation}\label{eq:ave_RCI}
\hat{\mu} = \frac{\tau}{\tau^{\text{OED}}}\frac{\mathbf{1}_T^{\mathsf{T}}\mathbf{P}\mathbf{a}}{\mathbf{a}^{\mathsf{T}}\mathbf{p}^* + c^{\text{extra}}} + \frac{\tau^{\text{OED}} - T\tau}{\tau^{\text{OED}}}\frac{\mathbf{a}^{\mathsf{T}}\hat{\mathbf{p}}^* + \hat{c}^{\text{extra}}}{\mathbf{a}^{\mathsf{T}}\mathbf{p}^* + c^{\text{extra}}} - 1.
\end{equation}
The RCI can be interpreted as a measure of the additional monetary charge that the the community served by the MG will be subjected to when operating autonomously using the proposed DOED protocol, without any access to external communication enabler.
We observe that $\hat{\mu}$ is a random variable whose pdf is parametrized w.r.t fixed $\boldsymbol{\theta}$.
In practice, it is desirable to optimize the performance of the upper layer application over the range of $\boldsymbol{\theta}$, which the MG is foreseen to operate in.
Therefore, we choose the \emph{average RCI}, denoted by ${\mu}$ and computed as an average of $\hat{\mu}$ over $\boldsymbol{\theta}$, to be the utility function for the DOED. 
The aim is to find the optimal training epoch configuration parameters, namely, $\tau$ and $\sqrt{\pi}$, that minimize the average RCI:
\begin{align}\label{eq:DOED_optimization}
  \tau^*,\sqrt{\pi^*} & = \min_{\tau,\sqrt{\pi}}\mu(\tau,\sqrt{\pi}),\\\nonumber
\text{s.t. } & \tau^{\text{transit}}<\tau\leq\tau_{\max},~0<\sqrt{\pi}<\Delta v.
\end{align}
Computing $\mu$ in closed form is far from trivial; therefore, we resort to Monte-Carlo simulation, run the DOED protocol for $100000$ different values of $\boldsymbol{\theta}$ and use the statistical average of the individual RCIs as an estimate of $\mu$.
In each trial, $\boldsymbol{\theta}$ is generated independently from the uniform distribution, i.e., $\mathbf{g}\in\mathsf{Unif}[\mathbf{0}_N,g\mathbf{1}_N]$, $\mathbf{d}^{\text{c}\cdot}\in\mathsf{Unif}[\mathbf{0}_N,d^{\text{c}\cdot}\mathbf{1}_N]$, where $g$, $d^{\text{c}\cdot}$ are given in Table~\ref{table:table2}; note that we keep the line conductances fixed to $y$ as the topology changes very infrequently compared to the generation and the load.

Rewriting $p_n(t)=\tilde{p}_n+\Delta p_n(t)$, where $\tilde{p}_n$ is the output power of DER $n$ corresponding to the nominal droop parameters, the time average of the power dissipation $\sum_{t\in\mathcal{T}}\Delta p_n(t)\approx 0$.
We conclude that with \eqref{eq:ave_RCI} and linear OED cost function, it is difficult to asses the impact of power dissipation during training.
Therefore, we introduce a \emph{quadratically-modified} RCI (QRCI), denoted with $\hat{\eta}$:
\begin{equation}\label{eq:eta}
\hat{\eta} = \hat{\mu} + q\mathbf{1}_T^{\mathsf{T}}\mathbf{Q}\mathbf{1}_N,
\end{equation}
where the $T\times N$ matrix $\mathbf{Q}$ is defined as $[\mathbf{Q}]_{t,n} = (p_n(t) - \tilde{p}_n)^2,~n\in\mathcal{N},~t\in\mathcal{T}$ and $0<q\leq \frac{\tau}{\tau^{\text{OED}}c^*}=q_{\max}$.
In similar way as $\mu$, we define the average QRCI, denoted with $\eta$ and restate the optimization problem \eqref{eq:DOED_optimization} with $\eta$ as utility function.

\begin{figure*}[t]
\centering
\subfloat[Average RCI.]{\includegraphics[scale=0.61]{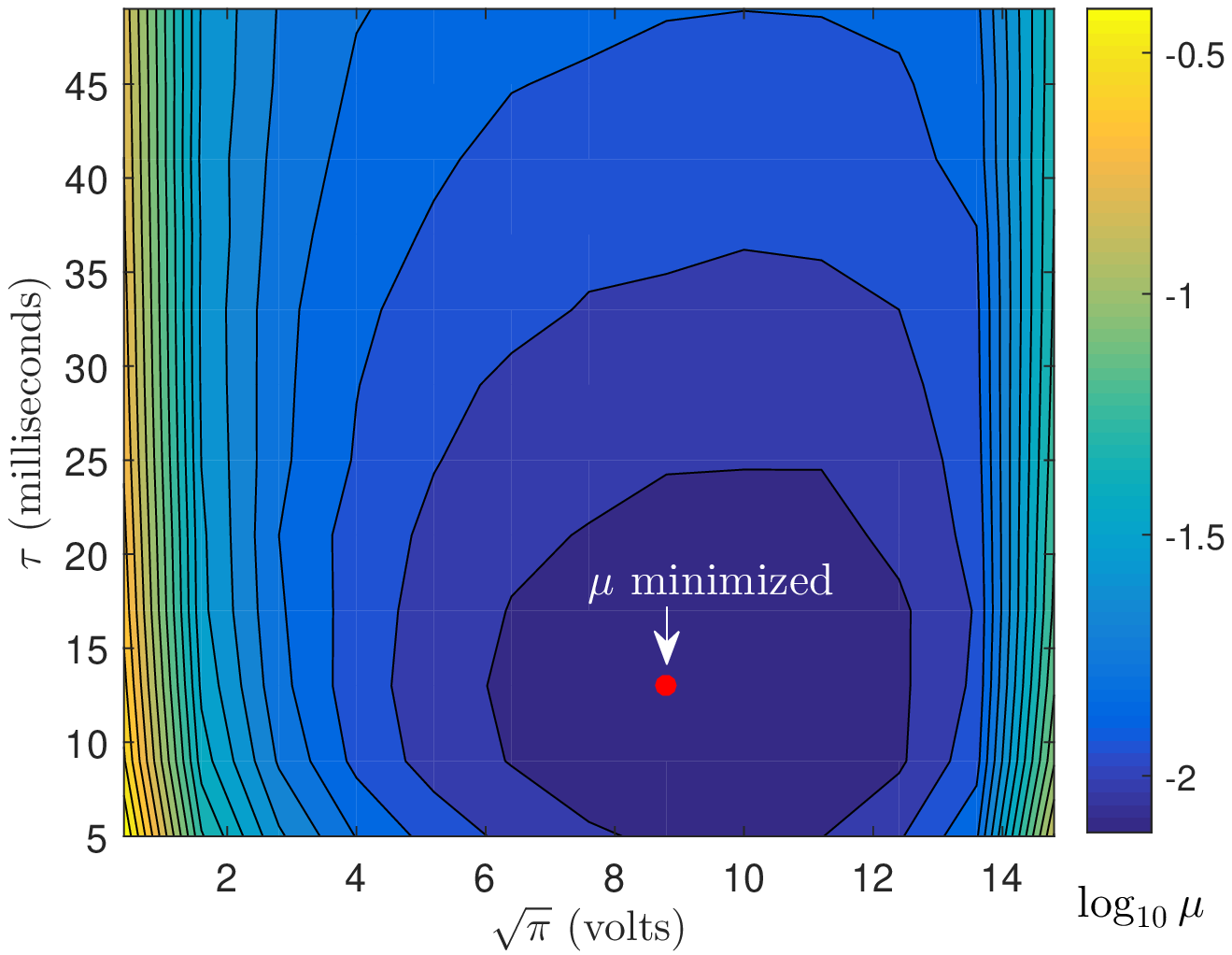}\label{Lambda}}
\hfil
\subfloat[Average QRCI.]{\includegraphics[scale=0.61]{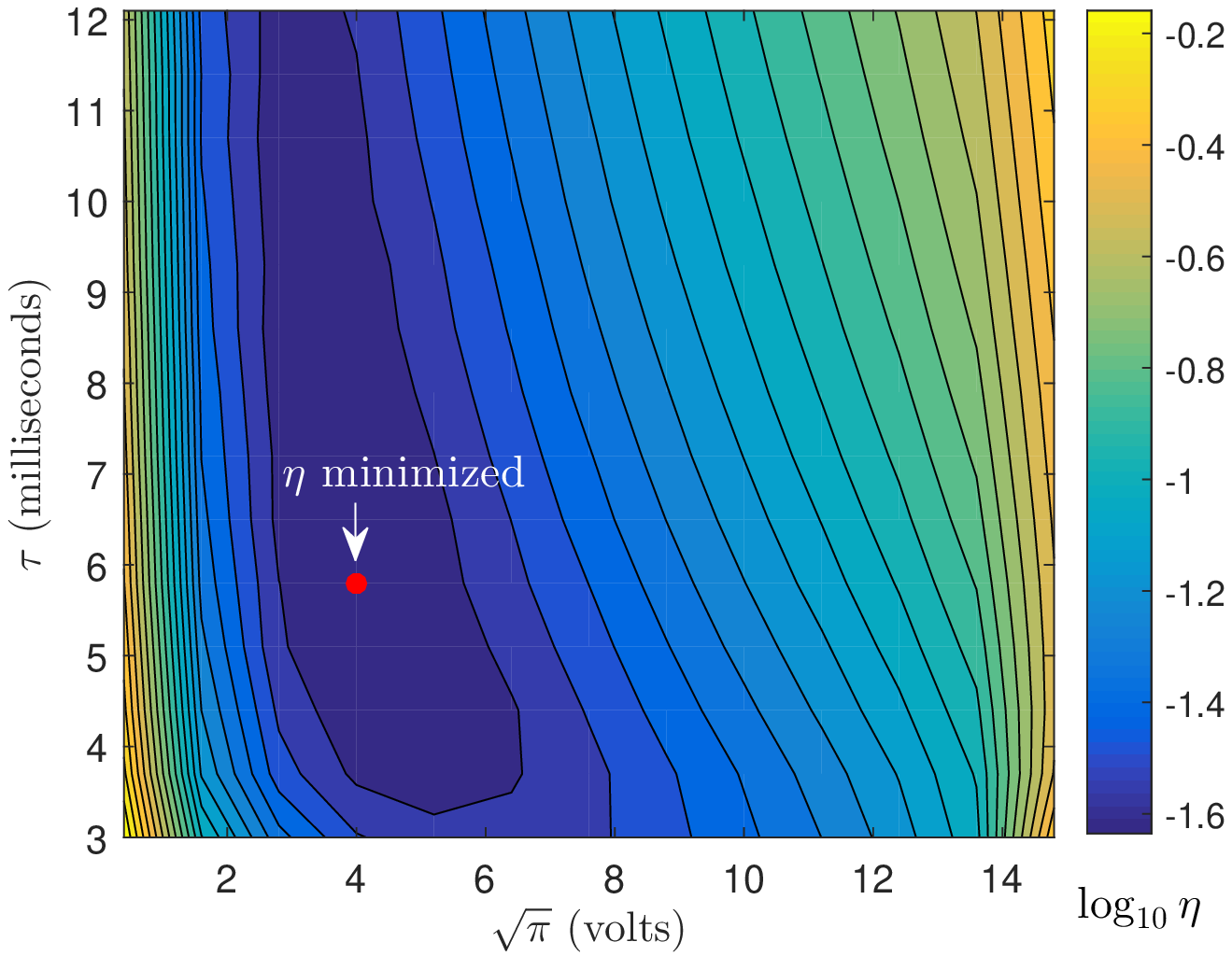}\label{qLambda}}
\caption{Finding the optimal training epoch parameters $\tau$ and $\sqrt{\pi}$ that minimize the DOED utility function in DC MG with $N=6$ DERs and marginal costs $\mathbf{a} = [3,3,5,5,8,11]^{\mathsf{T}}$ units/W, $q=q_{\max}$. The utility functions, i.e., the average RCI \eqref{eq:ave_RCI} and its quadratically modified variant \eqref{eq:eta} that penalizes power dissipation during training more heavily, quantify the additional monetary cost (expressed in log-domain in the figure) entailed by the proposed solution relative to the cost of the optimal dispatch policy. It is interesting to note that both, the average RCI and the QRCI are below $5\%$ (in log-domain below $\approx -1.3$) for a wide range of training epoch configurations.}
\label{Lambda_all}
\end{figure*}

The results are presented in Fig.~\ref{Lambda_all}.
We observe that within the investigated domain, the average RCI, see Fig.~\ref{Lambda} is a convex function of $\tau$ and $\sqrt{\pi}$.
Specifically, for fixed $\sqrt{\pi}$, $\mu$ decrease as $\tau$ increases due to the effect of noise suppression, see \eqref{eq:noise_calculation}.
In this regime, the duration of the training epoch is still very short relative to $\tau^{\text{OED}}$, such that the first term in \eqref{eq:OED_cost_total} is negligible and the RCI is dominated by the second term which decreases towards $c^*$ as the estimation error is reduced.
However, $\mu$ hits a turning point when $\tau$ and, consequently, the duration of the training epoch become long enough such that the first term in \eqref{eq:OED_cost_total} starts to dominate over the second; after this, it makes no sense to keep increasing $\tau$ as $\mu$ will also increase.
Conversely, for fixed $\tau$, $\mu$ decreases as $\sqrt{\pi}$ increases until it hits the turning point after which it starts to increase quickly; evidently, this is happening when we get very close to $\Delta v$.
As discussed in the previous subsection, the performance of J-SISE starts to deteriorate when $\sqrt{\pi}\rightarrow\Delta v$, pushing the second term in \eqref{eq:OED_cost_total} away from its lower bound $c^*$.
Hence, within the domain of interest, the average RCI for an MG, specified in Table~\ref{table:table2} and the caption of Fig.~\ref{Lambda}, is minimized when $\sqrt{\pi}\approx 8.8$ volts and $\tau\approx 13$ milliseconds.
The minimized average RCI is $\mu^*\approx 0.008$; in other words, the average increase of the cost is less than $1\%$ of the optimal cost $c^*$.
This increase, besides being completely tolerable by the OED \cite{3}, it is also comparable to the additional operating cost charges imposed by mobile operators when employing wireless cellular solution not including the cost of installing dedicated communication hardware \cite{3,6}.

Similarly as the average RCI, the average QRCI, see Fig.~\ref{qLambda}, is also a convex function of $\tau$ and $\sqrt{\pi}$ within the investigated domain with behaviour governed by the same reasoning we used on the average RCI.
However, the minimum this time moves closer to the down-left corner due to the second term in \eqref{eq:eta}.
Specifically, $\eta$ is minimized when $\sqrt{\pi}\approx 4$ volts and $\tau\approx 5.1$ milliseconds with average RCI $\mu\approx 0.015$, i.e., still around $1\%$ of $c^*$.

\section{Concluding Remarks}
\label{sec:conc}

We introduced autonomous system identification solution, based on temporary primary control perturbations and iterative ML-based algorithm for DC MGs and without access to an external communication system.
The method is implemented in a decentralized manner within the primary droop controllers of the PECs and enables the controllers to learn i) the generation capacities of power sources, ii) the load demands, and iii) distribution network topology using only local bus voltage measurements.
The key enabling tool is the decentralized training where the primary controllers inject small, amplitude-modulated training sequences that complete the rank of the estimation problem and enable regaining full system observability.
We evaluated the performance of the ML-based algorithm, showing that we can achieve high reliability in DC MGs of small to moderate size ($N\leq 12$).
Then, we showcased the potential of the solution in fully decentralized OED where the controllers perform training periodically and reconfigure according to the locally estimated information.
Last but not least, we illustrated an elaborate methodology for designing training epochs that optimize the operational cost of an autonomous DC MG.

Although we focused on DC MGs and we used several assumptions that simplified the developments, the same design principles introduced in this paper can be applied to any cyber-physical system with dual-layer control architecture that does not not have access to external communication resources, under broader circumstances.
Such investigations are part of our on-going and future work.

\appendices
\section{Proof of Proposition I}\label{app:propI}
The power balance condition in each slot states that:
\begin{equation}\label{TSP2017_eq:MBMG_power_balance_compact}
\omega_n(t) = 0,~n\in\mathcal{N},~t\in\mathcal{T}.
\end{equation}
Recall that during training all DERs are in droop-controlled VSC mode configured for proportional power sharing.
Hence $\zeta_n=1$ for any $n\in\mathcal{N}$.
In such case \eqref{TSP2017_eq:MBMG_power_balance_compact} can be rewritten as:
\begin{align}\label{TSP2017_eq:MBMG_power_balance_compact_proof}
\omega_n(t) & = v_n^2(t)\left( s_n(t)g_n + \frac{1}{x^2}d_n^{\text{ca}}\right) + v_n(t)\sum_{m\in\mathcal{N}}(v_n(t)-v_m(t))y_{n,m} - v_n(t)\left(x_n(t)s_n(t)g_n - \frac{1}{x}d_n^{\text{cc}}\right) + d_n^{\text{cp}}. 
\end{align}
Let $\boldsymbol{\omega}^t$ be defined as $[\boldsymbol{\omega}^t]_n=\omega_n(t)$; we get:
\begin{align}
\boldsymbol{\omega}^t & = \left(\mathsf{D}(\mathbf{g})\mathbf{s}^t + \frac{1}{x^2}\mathbf{d}^{\text{ca}}\right)\odot\mathbf{v}^t\odot\mathbf{v}^t + (\mathbf{Y}\mathbf{v}^t)\odot\mathbf{v}^t - \left(\mathsf{D}(\mathbf{g})(\mathbf{x}^t\odot\mathbf{s}^t) - \frac{1}{x}\mathbf{d}^{\text{cc}}\right)\odot\mathbf{v}^t + \mathbf{d}^{\text{cp}},
\end{align}
where $\mathbf{v}^{t}$, $\mathbf{x}^{t}$ and $\mathbf{v}^{t}$, defined as $[\mathbf{v}^{t}]_n = v_n(t)$, $[\mathbf{x}^{t}]_n = x_n(t)$ and $[\mathbf{s}^{t}]_n = s_n(t)$, represent the $t$-th rows of $\mathbf{V}$, $\mathbf{X}$ and $\mathbf{S}$, respectively.
Stacking $(\boldsymbol{\omega}^t)^{\mathsf{T}}$ vertically for each $t\in\mathcal{T}$, we get the power balance matrix $\mathbf{\Omega}$:
\begin{align}\nonumber
\mathbf{\Omega} & = 
\begin{bmatrix}
(\mathbf{s}^1)^{\mathsf{T}}\mathsf{D}(\mathbf{g}) + \frac{(\mathbf{d}^{\text{ca}})^{\mathsf{T}}}{x^2}\\
\vdots\\
(\mathbf{s}^T)^{\mathsf{T}}\mathsf{D}(\mathbf{g}) + \frac{(\mathbf{d}^{\text{ca}})^{\mathsf{T}}}{x^2}
\end{bmatrix}\odot
\begin{bmatrix}
(\mathbf{v}^1)^{\mathsf{T}}\\
\vdots\\
(\mathbf{v}^T)^{\mathsf{T}}
\end{bmatrix}\odot
\begin{bmatrix}
(\mathbf{v}^1)^{\mathsf{T}}\\
\vdots\\
(\mathbf{v}^T)^{\mathsf{T}}
\end{bmatrix} +
\begin{bmatrix}
(\mathbf{v}^1)^{\mathsf{T}}\mathbf{Y}\\
\vdots\\
(\mathbf{v}^T)^{\mathsf{T}}\mathbf{Y}
\end{bmatrix}\odot
\begin{bmatrix}
(\mathbf{v}^1)^{\mathsf{T}}\\
\vdots\\
(\mathbf{v}^T)^{\mathsf{T}}
\end{bmatrix}\\
& -
\begin{bmatrix}
((\mathbf{s}^1)^{\mathsf{T}}\odot(\mathbf{x}^1)^{\mathsf{T}})\mathsf{D}(\mathbf{g}) - \frac{(\mathbf{d}^{\text{cc}})^{\mathsf{T}}}{x}\\
\vdots\\
((\mathbf{s}^1)^{\mathsf{T}}\odot(\mathbf{x}^1)^{\mathsf{T}})\mathsf{D}(\mathbf{g}) - \frac{(\mathbf{d}^{\text{cc}})^{\mathsf{T}}}{x}
\end{bmatrix}\odot
\begin{bmatrix}
(\mathbf{v}^1)^{\mathsf{T}}\\
\vdots\\
(\mathbf{v}^T)^{\mathsf{T}}
\end{bmatrix} + 
\begin{bmatrix}
(\mathbf{d}^{\text{cp}})^{\mathsf{T}}\\
\vdots\\
(\mathbf{d}^{\text{cp}})^{\mathsf{T}}
\end{bmatrix} = \mathbf{0}_{T\times N},
\end{align}
yielding the compact form \eqref{eq:PowerBalanceMatrix} which completes the derivation.

\section{$\boldsymbol{\theta}_{-n}$ {is not identifiable when the system is not observable}}\label{app:nonident}
We consider the following situation: controller $n$ knows only $\mathbf{w}_n$ and knows $\mathbf{X}$ and $\mathbf{S}$ \emph{completely}.
In other words, the controllers \emph{do not} exchange any local steady state voltage measurements as in the proposed solution, i.e., the $C$-phase training matrices are completely deterministic and known.
Hence, all other columns $\mathbf{w}_m,m\neq n$ are \emph{not observable}. 
Since the power balance equation concerning the observable voltages $\boldsymbol{\omega}_n=\mathbf{0}_T$ also includes and depends on $\mathbf{v}_m,m\neq n$ (as a result of the fact that the buses are connected through $\mathbf{Y}$) and if classical, non-Bayesian framework is employed (without exploiting any prior knowledge), $\mathbf{v}_m,m\neq n$ should be treated as \emph{unknown} parameters in the same way as the generation capacities, load demands and line conductances. 
Therefore, the parameter vector $\boldsymbol{\theta}$ should be redefined as:
\begin{equation}
\boldsymbol{\theta} = [\mathbf{g}^{\mathsf{T}},\mathbf{d}^{\mathsf{T}},\boldsymbol{\psi}^{\mathsf{T}},\mathbf{v}_m^{\mathsf{T}}]_{m\neq n}^{\mathsf{T}},
\end{equation}
with
\begin{equation}
\mathsf{dim}(\boldsymbol{\theta}) = \frac{1}{2}N(N+7) + (N-1)T.
\end{equation}
The sufficient excitation conditions in this case should be restated in term of $\mathbf{\omega}_n$ since only $\mathbf{v}_n$ is observable; we get:
\begin{align}
\mathsf{rank}(\mathbf{\Upsilon}_{-n}) & = \mathsf{dim}(\boldsymbol{\theta}_{-n}),\\
\mathsf{rank}(\mathbf{\Gamma}) & = NT,
\end{align}
where $\mathbf{\Upsilon}_{-n}$ and $\mathbf{\Gamma}$ are the Jacobians of $\mathbf{\omega}_n$ w.r.t. $\boldsymbol{\theta}_{-n}$ and $\mathsf{vec}(\mathbf{V})$, respectively.
It becomes immediately evident that $\mathbf{\Upsilon}_{-n}$ is a fat matrix, i.e., $\mathsf{dim}(\mathbf{\Upsilon}_{-n}) = T\times\mathsf{dim}(\boldsymbol{\theta}_{-n})$ with column rank at most $T<\mathsf{dim}(\boldsymbol{\theta}_{-n})$; hence, the first sufficient excitation condition is not satisfied and $\boldsymbol{\theta}_{-n}$ cannot be uniquely identified.

Equivalently, one can look at the same problem from the perspective of the constrained ML optimization.
Namely, the joint parameter/state vector now is:
\begin{equation}
\boldsymbol{\vartheta} = \begin{bmatrix}
\boldsymbol{\theta}\\\mathbf{v}_n
\end{bmatrix}.
\end{equation}
The constrained ML optimization problem should be formulated over $\mathbf{\omega}_n$ since only $\mathbf{v}_n$ is observable: 
\begin{align}
\hat{\boldsymbol{\vartheta}}_{-n} & = \min_{{\boldsymbol{\vartheta}}_{-n}}\left\{-\ln\rho(\mathbf{w}_n;\boldsymbol{\theta})\right\}\\\nonumber
\text{s.t. }&\mathbf{\omega}_n = \mathbf{0}_{T}.
\end{align}
Clearly, the number of linearly independent equality constraints is at most $T<\mathsf{dim}(\boldsymbol{\theta}_{-n})$, yielding an ill-conditioned optimization problem that does not converge to any meaningful solution. 

\section{Proof of Proposition II}\label{app:propII}
Controller $n$ derives the channel estimator $\hat{\mathbf{h}}_n$ using the measurement vector from sub-phase $\alpha$, i.e., $\mathbf{w}_n^{\alpha}$.
Replacing $\sqrt{\pi_n(t)}=\sqrt{\pi^{\alpha}},~n\in\mathcal{N},~t\in\mathcal{T}^{\alpha}$ in the linear model
\begin{equation}\label{eq:linear_model}
\mathbf{w}_n^{\alpha/\beta} \approx \tilde{v}_n\mathbf{1}_{T^{\alpha/\beta}} + (\mathbf{\Pi}^{\alpha/\beta}\odot\Delta\mathbf{X}^{\alpha/\beta})\mathbf{h}_n + \mathbf{z}_n^{\alpha/\beta},
\end{equation}
we get:
\begin{align}
\mathbf{w}_n^{\alpha} \approx \tilde{v}_n\mathbf{1}_{T^{\alpha}} + \sqrt{\pi^{\alpha}}\Delta\mathbf{X}^{\alpha}\mathbf{h}_n + \mathbf{z}_n^{\alpha}.
\end{align}
Using the above, $\hat{\mathbf{h}}_n$ is obtained by solving the linear least squares problem:
\begin{align}\label{TSP2017_eq:LLSE_ch_est}
\hat{\mathbf{h}}_n & = \min_{\mathbf{h}_n}\|\mathbf{w}_n^{\alpha} - \tilde{v}_n\mathbf{1}_{T^{\alpha}} - \sqrt{\pi^{\alpha}}\Delta\mathbf{X}^{\alpha}\mathbf{h}_n\|_2^2\\
& = \frac{1}{\sqrt{\pi^{\alpha}}\delta^{\alpha}}(\Delta\mathbf{X}^{\alpha})^{\mathsf{T}}\mathbf{w}_n^{\alpha}.
\end{align}
Using $\sqrt{\pi_n(t)}=\sqrt{\pi^{\beta}}(\overline{w}_n(b) - \chi_n),~n\in\mathcal{N},~t\in\mathcal{T}^{\beta;b},~b\in\overline{\mathcal{T}}$, \eqref{eq:linear_model} can be rewritten as:
\begin{align}
\mathbf{w}_n^{\beta;b} & = \tilde{v}_n\mathbf{1}_{L} + \sqrt{\pi^{\beta}}\Delta\mathbf{X}^{\beta;b}\mathsf{D}(\overline{\mathbf{w}}^b - \boldsymbol{\chi})\mathbf{h}_n + \mathbf{z}_n^{\beta;b}\\
& = \tilde{v}_n\mathbf{1}_{L} + \sqrt{\pi^{\beta}}\Delta\mathbf{X}^{\beta;b}\mathsf{D}(\mathbf{h}_n)(\overline{\mathbf{w}}^b - \boldsymbol{\chi}) + \mathbf{z}_n^{\beta;b},~b\in\overline{\mathcal{T}},
\end{align}
where we used the commutative property of the product $\mathsf{D}(\overline{\mathbf{w}}^b - \boldsymbol{\chi})\mathbf{h}_n$.
Note that $\overline{\mathbf{w}}^b$ is the $b$-th row of the $M$-phase measurement matrix $\overline{\mathbf{W}}$ and contains the data transmitted by the controllers in block $b$.
Using the the channel estimate, controller $n$ obtains a local copy of $\overline{\mathbf{w}}^b$, denoted with $\overline{\mathbf{w}}_{(n)}^b$ by solving the following linear least squares problem:
\begin{align}
\overline{\mathbf{w}}_{(n)}^b & =\min_{\overline{\mathbf{w}}^b}\| \mathbf{w}_n^{\beta;b} - \tilde{v}_n\mathbf{1}_{L} - \sqrt{\pi^{\beta}}\Delta\mathbf{X}^{\beta;b}\mathsf{D}(\hat{\mathbf{h}}_n)(\overline{\mathbf{w}}^b - \boldsymbol{\chi})\|_2^2\\
& = \frac{1}{\sqrt{\pi^{\beta}}\delta^{\beta}}\mathsf{D}^{-1}(\hat{\mathbf{h}}_n)(\Delta\mathbf{X}^{\beta;b})^{\mathsf{T}}\mathbf{w}_n^{\beta;b} + \boldsymbol{\chi}.
\end{align}
Note that $\overline{\mathbf{W}}_{(n)} = \sum_{b\in\overline{\mathcal{T}}}\mathbf{e}_b(\overline{\mathbf{w}}_{(n)}^b)^{\mathsf{T}}$; so we get:
\begin{align}
\overline{\mathbf{W}}_{(n)} = \frac{1}{\sqrt{\pi^{\beta}}\delta^{\beta}}\left(\sum_{b\in\overline{\mathcal{T}}}\mathbf{e}_b({\mathbf{w}}_{n}^{\beta;b})^{\mathsf{T}}\Delta\mathbf{X}^{\beta;b}\right)\oslash(\mathbf{1}_{\overline{T}}\hat{\mathbf{h}}_n^{\mathsf{T}}) - \mathbf{1}_{\overline{T}}\boldsymbol{\chi}^{\mathsf{T}}
\end{align}.
Vectorizing the above, we obtain:
\begin{align}
\mathsf{vec}(\overline{\mathbf{W}}_{(n)}) & = \frac{1}{\sqrt{\pi^{\beta}}\delta^{\beta}}\left(\sum_{b\in\overline{\mathcal{T}}}\mathsf{vec}(\mathbf{e}_b({\mathbf{w}}_{n}^{\beta;b})^{\mathsf{T}}\Delta\mathbf{X}^{\beta;b})\right)\oslash\mathsf{vec}(\mathbf{1}_{\overline{T}}\hat{\mathbf{h}}_n^{\mathsf{T}}) + \mathsf{vec}(\mathbf{1}_{\overline{T}}\boldsymbol{\chi}^{\mathsf{T}})\\
& = \frac{\sqrt{\pi^{\alpha}}\delta^{\alpha}}{\sqrt{\pi^{\beta}}\delta^{\beta}}\left(\sum_{b\in\overline{\mathcal{T}}}\boldsymbol{\mathcal{X}}^{\beta;b}{\mathbf{w}}_{n}^{\beta;b}\right)\oslash(\boldsymbol{\mathcal{X}}^{\alpha}\mathbf{w}_n^{\alpha}) + \boldsymbol{\mathcal{I}}\boldsymbol{\chi}\\
& = \frac{\sqrt{\pi^{\alpha}}\delta^{\alpha}}{\sqrt{\pi^{\beta}}\delta^{\beta}}\mathsf{D}^{-1}(\boldsymbol{\mathcal{X}}^{\alpha}\mathbf{w}_n^{\alpha})\sum_{b\in\overline{\mathcal{T}}}\boldsymbol{\mathcal{X}}^{\beta;b}{\mathbf{w}}_{n}^{\beta;b} + \boldsymbol{\mathcal{I}}\boldsymbol{\chi},
\end{align}
which completes the derivation.

\section{{Derivation of the Gaussian approximation of} $\rho(\mathsf{vec}(\overline{\mathbf{W}});\boldsymbol{\theta})$}\label{app:Gauss_app}
Let $\mathbf{w}_n^{\alpha}=\mathbf{v}_n^{\alpha} + \Delta\mathbf{w}_n^{\alpha}$ where $\Delta\mathbf{w}_n^{\alpha}\sim\mathsf{N}(\mathbf{0}_{T^{\alpha}},\sigma^2\mathbf{I}_{T^{\alpha}})$.
Similarly, $\mathbf{w}_n^{\beta;b}=\mathbf{v}_n^{\beta;b} + \Delta\mathbf{w}_n^{\beta;b}$ where $\Delta\mathbf{w}_n^{\beta;b}\sim\mathsf{N}(\mathbf{0}_{L},\sigma^2\mathbf{I}_{L})$ for $b\in\overline{\mathcal{T}}$.
Then, \eqref{eq:Wk_full_est} can be written as:
\begin{align}
\mathsf{vec}(\overline{\mathbf{W}}_{(n)}) & = \frac{\sqrt{\pi^{\alpha}}\delta^{\alpha}}{\sqrt{\pi^{\beta}}\delta^{\beta}}\mathsf{D}^{-1}(\boldsymbol{\mathcal{X}}^{\alpha}\mathbf{w}_n^{\alpha})\sum_{b\in\overline{\mathcal{T}}}\boldsymbol{\mathcal{X}}^{\beta;b}{\mathbf{w}}_{n}^{\beta;b} + \boldsymbol{\mathcal{I}}\boldsymbol{\chi}\\
& = \frac{\sqrt{\pi^{\alpha}}\delta^{\alpha}}{\sqrt{\pi^{\beta}}\delta^{\beta}}\mathsf{D}^{-1}(\boldsymbol{\mathcal{X}}^{\alpha}(\mathbf{v}_n^{\alpha} + \Delta\mathbf{w}_n^{\alpha}))\sum_{b\in\overline{\mathcal{T}}}\boldsymbol{\mathcal{X}}^{\beta;b}(\mathbf{v}_n^{\beta;b} + \Delta\mathbf{w}_n^{\beta;b}) + \boldsymbol{\mathcal{I}}\boldsymbol{\chi}\\
& = \frac{\sqrt{\pi^{\alpha}}\delta^{\alpha}}{\sqrt{\pi^{\beta}}\delta^{\beta}}\mathsf{D}^{-1}(\boldsymbol{\mathcal{X}}^{\alpha}\mathbf{v}_n^{\alpha})(\mathbf{I}_{\overline{T}N} + \mathsf{D}^{-1}(\boldsymbol{\mathcal{X}}^{\alpha}\mathbf{v}_n^{\alpha})\mathsf{D}(\boldsymbol{\mathcal{X}}^{\alpha}\Delta\mathbf{w}_n^{\alpha}))^{-1}\sum_{b\in\overline{\mathcal{T}}}\boldsymbol{\mathcal{X}}^{\beta;b}(\mathbf{v}_n^{\beta;b} + \Delta\mathbf{w}_n^{\beta;b}) + \boldsymbol{\mathcal{I}}\boldsymbol{\chi}\\
& \overset{(a)}{\approx}\frac{\sqrt{\pi^{\alpha}}\delta^{\alpha}}{\sqrt{\pi^{\beta}}\delta^{\beta}}\mathsf{D}^{-1}(\boldsymbol{\mathcal{X}}^{\alpha}\mathbf{v}_n^{\alpha})(\mathbf{I}_{\overline{T}N} - \mathsf{D}^{-1}(\boldsymbol{\mathcal{X}}^{\alpha}\mathbf{v}_n^{\alpha})\mathsf{D}(\boldsymbol{\mathcal{X}}^{\alpha}\Delta\mathbf{w}_n^{\alpha}))\sum_{b\in\overline{\mathcal{T}}}\boldsymbol{\mathcal{X}}^{\beta;b}(\mathbf{v}_n^{\beta;b} + \Delta\mathbf{w}_n^{\beta;b}) + \boldsymbol{\mathcal{I}}\boldsymbol{\chi}\\\nonumber
& \approx \underbrace{\frac{\sqrt{\pi^{\alpha}}\delta^{\alpha}}{\sqrt{\pi^{\beta}}\delta^{\beta}}\mathsf{D}^{-1}(\boldsymbol{\mathcal{X}}^{\alpha}\mathbf{v}_n^{\alpha})\sum_{b\in\overline{\mathcal{T}}}\boldsymbol{\mathcal{X}}^{\beta;b}\mathbf{v}_n^{\beta;b} + \boldsymbol{\mathcal{I}}\boldsymbol{\chi}}_{\mathsf{vec}(\overline{\mathbf{W}})}\\\label{eq:first_order_approximation}
& + \frac{\sqrt{\pi^{\alpha}}\delta^{\alpha}}{\sqrt{\pi^{\beta}}\delta^{\beta}}\mathsf{D}^{-1}(\boldsymbol{\mathcal{X}}^{\alpha}\mathbf{v}_n^{\alpha})\sum_{b\in\overline{\mathcal{T}}}\boldsymbol{\mathcal{X}}^{\beta;b}\Delta\mathbf{w}_n^{\beta;b} + \frac{\sqrt{\pi^{\alpha}}\delta^{\alpha}}{\sqrt{\pi^{\beta}}\delta^{\beta}}\mathsf{D}^{-2}(\boldsymbol{\mathcal{X}}^{\alpha}\mathbf{v}_n^{\alpha})\sum_{b\in\overline{\mathcal{T}}}\mathsf{D}(\boldsymbol{\mathcal{X}}^{\beta;b}\mathbf{v}_n^{\beta;b})(\boldsymbol{\mathcal{X}}^{\alpha}\Delta\mathbf{w}_n^{\alpha}),
\end{align}
where $(a)$ follows from the Neumann expansion valid on the subset $\mathbf{0}_{T^{\alpha}}<\mathbf{w}_n^{\alpha}<2\mathbf{v}_n^{\alpha}$.
From \eqref{eq:first_order_approximation}, we see that $\mathsf{vec}(\overline{\mathbf{W}}_{(n)})$ can be approximated with Gaussian random vector with mean:
\begin{equation}
\mathbb{E}\left\{\mathsf{vec}(\overline{\mathbf{W}}_{(n)})\right\} \approx \mathbb{E}\left\{\mathsf{vec}(\overline{\mathbf{W}})\right\} = \mathsf{vec}(\overline{\mathbf{V}}),
\end{equation}
and covariance matrix:
\begin{align}
\mathbf{\Sigma} & \approx \mathsf{cov}\left\{\mathsf{vec}(\overline{\mathbf{W}})\right\} + \mathsf{cov}\left\{\frac{\sqrt{\pi^{\alpha}}\delta^{\alpha}}{\sqrt{\pi^{\beta}}\delta^{\beta}}\mathsf{D}^{-1}(\boldsymbol{\mathcal{X}}^{\alpha}\mathbf{v}_n^{\alpha})\sum_{b\in\overline{\mathcal{T}}}\boldsymbol{\mathcal{X}}^{\beta;b}\Delta\mathbf{w}_n^{\beta;b}\right\} + \mathsf{cov}\left\{\frac{\sqrt{\pi^{\alpha}}\delta^{\alpha}}{\sqrt{\pi^{\beta}}\delta^{\beta}}\mathsf{D}^{-2}(\boldsymbol{\mathcal{X}}^{\alpha}\mathbf{v}_n^{\alpha})\sum_{b\in\overline{\mathcal{T}}}\mathsf{D}(\boldsymbol{\mathcal{X}}^{\beta;b}\mathbf{v}_n^{\beta;b})(\boldsymbol{\mathcal{X}}^{\alpha}\Delta\mathbf{w}_n^{\alpha})\right\}\\\nonumber
& = \sigma^2\mathbf{I}_{\overline{T}N} + \frac{\pi^{\alpha}(\delta^{\alpha})^2}{\pi^{\beta}(\delta^{\beta})^2}\mathsf{D}^{-1}(\boldsymbol{\mathcal{X}}^{\alpha}\mathbf{v}_n^{\alpha})\left(\sum_{b\in\overline{\mathcal{T}}}\boldsymbol{\mathcal{X}}^{\beta;b}\mathsf{cov}\left\{\Delta\mathbf{w}_n^{\beta;b}\right\}(\boldsymbol{\mathcal{X}}^{\beta;b})^{\mathsf{T}}\right)\mathsf{D}^{-1}(\boldsymbol{\mathcal{X}}^{\alpha}\mathbf{v}_n^{\alpha})\\
& + \frac{\pi^{\alpha}(\delta^{\alpha})^2}{\pi^{\beta}(\delta^{\beta})^2}\mathsf{D}^{-2}(\boldsymbol{\mathcal{X}}^{\alpha}\mathbf{v}_n^{\alpha})\left(\sum_{b\in\overline{\mathcal{T}}}\mathsf{D}(\boldsymbol{\mathcal{X}}^{\beta;b}\mathbf{v}_n^{\beta;b})\boldsymbol{\mathcal{X}}^{\alpha}\mathsf{cov}\left\{\Delta\mathbf{w}_n^{\alpha}\right\}(\boldsymbol{\mathcal{X}}^{\alpha})^{\mathsf{T}}\mathsf{D}(\boldsymbol{\mathcal{X}}^{\beta;b}\mathbf{v}_n^{\beta;b})\right)\mathsf{D}^{-2}(\boldsymbol{\mathcal{X}}^{\alpha}\mathbf{v}_n^{\alpha})\\\nonumber
& = \sigma^2\mathbf{I}_{\overline{T}N} + \sigma^2\frac{\pi^{\alpha}(\delta^{\alpha})^2}{\pi^{\beta}\delta^{\beta}}\mathsf{D}^{-1}(\boldsymbol{\mathcal{X}}^{\alpha}\mathbf{v}_n^{\alpha})\left(\sum_{b\in\overline{\mathcal{T}}}\mathbf{I}_N\otimes(\mathbf{e}_b\mathbf{e}_b^{\mathsf{T}})\right)\mathsf{D}^{-1}(\boldsymbol{\mathcal{X}}^{\alpha}\mathbf{v}_n^{\alpha})\\
& + \sigma^2\frac{\pi^{\alpha}(\delta^{\alpha})^3}{\pi^{\beta}(\delta^{\beta})^2}\mathsf{D}^{-2}(\boldsymbol{\mathcal{X}}^{\alpha}\mathbf{v}_n^{\alpha})\left(\sum_{b\in\overline{\mathcal{T}}}\mathsf{D}(\boldsymbol{\mathcal{X}}^{\beta;b}\mathbf{v}_n^{\beta;b})(\mathbf{I}_N\otimes\mathbf{1}_{\overline{T}\times\overline{T}})\mathsf{D}(\boldsymbol{\mathcal{X}}^{\beta;b}\mathbf{v}_n^{\beta;b})\right)\mathsf{D}^{-2}(\boldsymbol{\mathcal{X}}^{\alpha}\mathbf{v}_n^{\alpha})\\
& = \sigma^2\left(\mathbf{I}_{\overline{T}N} + \frac{\pi^{\alpha}(\delta^{\alpha})^2}{\pi^{\beta}\delta^{\beta}}\mathsf{D}^{-2}(\boldsymbol{\mathcal{X}}^{\alpha}\mathbf{v}_n^{\alpha}) + \frac{\pi^{\alpha}(\delta^{\alpha})^3}{\pi^{\beta}(\delta^{\beta})^2}\mathsf{D}^{-2}(\boldsymbol{\mathcal{X}}^{\alpha}\mathbf{v}_n^{\alpha})\left(\sum_{b\in\overline{\mathcal{T}}}\mathsf{D}(\boldsymbol{\mathcal{X}}^{\beta;b}\mathbf{v}_n^{\beta;b})(\mathbf{I}_N\otimes\mathbf{1}_{\overline{T}\times\overline{T}})\mathsf{D}(\boldsymbol{\mathcal{X}}^{\beta;b}\mathbf{v}_n^{\beta;b})\right)\mathsf{D}^{-2}(\boldsymbol{\mathcal{X}}^{\alpha}\mathbf{v}_n^{\alpha})\right).
\end{align}

\section{Proof of Proposition III}\label{app:propIII}
The Lagrange method of multipliers casts the original constrained ML problem into an unconstrained as follows:
\begin{equation}\label{TSP2017_eq:MLE_Lagrange}
 \hat{\boldsymbol{\vartheta}}_{-n} = \min_{\hat{\boldsymbol{\vartheta}}_{-n},\boldsymbol{\lambda}}\left\{-\frac{1}{2}\left\|\mathbf{\Sigma}^{-\frac{1}{2}}(\mathsf{vec}(\overline{\mathbf{W}}_{(n)}) - \mathsf{vec}(\overline{\mathbf{V}}))\right\|_2^2 + \boldsymbol{\lambda}^{\mathsf{T}}\mathsf{vec}(\overline{\mathbf{\Omega}})\right\},
\end{equation}
where we used the Gaussian approximation for the pdf $\rho(\mathsf{vec}(\overline{\mathbf{W}}_{(n)});\boldsymbol{\theta})$.
$\boldsymbol{\lambda}$ is $\overline{T}N\times 1$ vector of multipliers.
Applying the KKT conditions to \eqref{TSP2017_eq:MLE_Lagrange} after replacing the power balance constraint with its first order approximation, we get the following system of equations:
\begin{align}\label{TSP2017_eq:KKT_normal_eq_linear}
	   & \mathbf{\Sigma}^{-1}(\mathsf{vec}(\overline{\mathbf{W}}_{(n)}) - \mathsf{vec}(\overline{\mathbf{V}})) - (\mathbf{\Gamma}^{(j)})^{\mathsf{T}}\boldsymbol{\lambda}  = \mathbf{0}_{\overline{T}N}, \\\label{TSP2017_eq:KKT_normal_eq2_linear}
     & (\mathbf{\Upsilon}_{-n}^{(j)})^{\mathsf{T}}\boldsymbol{\lambda}  = \mathbf{0}_{\mathsf{dim}(\boldsymbol{\theta}_{-n})}, \\\label{TSP2017_eq:KKT_normal_eq3_linear}
		 & {\mathbf{\Upsilon}}^{(j)}\boldsymbol{\theta} + {\mathbf{\Gamma}}^{(j)}(\mathsf{vec}(\overline{\mathbf{V}})-\mathsf{vec}(\overline{\mathbf{V}}^{(j)}))  = \mathbf{0}_{\overline{T}N},
\end{align}
The above system is linear in $\boldsymbol{\vartheta}_{-n}$ and can be solved efficiently; the derivation of the solution follows similar steps as in \eqref{31}.
Multiplying \eqref{TSP2017_eq:KKT_normal_eq_linear} with $\mathbf{\Gamma}^{(j)}\mathbf{\Sigma}$ yields:
\begin{align}
\mathbf{\Gamma}^{(j)}\mathsf{vec}(\overline{\mathbf{V}}) = \mathbf{\Gamma}^{(j)}\mathsf{vec}(\overline{\mathbf{W}}_{(n)}) - \mathbf{\Gamma}^{(j)}\mathbf{\Sigma}(\mathbf{\Gamma}^{(j)})^{\mathsf{T}}\boldsymbol{\lambda},
\end{align}
which is substituted in \eqref{TSP2017_eq:KKT_normal_eq3_linear} to yield:
\begin{align}
\mathbf{\Upsilon}^{(j)}\boldsymbol{\theta} + \mathbf{\Gamma}^{(j)}(\mathsf{vec}(\overline{\mathbf{W}}_{(n)})-\mathsf{vec}(\overline{\mathbf{V}}^{(j)})) - \mathbf{\Gamma}^{(j)}\mathbf{\Sigma}(\mathbf{\Gamma}^{(j)})^{\mathsf{T}}\boldsymbol{\lambda} = \mathbf{0}_{\overline{T}N}.
\end{align}
Solving for $\boldsymbol{\lambda}$ gives:
\begin{align}\label{TSP2017_eq:sol_KKT_0}
\boldsymbol{\lambda} & = (\mathbf{\Gamma}^{(j)}\mathbf{\Sigma}(\mathbf{\Gamma}^{(j)})^{\mathsf{T}})^{-1}(\mathbf{\Upsilon}^{(j)}\boldsymbol{\theta} + \mathbf{\Gamma}^{(j)}(\mathsf{vec}(\overline{\mathbf{W}}_{(n)})-\mathsf{vec}(\overline{\mathbf{V}}^{(j)}))).
\end{align}
Multiplying \eqref{TSP2017_eq:sol_KKT_0} with $(\mathbf{\Upsilon}_{-n}^{(j)})^{\mathsf{T}}$ on both sides, gives:
\begin{align}
(\mathbf{\Upsilon}_{-k}^{(j)})\boldsymbol{\lambda} = (\mathbf{\Upsilon}_{-k}^{(j)})^{\mathsf{T}}(\mathbf{\Gamma}^{(j)}\mathbf{\Sigma}(\mathbf{\Gamma}^{(j)})^{\mathsf{T}})^{-1}(\mathbf{\Upsilon}^{(j)}\boldsymbol{\theta} + \mathbf{\Gamma}^{(j)}(\mathsf{vec}(\overline{\mathbf{W}}_{(n)})-\mathsf{vec}(\overline{\mathbf{V}}^{(j)}))),
\end{align}
which, after replacing $\mathbf{\Upsilon}^{(j)}\boldsymbol{\theta} = \mathbf{\Upsilon}_{-n}^{(j)}\boldsymbol{\theta}_{-n} + \boldsymbol{\upsilon}_n^{(j)}g_n$ and solving for $\boldsymbol{\theta}_{-n}$ gives \eqref{eq:sol_KKT_1}.
Finally, replacing \eqref{TSP2017_eq:sol_KKT_0} in \eqref{TSP2017_eq:KKT_normal_eq_linear} and solving for $\mathsf{vec}(\overline{\mathbf{V}})$ produces \eqref{eq:sol_KKT_2}, completing the proof.

\section{Proof of Proposition IV}\label{app:propIV}
Recall that the implicit function theorem governs the existence of an explicit solution of the system of power balance equations $\omega_n = 0,~n\in\mathcal{N}$ of the following form:
\begin{equation}
v_n = f_n(\boldsymbol{\theta}),~n\in\mathcal{N}.
\end{equation}
Hence, again by the implicit function theorem, the solution of the $M$-phase power balance equation $\overline{\mathbf{\Omega}}=\mathbf{0}_{\overline{T}\times N}$ exists and can be written in the following form:
\begin{equation}\label{eq:V_theta_explicit}
\overline{\mathbf{V}} = \overline{\mathbf{F}}(\boldsymbol{\theta}),
\end{equation}
where the $\overline{T}\times N$ matrix $\overline{\mathbf{F}}$ is defined as $[\overline{\mathbf{F}}]_{b,n}=f_n(b)~n\in\mathcal{N},~b\in\overline{\mathcal{T}}$.
If $\overline{\mathbf{F}}$ is available in closed form, the $M$-phase measurement matrix (i.e., its vectorization) can be written explicitly in terms of $\boldsymbol{\theta}$ as:
\begin{equation}
\mathsf{vec}(\overline{\mathbf{W}}) = \mathsf{vec}(\overline{\mathbf{F}}(\boldsymbol{\theta})) + \mathsf{vec}(\overline{\mathbf{W}}).
\end{equation}
Using the above, we derive the CRLB.
In particular, the MSE matrix of $\hat{\boldsymbol{\theta}}_{-n}$ can be bounded from below as:
\begin{equation}
\text{MSE}(\hat{\boldsymbol{\theta}}_{-n}) \succeq \boldsymbol{\mathcal{J}}^{-1}(\boldsymbol{\theta}_{-n}),
\end{equation}
where $\boldsymbol{\mathcal{J}}(\boldsymbol{\theta}_{-n})$ is the Fisher Information Matrix (FIM) defined as:
\begin{align}
\boldsymbol{\mathcal{J}}(\boldsymbol{\theta}_{-n}) & = \mathbb{E}\left\{\nabla_{\boldsymbol{\theta}_{-n}}^{\mathsf{T}}\ln\rho(\mathsf{vec}(\overline{\mathbf{W}});\boldsymbol{\theta})\nabla_{\boldsymbol{\theta}_{-n}}\ln\rho(\mathsf{vec}(\overline{\mathbf{W}});\boldsymbol{\theta})\right\}.
\end{align}
Using the Gaussian approximation for the pdf of $\mathsf{vec}(\overline{\mathbf{W}})$, the FIM can be approximated with the following Grammian:
\begin{align}\label{eq:FIM_GAuss}
\boldsymbol{\mathcal{J}}(\boldsymbol{\theta}_{-n}) & \approx \nabla_{\boldsymbol{\theta}_{-n}}^{\mathsf{T}}\mathsf{vec}(\overline{\mathbf{F}}(\boldsymbol{\theta}))\mathbf{\Sigma}^{-1}\nabla_{\boldsymbol{\theta}_{-n}}\mathsf{vec}(\overline{\mathbf{F}}(\boldsymbol{\theta})).
\end{align}
Applying the implicit function theorem, we obtain the following expression for the Jacobian $\nabla_{\boldsymbol{\theta}_{-n}}\mathsf{vec}(\overline{\mathbf{F}}(\boldsymbol{\theta}))$:
\begin{align}
\nabla_{\boldsymbol{\theta}_{-n}}\mathsf{vec}(\overline{\mathbf{F}}(\boldsymbol{\theta})) & = - \nabla_{\mathsf{vec}(\overline{\mathbf{V})}}^{-1}\mathsf{vec}(\overline{\mathbf{\Omega}})\nabla_{\boldsymbol{\theta}_{-n}}\mathsf{vec}(\overline{\mathbf{\Omega}})\\
& = - \mathbf{\Gamma}^{-1}\mathbf{\Upsilon}_{-n}.
\end{align}
Substituting the above in \eqref{eq:FIM_GAuss} gives expression \eqref{eq:CRLB_theta}.
To bound the MSE matrix of $\mathsf{vec}(\hat{\overline{\mathbf{V}}})$, we use \eqref{eq:V_theta_explicit}, i.e., the fact that $\mathsf{vec}(\overline{\mathbf{V}})$ is a transformed version of $\boldsymbol{\theta}$ and apply the corresponding CRLB formula, i.e.:
\begin{align}
\text{MSE}(\mathsf{vec}(\hat{\overline{\mathbf{V}}})) & \succeq \nabla_{\boldsymbol{\theta}_{-n}}\mathsf{vec}(\overline{\mathbf{F}}(\boldsymbol{\theta}))\boldsymbol{\mathcal{J}}^{-1}(\boldsymbol{\theta}_{-n})\nabla_{\boldsymbol{\theta}_{-n}}^{\mathsf{T}}\mathsf{vec}(\overline{\mathbf{F}}(\boldsymbol{\theta}))\\
& = \mathbf{\Gamma}^{-1}\mathbf{\Upsilon}_{-n}\boldsymbol{\mathcal{J}}^{-1}(\boldsymbol{\theta}_{-n})\mathbf{\Upsilon}_{-n}^{\mathsf{T}}(\mathbf{\Gamma}^{-1})^{\mathsf{T}},
\end{align}
completing the proof.

\bibliographystyle{IEEEtran}

\begin{thebibliography}{}
\providecommand{\url}[1]{#1}
\csname url@samestyle\endcsname
\providecommand{\newblock}{\relax}
\providecommand{\bibinfo}[2]{#2}
\providecommand{\BIBentrySTDinterwordspacing}{\spaceskip=0pt\relax}
\providecommand{\BIBentryALTinterwordstretchfactor}{4}
\providecommand{\BIBentryALTinterwordspacing}{\spaceskip=\fontdimen2\font plus
\BIBentryALTinterwordstretchfactor\fontdimen3\font minus
  \fontdimen4\font\relax}
\providecommand{\BIBforeignlanguage}[2]{{%
\expandafter\ifx\csname l@#1\endcsname\relax
\typeout{** WARNING: IEEEtran.bst: No hyphenation pattern has been}%
\typeout{** loaded for the language `#1'. Using the pattern for}%
\typeout{** the default language instead.}%
\else
\language=\csname l@#1\endcsname
\fi
#2}}
\providecommand{\BIBdecl}{\relax}
\BIBdecl

\end{thebibliography}


\begin{thebibliography}{10}
\providecommand{\url}[1]{#1}
\csname url@samestyle\endcsname
\providecommand{\newblock}{\relax}
\providecommand{\bibinfo}[2]{#2}
\providecommand{\BIBentrySTDinterwordspacing}{\spaceskip=0pt\relax}
\providecommand{\BIBentryALTinterwordstretchfactor}{4}
\providecommand{\BIBentryALTinterwordspacing}{\spaceskip=\fontdimen2\font plus
\BIBentryALTinterwordstretchfactor\fontdimen3\font minus
  \fontdimen4\font\relax}
\providecommand{\BIBforeignlanguage}[2]{{%
\expandafter\ifx\csname l@#1\endcsname\relax
\typeout{** WARNING: IEEEtran.bst: No hyphenation pattern has been}%
\typeout{** loaded for the language `#1'. Using the pattern for}%
\typeout{** the default language instead.}%
\else
\language=\csname l@#1\endcsname
\fi
#2}}
\providecommand{\BIBdecl}{\relax}
\BIBdecl

\bibitem{1}
R.~H. Lasseter and P.~Paigi, ``Microgrid: a conceptual solution,'' in
  \emph{2004 IEEE 35th Annual Power Electronics Specialists Conference (IEEE
  Cat. No.04CH37551)}, vol.~6, June 2004, pp. 4285--4290 Vol.6.

\bibitem{2}
L.~E. Zubieta, ``Are microgrids the future of energy?: Dc microgrids from
  concept to demonstration to deployment,'' \emph{IEEE Electrification
  Magazine}, vol.~4, no.~2, pp. 37--44, June 2016.

\bibitem{3}
T.~Dragicevic, X.~Lu, J.~C. Vasquez, and J.~M. Guerrero, ``Dc microgrids; part
  i: A review of control strategies and stabilization techniques,'' \emph{IEEE
  Transactions on Power Electronics}, vol.~31, no.~7, pp. 4876--4891, July
  2016.

\bibitem{ref:new_last_1}
T.~Dragičević, X.~Lu, J.~C. Vasquez, and J.~M. Guerrero, ``Dc microgrids;
  part ii: A review of power architectures, applications, and standardization
  issues,'' \emph{IEEE Transactions on Power Electronics}, vol.~31, no.~5, pp.
  3528--3549, May 2016.

\bibitem{2018_new1}
L.~Strenge, H.~Kirchhoff, G.~L. Ndow, and F.~Hellmann, ``Stability of meshed dc
  microgrids using probabilistic analysis,'' in \emph{2017 IEEE Second
  International Conference on DC Microgrids (ICDCM)}, June 2017, pp. 175--180.

\bibitem{2018_new2}
C.~Marnay, S.~Lanzisera, M.~Stadler, and J.~Lai, ``Building scale dc
  microgrids,'' in \emph{2012 IEEE Energytech}, May 2012, pp. 1--5.

\bibitem{2018_new5}
D.~Zhang, J.~Jiang, L.~Y. Wang, and W.~Zhang, ``Robust and scalable management
  of power networks in dual-source trolleybus systems: A consensus control
  framework,'' \emph{IEEE Transactions on Intelligent Transportation Systems},
  vol.~17, no.~4, pp. 1029--1038, April 2016.

\bibitem{2018_new6}
M.~A. Masrur, A.~G. Skowronska, J.~Hancock, S.~W. Kolhoff, D.~Z. McGrew, J.~C.
  Vandiver, and J.~Gatherer, ``Military-based vehicle-to-grid and
  vehicle-to-vehicle microgrid; system architecture and implementation,''
  \emph{IEEE Transactions on Transportation Electrification}, vol.~4, no.~1,
  pp. 157--171, March 2018.

\bibitem{2018_new7}
K.~Cavanagh, J.~A. Belk, and K.~Turitsyn, ``Transient stability guarantees for
  ad hoc dc microgrids,'' \emph{IEEE Control Systems Letters}, vol.~2, no.~1,
  pp. 139--144, Jan 2018.

\bibitem{2018_new8}
L.~Mackay, T.~Hailu, L.~Ramirez-Elizondo, and P.~Bauer, ``Decentralized current
  limiting in meshed dc distribution grids,'' in \emph{2015 IEEE First
  International Conference on DC Microgrids (ICDCM)}, June 2015, pp. 234--238.

\bibitem{2018_new9}
M.~Hamza, M.~Shehroz, S.~Fazal, M.~Nasir, and H.~A. Khan, ``Design and analysis
  of solar pv based low-power low-voltage dc microgrid architectures for rural
  electrification,'' in \emph{2017 IEEE Power Energy Society General Meeting},
  July 2017, pp. 1--5.

\bibitem{14}
J.~Schonberger, R.~Duke, and S.~D. Round, ``Dc-bus signaling: A distributed
  control strategy for a hybrid renewable nanogrid,'' \emph{IEEE Transactions
  on Industrial Electronics}, vol.~53, no.~5, pp. 1453--1460, Oct 2006.

\bibitem{15}
D.~Chen, L.~Xu, and L.~Yao, ``Dc voltage variation based autonomous control of
  dc microgrids,'' \emph{IEEE Transactions on Power Delivery}, vol.~28, no.~2,
  pp. 637--648, April 2013.

\bibitem{16}
T.~L. Vandoorn, B.~Renders, L.~Degroote, B.~Meersman, and L.~Vandevelde,
  ``Active load control in islanded microgrids based on the grid voltage,''
  \emph{IEEE Transactions on Smart Grid}, vol.~2, no.~1, pp. 139--151, March
  2011.

\bibitem{4}
C.~Jin, P.~Wang, J.~Xiao, Y.~Tang, and F.~H. Choo, ``Implementation of
  hierarchical control in dc microgrids,'' \emph{IEEE Transactions on
  Industrial Electronics}, vol.~61, no.~8, pp. 4032--4042, Aug 2014.

\bibitem{5}
S.~Moayedi and A.~Davoudi, ``Unifying distributed dynamic optimization and
  control of islanded dc microgrids,'' \emph{IEEE Transactions on Power
  Electronics}, vol.~32, no.~3, pp. 2329--2346, March 2017.

\bibitem{6}
H.~Liang, B.~J. Choi, A.~Abdrabou, W.~Zhuang, and X.~S. Shen, ``Decentralized
  economic dispatch in microgrids via heterogeneous wireless networks,''
  \emph{IEEE Journal on Selected Areas in Communications}, vol.~30, no.~6, pp.
  1061--1074, July 2012.

\bibitem{2018_new3}
P.~Lin, C.~Jin, J.~Xiao, X.~Li, D.~Shi, Y.~Tang, and P.~Wang, ``A distributed
  control architecture for global system economic operation in autonomous
  hybrid ac/dc microgrids,'' \emph{IEEE Transactions on Smart Grid}, vol.~PP,
  no.~99, pp. 1--1, 2018.

\bibitem{2018_new4}
J.~Li, F.~Liu, Z.~Wang, S.~Low, and S.~Mei, ``Optimal power flow in stand-alone
  dc microgrids,'' \emph{IEEE Transactions on Power Systems}, vol.~PP, no.~99,
  pp. 1--1, 2018.

\bibitem{2018_new11}
G.~Zizzo, E.~R. Sanseverino, M.~G. Ippolito, M.~L.~D. Silvestre, and P.~Gallo,
  ``A technical approach to p2p energy transactions in microgrids,'' \emph{IEEE
  Transactions on Industrial Informatics}, vol.~PP, no.~99, pp. 1--1, 2018.

\bibitem{8}
G.~B. Giannakis, V.~Kekatos, N.~Gatsis, S.~J. Kim, H.~Zhu, and B.~F.
  Wollenberg, ``Monitoring and optimization for power grids: A signal
  processing perspective,'' \emph{IEEE Signal Processing Magazine}, vol.~30,
  no.~5, pp. 107--128, Sept 2013.

\bibitem{9}
P.~Chavali and A.~Nehorai, ``Distributed power system state estimation using
  factor graphs,'' \emph{IEEE Transactions on Signal Processing}, vol.~63,
  no.~11, pp. 2864--2876, June 2015.

\bibitem{10}
T.~Erseghe, S.~Tomasin, and A.~Vigato, ``Topology estimation for smart micro
  grids via powerline communications,'' \emph{IEEE Transactions on Signal
  Processing}, vol.~61, no.~13, pp. 3368--3377, July 2013.

\bibitem{12}
X.~Zhong, L.~Yu, R.~Brooks, and G.~K. Venayagamoorthy, ``Cyber security in
  smart dc microgrid operations,'' in \emph{2015 IEEE First International
  Conference on DC Microgrids (ICDCM)}, June 2015, pp. 86--91.

\bibitem{13}
O.~Beg, T.~Johnson, and A.~Davoudi, ``Detection of false-data injection attacks
  in cyber-physical dc microgrids,'' \emph{IEEE Transactions on Industrial
  Informatics}, vol.~PP, no.~99, pp. 1--1, 2017.

\bibitem{ref:new_last_3}
U.~Adhikari, T.~Morris, and S.~Pan, ``Wams cyber-physical test bed for power
  system, cybersecurity study, and data mining,'' \emph{IEEE Transactions on
  Smart Grid}, vol.~8, no.~6, pp. 2744--2753, Nov 2017.

\bibitem{ref:new_last_4}
Z.~Li, M.~Shahidehpour, and F.~Aminifar, ``Cybersecurity in distributed power
  systems,'' \emph{Proceedings of the IEEE}, vol. 105, no.~7, pp. 1367--1388,
  July 2017.

\bibitem{17}
S.~Galli, A.~Scaglione, and Z.~Wang, ``For the grid and through the grid: The
  role of power line communications in the smart grid,'' \emph{Proceedings of
  the IEEE}, vol.~99, no.~6, pp. 998--1027, June 2011.

\bibitem{25}
A.~Vosoughi and A.~Scaglione, ``Everything you always wanted to know about
  training: guidelines derived using the affine precoding framework and the
  crb,'' \emph{IEEE Transactions on Signal Processing}, vol.~54, no.~3, pp.
  940--954, March 2006.

\bibitem{27}
S.~Cobreces, E.~J. Bueno, D.~Pizarro, F.~J. Rodriguez, and F.~Huerta, ``Grid
  impedance monitoring system for distributed power generation electronic
  interfaces,'' \emph{IEEE Transactions on Instrumentation and Measurement},
  vol.~58, no.~9, pp. 3112--3121, Sept 2009.

\bibitem{18}
M.~Angjelichinoski, C.~Stefanovic, P.~Popovski, H.~Liu, P.~C. Loh, and
  F.~Blaabjerg, ``Power talk: How to modulate data over a dc micro grid bus
  using power electronics,'' in \emph{2015 IEEE Global Communications
  Conference (GLOBECOM)}, Dec 2015, pp. 1--7.

\bibitem{19}
M.~Angjelichinoski, C.~Stefanovic, P.~Popovski, H.~Liu, P.~C. Loh, and
  F.~Blaabjer, ``Multiuser communication through power talk in dc microgrids,''
  \emph{IEEE Journal on Selected Areas in Communications}, vol.~PP, no.~99, pp.
  1--1, 2016.

\bibitem{20}
M.~Angjelichinoski, C.~Stefanovic, and P.~Popovski, ``Power talk for multibus
  dc microgrids: Creating and optimizing communication channels,'' in
  \emph{2016 IEEE Global Communications Conference (GLOBECOM)}, Dec 2016, pp.
  1--7.

\bibitem{21}
M.~Angjelichinoski, A.~Scaglione, P.~Popovski, and C.~Stefanovic, ``Distrabuted
  estimation of the operating state of a single-bus dc microgrid without an
  external communication interface,'' in \emph{2016 IEEE Global Signal and
  Information Processing Conference (GlobalSIP)}, Dec 2016, pp. 1--4.

\bibitem{last}
M.~Angjelichinoski, {\v{C}}.~Stefanovi{\'{c}}, and P.~Popovski, \emph{Modemless
  Multiple Access Communications Over Powerlines for DC Microgrid
  Control}.\hskip 1em plus 0.5em minus 0.4em\relax Springer International
  Publishing, 2016, pp. 30--44.

\bibitem{28}
A.~B. et.al, ``Experimental determination of the zip coefficients for modern
  residential, commercial, and industrial loads,'' \emph{IEEE Transactions on
  Power Delivery}, vol.~29, no.~3, pp. 1372--1381, June 2014.

\bibitem{rev8}
T.~Dragičević, J.~M. Guerrero, J.~C. Vasquez, and D.~Škrlec, ``Supervisory
  control of an adaptive-droop regulated dc microgrid with battery management
  capability,'' \emph{IEEE Transactions on Power Electronics}, vol.~29, no.~2,
  pp. 695--706, Feb 2014.

\bibitem{29}
J.~W. Simpson-Porco, F.~Dörfler, and F.~Bullo, ``On resistive networks of
  constant-power devices,'' \emph{IEEE Transactions on Circuits and Systems II:
  Express Briefs}, vol.~62, no.~8, pp. 811--815, Aug 2015.

\bibitem{rev9}
X.~Li, H.~V. Poor, and A.~Scaglione, ``Blind topology identification for power
  systems,'' in \emph{2013 IEEE International Conference on Smart Grid
  Communications (SmartGridComm)}, Oct 2013, pp. 91--96.

\bibitem{ref:new_last_6}
S.~Bolognani, N.~Bof, D.~Michelotti, R.~Muraro, and L.~Schenato,
  ``Identification of power distribution network topology via voltage
  correlation analysis,'' in \emph{52nd IEEE Conference on Decision and
  Control}, Dec 2013, pp. 1659--1664.

\bibitem{rev11}
P.~Midya and P.~T. Krein, ``Noise properties of pulse-width modulated power
  converters: open-loop effects,'' \emph{IEEE Transactions on Power
  Electronics}, vol.~15, no.~6, pp. 1134--1143, Nov 2000.

\bibitem{rev4}
K.~Iwanicki, M.~van Steen, and S.~Voulgaris, \emph{Gossip-Based Clock
  Synchronization for Large Decentralized Systems}.\hskip 1em plus 0.5em minus
  0.4em\relax Berlin, Heidelberg: Springer Berlin Heidelberg, 2006, pp. 28--42.

\bibitem{32}
H.~L.~V. Trees, \emph{Detection, Estimation, and Modulation Theory: Radar-Sonar
  Signal Processing and Gaussian Signals in Noise}.\hskip 1em plus 0.5em minus
  0.4em\relax Melbourne, FL, USA: Krieger Publishing Co., Inc., 1992.

\bibitem{rev3}
Q.~Han, J.~Ding, E.~M. Airoldi, and V.~Tarokh, ``Slants: Sequential adaptive
  nonlinear modeling of time series,'' \emph{IEEE Transactions on Signal
  Processing}, vol.~65, no.~19, pp. 4994--5005, Oct 2017.

\bibitem{31}
\BIBentryALTinterwordspacing
H.~I. Britt and R.~H. Luecke, ``The estimation of parameters in nonlinear,
  implicit models,'' \emph{Technometrics}, vol.~15, no.~2, pp. 233--247, 1973.
  [Online]. Available: \url{http://www.jstor.org/stable/1266984}
\BIBentrySTDinterwordspacing

\bibitem{30}
S.~Boyd and L.~Vandenberghe, \emph{Convex Optimization}.\hskip 1em plus 0.5em
  minus 0.4em\relax New York, NY, USA: Cambridge University Press, 2004.

\bibitem{rev1}
P.~Stoica and B.~C. Ng, ``On the cramer-rao bound under parametric
  constraints,'' \emph{IEEE Signal Processing Letters}, vol.~5, no.~7, pp.
  177--179, July 1998.

\bibitem{33}
F.~Dorfler and F.~Bullo, ``Kron reduction of graphs with applications to
  electrical networks,'' \emph{IEEE Transactions on Circuits and Systems I:
  Regular Papers}, vol.~60, no.~1, pp. 150--163, Jan 2013.

\bibitem{rev10}
Z.~Ben-Haim and Y.~C. Eldar, ``The cramer-rao bound for estimating a sparse
  parameter vector,'' \emph{IEEE Transactions on Signal Processing}, vol.~58,
  no.~6, pp. 3384--3389, June 2010.

\end{thebibliography}

\end{document}